\documentclass[a4paper, 12pt]{article}

\usepackage[utf8]{inputenc}
\usepackage[T1]{fontenc}
\usepackage{lmodern}

\usepackage{graphicx}
\usepackage{color}

\usepackage{amsmath}
\usepackage{amsfonts}
\usepackage{amsthm}
\usepackage{amssymb}
\usepackage{fullpage}
\usepackage{enumitem}
\usepackage{tikz}
\usetikzlibrary{arrows}
\usetikzlibrary{calc}
\usetikzlibrary{shapes}
\usetikzlibrary{decorations.pathmorphing}
\usetikzlibrary{decorations.pathreplacing}
\usetikzlibrary{positioning}
\usetikzlibrary{backgrounds}
\usetikzlibrary{arrows.meta}
\usepackage{url}
\usepackage{stmaryrd}
\usepackage{xparse}
\usepackage{bbm}
\usepackage{mathtools}
\mathtoolsset{centercolon}
\usepackage{chngcntr}
\usepackage[colorlinks]{hyperref}
\usepackage{subfig}
\usepackage[font=small]{caption}

\usepackage{xcolor}
\usepackage{url}

\usepackage{etoolbox}

\usetikzlibrary{decorations.pathmorphing}
\usetikzlibrary{decorations.pathreplacing}

\usepackage[colorlinks]{hyperref}
\definecolor{lightblue}{rgb}{0.5,0.5,1.0}
\definecolor{darkred}{rgb}{0.5,0,0}
\definecolor{darkgreen}{rgb}{0,0.5,0}
\definecolor{darkblue}{rgb}{0,0,0.5}

\hypersetup{colorlinks,linkcolor=darkred,filecolor=darkgreen,urlcolor=darkred,citecolor=darkblue}

\DeclareMathAlphabet{\mymathbb}{U}{bbold}{m}{n}

\newtheorem{theorem}{Theorem}
\newtheorem{definition}[theorem]{Definition}
\newtheorem{lemma}[theorem]{Lemma}
\newtheorem{corollary}[theorem]{Corollary}
\newtheorem{remark}[theorem]{Remark}
\newtheorem{claim}{Claim}
\newenvironment{claimproof}[1][\proofname]{\begin{claimprooftemp}[#1]}{\end{claimprooftemp}}

\newcommand{\ifmultiroweq}[1]{#1}
\newcommand{\qspace}{\mathchoice{\;}{\,}{}{}}

\newcommand{\primeSub}[2]{#1\mathrlap{'}_{#2}}
\renewcommand{\phi}{\varphi}
\newcommand{\nat}{\mathbb{N}}
\newcommand{\ZZ}{\mathbb{Z}}
\newcommand{\FF}{\mathbb{F}}

\newcommand{\defining}[1]{\textbf{#1}}

\newcommand{\iso}{\cong}

\newcommand{\disunion}{\mathbin{\uplus}}

\DeclarePairedDelimiter\set{\lbrace}{\rbrace}
\DeclarePairedDelimiterX\setcond[2]{\{}{\}}{\mathchoice{\,}{}{}{}#1 \;\delimsize\vert\; #2\mathchoice{\,}{}{}{}}

\newcommand{\numVarA}{\nu}
\newcommand{\numVarB}{\mu}
\newcommand{\uniVarA}{x}
\newcommand{\uniVarB}{y}
\newcommand{\varVec}[1]{\bar{#1}}
\newcommand{\numVarVA}{\varVec{\numVarA}}
\newcommand{\numVarVB}{\varVec{\numVarB}}
\newcommand{\uniVarVA}{\varVec{\uniVarA}}
\newcommand{\uniVarVB}{\varVec{\uniVarB}}

\newcommand{\termA}{s}
\newcommand{\termB}{t}
\newcommand{\termVA}{\varVec{\termA}}

\newcommand{\formA}{\Phi}
\newcommand{\formB}{\Psi}

\newcommand{\rel}{R}

\newcommand{\ifp}{\mathsf{ifp}}
\newcommand{\rk}{\mathsf{rk}}
\newcommand{\slv}{\mathsf{slv}}

\newcommand{\definedMatrix}[2]{\eval{#1}{M_{#2}}}

\newcommand{\tup}[1]{\bar{#1}}

\newcommand{\Struct}{\mathfrak{A}}
\newcommand{\StructV}{A}

\newcommand{\vertA}{u}
\newcommand{\vertB}{v}
\newcommand{\vertC}{w}

\newcommand{\bvertA}{x}
\newcommand{\bvertB}{y}
\newcommand{\bvertC}{z}

\newcommand{\PTime}{\textsc{Ptime}}
\newcommand{\CPT}{\ensuremath{\text{CPT}}}
\newcommand{\IFPC}{\ensuremath{\text{IFP+C}}}
\newcommand{\IFPR}{\ensuremath{\text{IFP+R}}}
\newcommand{\IFPS}{\ensuremath{\text{IFP+CS}}}
\newcommand{\IFPRP}[1]{\ensuremath{\IFPR{}_{#1}}}
\newcommand{\IFPSP}[1]{\ensuremath{\IFPS{}_{#1}}}

\newcommand{\RanFinArEquiv}[3]{\equiv^{#1,#2,#3}_{\mathcal{R}}}

\newcommand{\restrictVect}[2]{#1|_{#2}}

\DeclareMathSymbol{\shortminus}{\mathbin}{AMSa}{"39}
\newcommand{\inv}[1]{#1^{\shortminus 1}}

\newcommand{\hiddenEq}{\phantom{{}={}}}

\newcommand{\gblurelem}[2]{g^{#1,#2}}
\newcommand{\Sblurelem}[2]{S^{#1,#2}}

\newcommand\blfootnote[1]{%
	\begingroup
	\renewcommand\thefootnote{}\footnote{#1}%
	\addtocounter{footnote}{-1}%
	\endgroup
}

\title{Separating Rank Logic from Polynomial Time}

\author{Moritz Lichter\\
	RWTH Aachen University\\
	\texttt{lichter@lics.rwth-aachen.de}}

\begin{document}

\maketitle

\begin{abstract}
	In the search for a logic capturing polynomial time,
	the most promising candidates are Choiceless Polynomial Time (CPT)
	and rank logic.
	Rank logic extends fixed-point logic with counting
	by a rank operator over prime fields.
	We show that the isomorphism problem for CFI graphs
	over $\ZZ_{2^i}$ cannot be defined in rank logic,
	even if the base graph is totally ordered.
	However, CPT can define this isomorphism problem.
	We thereby separate rank logic from CPT
	and in particular from polynomial time.
	{\blfootnote{\noindent The research leading to these results has received funding from the European Research Council (ERC) under the European Union’s Horizon 2020 research and innovation programme (EngageS: grant agreement No.\ 820148).}}
\end{abstract}

\section{Introduction}

The quest for a logic capturing polynomial time (\PTime{}) is one of the central open questions in the field of descriptive complexity theory~\cite{Grohe2008}.
This question~\cite{ChandraHarel82} asks whether there is a logic within which we can define exactly the
polynomial-time computable properties of finite relational structures.
The two most promising candidates for such a logic
are Choiceless Polynomial Time and rank logic~\cite{GraedelGrohe15}.
In this article we rule out rank logic as a candidate.
We show that rank logic neither captures \PTime{} nor Choiceless Polynomial Time.

Rank logic was introduced in~\cite{DawarGHL09}
and extends fixed-point logic with counting (\IFPC{}) by a rank operator.
Using this rank operator, the rank of definable matrices can be accessed in the logic.
Multiple variants of rank logic were proposed.
In its first version~\cite{DawarGHL09},
rank logic comes with a rank operator $\rk_p$ for each prime $p$.
If the universe of a finite structure is $A$,
then an $A^k \times A^k$ matrix is defined by a term $\termA(\uniVarVA, \uniVarVB)$
by setting the entry indexed by $(\tup{u},\tup{v})$ to 
the value $\termA(\tup{u}, \tup{v})$,
to which~$\termA$ evaluates in the structure.
The rank operator~$\rk_p$
evaluates to the rank of said matrix over~$\FF_p$.
When considering $A^k \times A^k$ matrices, we call the rank operator $k$-ary.

Crucially, rank logic defines the isomorphism problem of the so-called CFI graphs.
These graphs were given by Cai,~Fürer,~and Immerman~\cite{CaiFI1992}
to separate \IFPC{} from \PTime{}.
From a base graph,
one obtains a CFI graph by replacing every vertex with a particular gadget and by connecting the gadgets of adjacent vertices.
A connection between two gadgets can either be straight or twisted.
The essential point of the construction is that two CFI graphs over the same base graph are isomorphic if and only if
they have the same parity of twisted connections.
In particular, for each base graph there is a pair of non-isomorphic CFI graphs.
The \emph{CFI query} is the task of defining whether the parity of twisted connection of a CFI graph is zero.
CFI graphs implicitly define a linear equation system over~$\FF_2$,
which is solvable if and only if the parity of twists is zero.
That is, the CFI query is decidable by checking these linear equations systems for solvability.
Given a CFI graph, the matrix corresponding to the linear equation system is definable in rank logic and, by considering ranks, also whether it is solvable.

The CFI construction is not restricted to the field~$\FF_2$
but can be generalized to other finite fields~$\FF_p$ or even groups (see e.g.~\cite{BerkholzG17,GradelPakusa19,NeuenSchweitzer17}).
In the case of $\FF_p$, there are $p$ many non-isomorphic CFI graphs for a given base graph.
Grädel and Pakusa~\cite{GradelPakusa19} used CFI graphs over prime fields to show
that extending \IFPC{} by the rank operators~$\rk_p$
is not sufficient to capture \PTime{}.
If the CFI graphs are defined over~$\FF_p$,
then the CFI query over~$\FF_p$ is not definable
only using rank operators~$\rk_q$ for other primes $q \neq p$.
This implies that there is no single formula of this variant of rank logic
defining the CFI query for all CFI graphs over an arbitrary prime field.
An alternative variant of rank logic was proposed in~\cite{GradelPakusa19,Holm11,Laubner2011,pakusa2015}.
It replaces the rank operators~$\rk_p$ for fixed fields
by a \emph{uniform} rank operator~$\rk$,
which defines the prime~$p$ using a formula,
i.e.,~$p$ depends on the structure on which the formula is evaluated.
This second variant of rank logic
defines the CFI query over all prime fields.

Another example that demonstrates the expressiveness of rank logic 
are multipedes.
These structures come also with an isomorphism problem,
which cannot be defined in \IFPC{} but in rank logic~\cite{GurevichS96, Holm11}.
Moreover, rank logic captures $\PTime$
on the class of structures with color class size two~\cite{AbuZaidGraedelGrohePakusa2014}.
An open question in~\cite{DawarGHKP2013} is
whether rank logic can express the solvability of linear equation systems over finite rings rather than only over finite fields.

In this article, we show that
rank logic fails to define the CFI query over the rings~$\ZZ_{2^i}$ for every $i \in \nat$.
As in the case for fields,
we consider the class of CFI graphs over all rings~$\ZZ_{2^i}$
and not a fixed one.
This result eliminates rank logic as a candidate for a logic capturing \PTime{}.
As for~$\FF_2$, the isomorphism problem for CFI graphs over~$\ZZ_{2^i}$
can be translated to a linear equation system over~$\ZZ_{2^i}$.
Hence, we also answer the question for solvability of linear equation systems
over finite rings, where the ring is part of the input and so encoded in a relational structure, in the negative.
Even more, we do not only separate rank logic from \PTime{}
but also from the logic of Choiceless Polynomial Time.

Choiceless Polynomial Time (CPT) was introduced in~\cite{bgs1999}.
It is a logic manipulating hereditarily finite sets
and expresses all common operations on finite sets.
The key point is,
that by definition of \CPT{}
it is impossible to pick an arbitrary element out of a set.
If one wants to process an element in a set, one has to process all of them.
This makes \CPT{} choiceless and thereby isomorphism-invariant.
\CPT{} defines the isomorphism problem of CFI graphs over $\FF_2$
if the base graph is totally ordered~\cite{DawarRicherbyRossman2008}
(which is sufficient to separate \IFPC{} from \PTime).
More generally, \CPT{} captures \PTime{} on the class of structures
with bounded color class size,
where the automorphism group of each color class is abelian~\cite{AbuZaidGraedelGrohePakusa2014}.
This result is established by showing that \CPT{}
defines the solvability problem of a certain class of linear equation systems.
Grädel and Grohe suggested that this class of equation systems
might be a candidate for separating \CPT{} from rank logic~\cite{GraedelGrohe15}.
The result on bounded abelian color classes in~\cite{AbuZaidGraedelGrohePakusa2014}
can be strengthened to not necessarily bounded color classes,
as long as a total order of the automorphism group of each color class is given~\cite{pakusa2015}.
Using this result we show that \CPT{} indeed defines the CFI query over~$\ZZ_{2^i}$ for every $i \in \nat$ for totally ordered base graphs.
Hence, we separate rank logic not only from \PTime{} but also from \CPT{}.

\paragraph{Our Techniques.}
We consider CFI graphs over $\ZZ_{2^i}$.
The automorphism groups of these graphs are $2$-groups.
We show that in this case
formulas containing the uniform rank operator~$\rk$
without a fixed prime
can be translated to formulas only using the~$\rk_2$ operator.
To do so, we use the result
of~\cite{GradelPakusa19}
stating that on CFI graphs,
whose automorphism groups are $p$-groups,
rank logic formulas only using rank operators~$\rk_q$ for $q\neq p$ can be simulated by \IFPC{}.
Hence, if there is a formula defining the CFI query over~$\ZZ_{2^i}$,
we can assume that it only uses the rank operator~$\rk_2$.

To prove that rank logic cannot define the CFI query over~$\ZZ_{2^i}$,
we use game-based methods.
For \IFPC{} there is the well-known Ehrenfeucht-Fraïssé-like pebble game,
called the bijective pebble game~\cite{Hella96}.
It can be used to show that a property is not \IFPC{} definable.
Such a game also exists for the extension by rank operators~\cite{DawarHolm12}.
During the game, ranks of matrices are computed,
which are defined over the two graphs, on which the game is played.
To show that a property is not definable in rank logic,
it must always be possible to play in a way that the ranks
of corresponding matrices for the two graphs are equal.

During the rank pebble game the two matrices (one for each graph) are partitioned into classes.
Then one has to consider all labelings
of the corresponding classes
with values ($0$ and~$1$ in the case of~$\FF_2$). 
This makes it in particular hard to prove that
for every such labeling the two matrices have the same rank.
To overcome this problem, we actually prove a stronger result
and consider matrix similarity.
Dawar and Holm~\cite{DawarHolm12} introduced the invertible-map game,
which requires matrix similarity instead of  matrix equivalence.
In fact, simultaneous similarity of two sequences of matrices is required:
each class in the partition gives rise to one matrix in the sequence,
which only labels that that particular class with~$1$ and all others with~$0$.
Every matrix obtained by labeling the classes
can be expressed as a sum of these matrices labeling exactly one class with~$1$.
Because matrix similarity implies matrix equivalence (and so equality of ranks),
this game is potentially more expressive
in the sense that it possibly distinguishes more structures. 

We use the invertible-map game to prove that there are non-isomorphic CFI graphs that cannot be distinguished by rank logic.
We partition the matrices into orbits
and show that the induced sequences of matrices are simultaneously similar.
Indeed, we show that for
every~$k$ there is an~$i$ such that
Duplicator has a winning strategy
in the $k$-ary invertible-map game played on
CFI graphs over~$\ZZ_{2^i}$
whenever the base graph is sufficiently connected
and its girth is sufficiently large.
Requiring large connectivity is common for these arguments~\cite{CaiFI1992, GradelPakusa19},
but the girth condition is specific to our construction.

The challenge for Duplicator in the invertible-map game
is to come up with an invertible matrix
proving simultaneous similarity of two sequences of matrices.
To construct such matrices, we use two ingredients.
The first consists of sets of local automorphisms we call \emph{blurrers}.
They satisfy some symmetry properties but contain a certain asymmetry.
We use this asymmetry to ``blur''
the twist between two non-isomorphic CFI graphs,
that is, we ``distribute'' it among multiple edges in the graphs.
Because of the symmetry of the blurrers,
it is ``hard'' to detect the blurred twist in the invertible-map game.
In particular, if we only consider $1$-ary rank operators,
blurrers suffice to construct a winning strategy for Duplicator.
Considering arity~$k$
becomes inherently more difficult.
While the argument for the $1$-ary case
is a local argument,
in the $k$-ary case we have to consider $k$\nobreakdash-tuples combining vertices scattered in the graph.
However, we can use the blurrers
such that only for $k$\nobreakdash-tuples
containing vertices of a single ``problematic'' gadget (and possibly other vertices far apart)
the blurrer is not sufficient to prove simultaneous similarity.
Here we use the second ingredient.
We make an arbitrary vertex of the problematic gadget a parameter.
This fixes the problematic gadgets and it suffices to consider
the $(k-1)$\nobreakdash\nobreakdash-tuples where the vertex of the problematic gadgets is removed.
We recurse on the arity
and obtain for every edge, between which we blurred the twist,
a similarity matrix for arity $(k-1)$.
Using the large girth of the graph,
these edges are chosen sufficiently far apart each other.
This is important to combine the $(k-1)$-ary similarity matrices with the blurrer to obtain a similarity matrix for the $k$-ary case.

\paragraph{Related Work.}
Hella~\cite{Hella96} showed that for generalized Lindström quantifiers
the expressiveness strictly increases with the arity.
A similar result can also be given for rank logic~\cite{DawarGHL09, Holm11,Laubner2011}.
In that light, the increased complexity
of our approach for the $k$-ary case
compared to the $1$-ary case is not surprising.

We use the already mentioned result of Grädel and Pakusa~\cite{GradelPakusa19}
to argue that it suffices to consider the rank operator~$\rk_2$ over~$\FF_2$ for CFI graphs over~$\ZZ_{2^i}$.
Consequently, we have to consider the invertible-map game~\cite{DawarHolm12}
for~$\FF_2$ only.
Indeed, it was shown by Dawar, Grädel, and Pakusa~\cite{DawarGraPak19}
that a similar result also holds for the invertible-map game and
the equally expressive linear-algebraic logic:
When considering CFI graphs over~$\FF_2$,
arbitrary linear-algebraic operators over~$\FF_p$ for $p\neq 2$
do not define the CFI query.
Recently, this result was combined with the results of this article
by Dawar, Grädel, and Lichter~\cite{DawarGraedelLichter22} to show that linear-algebraic logic does not capture \PTime{}, either.

Closely related to computing ranks is checking linear equation systems for solvability.
Atserias, Bulatov, and Dawar proved that
\IFPC{} does not define solvability of linear equation systems over finite rings~\cite{AtseriasBD09}.
Solvability of linear equation systems of prime-power fields
is definable in rank logic~\cite{Holm11}.
So, because prime-power fields reduce to prime fields in rank logic,
it is conceivable that a variant of rank operators using prime-power fields
does not define the CFI query over~$\ZZ_{2^i}$, either.

While \IFPC{} fails to capture \PTime{} for CFI graphs,
there are many other graph classes on which \IFPC{} captures \PTime{}.
These include, e.g.,~graphs with excluded minors~\cite{Grohe2010}
and graphs with bounded rank width~\cite{GroheN19}.
While showing that rank logic defines the CFI query for prime fields is rather simple~\cite{DawarGHL09},
this is a non-trivial result for \CPT{}.
The already mentioned result by
Dawar, Richerby, and Rossman~\cite{DawarRicherbyRossman2008}
uses deeply nested sets 
and is restricted to totally ordered base graphs.
This result was strengthened by Pakusa, Schalthöfer, and Selman~\cite{PakusaSchalthoeferSelman2016}
to base graphs with logarithmic color class size.
Recently,
the result for bounded abelian color classes by Abu Zaid, Grädel, Grohe, and Pakusa~\cite{AbuZaidGraedelGrohePakusa2014}
was extended by Lichter and Schweitzer~\cite{LichterS21}
to graphs with bounded color classes with dihedral colors
and also for certain structures of arity~$3$.

\paragraph{Structure of this Article.}
After providing some preliminaries in Section~\ref{sec:preliminaries},
we introduce variants of rank logic in Section~\ref{sec:rank-logc}
and the invertible-map game in Section~\ref{sec:invertible-map-game}.
Then we give a CFI construction suitable for our arguments in Section~\ref{sec:cfi-structures}.
Next, we discuss matrices defined over CFI structures in Section~\ref{sec:matrices-cfi}
and develop a criterion for invertibility of such matrices over $\FF_2$.
This will be used in Section~\ref{sec:1-ary-case},
where we treat the case of $1$-ary rank operators
and introduce blurrers.
Section~\ref{sec:active-region}
defines the notion of the active region of a matrix.
This is used in the case of general $k$-ary
rank operators in Section~\ref{sec:higher-arities}
to successfully combine the matrices obtained from recursion.
There we also generalize blurrers to the $k$-ary case.
Finally, we separate rank logic from \CPT{} in Section~\ref{sec:separating}
and conclude with a discussion in Section~\ref{sec:discussion}.

\newcommand{\comp}{C}

\newcommand{\eval}[2]{#2^{#1}}

\newcommand{\oneVec}{\mathbbm{1}}

\newcommand{\sig}{\tau}
\newcommand{\spleq}{\preceq}
\newcommand{\countingExt}[1]{#1^{\#}}

\newcommand{\textop}[1]{\mathsf{#1}} %{\operatorname{#1}}

\newcommand{\neighbors}[2]{N_{#1}(#2)}
\newcommand{\neighborsK}[3]{N_{#1}^{#2}(#3)}
\newcommand{\autgrp}[1]{\textop{Aut}(#1)}

\newcommand{\auto}{\autoA}
\newcommand{\autoA}{\phi}
\newcommand{\autoB}{\psi}

\newcommand{\orig}[1]{\textop{orig}(#1)}

\newcommand{\minRad}{r}

\newcommand{\activeRegionSym}{\mathsf{A}}
\newcommand{\activeRegion}[1]{\mathsf{A}(#1)}
\newcommand{\inactiveRegionB}[2]{\mathsf{N}(#1, #2)}
\newcommand{\activeRegionB}[2]{\mathsf{A}(#1, #2)}

\newcommand{\activeRegionIdx}[3]{\mathsf{A}^{#1,#2}(#3)}
\newcommand{\inactiveRegionBIdx}[4]{\mathsf{N}^{#1,#2}(#3,#4)}
\newcommand{\activeRegionBIdx}[4]{\mathsf{A}^{#1,#2}(#3,#4)}

\newcommand{\charMat}[1]{\chi^{#1}}

\newcommand{\orbs}[2]{\textop{orbs}_{#1}(#2)}

\newcommand{\pVertA}[0]{p}
\newcommand{\pTupA}[0]{\tup{\pVertA}}

\newcommand{\kroneckerEq}[2]{\delta_{#1,#2}}
\newcommand{\spath}[0]{\tup{s}}
\newcommand{\distance}[3]{\textop{dist}_{#1}(#2,#3)}

\newcommand{\CFIop}{\mathsf{CFI}}

\newcommand{\CFIgraph}[3]{\CFIop_{#1}(#2,#3)}

\newcommand{\GraphClass}{\mathcal{K}}
\newcommand{\CFIClass}[2]{\CFIop_{#1}(#2)}
\newcommand{\CFITwoClass}[1]{\CFIop_{2^\omega}(#1)}

\newcommand{\StructA}{\mathfrak{A}}
\newcommand{\StructB}{\mathfrak{B}}
\newcommand{\StructC}{\mathfrak{H}}
\newcommand{\StructVA}{A}
\newcommand{\StructVB}{B}
\newcommand{\StructVC}{H}
\newcommand{\PartA}{\mathbf{P}}
\newcommand{\PartB}{\mathbf{Q}}
\newcommand{\PartC}{\mathbf{R}}
\newcommand{\BlockA}{P}
\newcommand{\BlockB}{Q}
\newcommand{\BlockC}{R}
\newcommand{\tupA}{\tup{u}}
\newcommand{\tupB}{\tup{v}}
\newcommand{\tupC}{\tup{w}}
\newcommand{\invertMapGame}[3]{\mathcal{M}^{#1,#2,#3}}
\newcommand{\invertibleMapEquiv}[3]{\equiv^{#1,#2,#3}_{\mathcal{M}}}
\newcommand{\rankPebbleEquiv}[3]{\equiv^{#1,#2,#3}_{\mathcal{R}}}
\newcommand{\invertMapGameBR}[4]{\mathcal{M}^{#1,#2,#3}_{#4}}
\newcommand{\invMapBij}[0]{\lambda}

\newcommand{\SymSetGroup}[1]{\textop{Sym}(#1)}
\newcommand{\group}{\Gamma}
\newcommand{\perm}{\sigma}

\newcommand{\twistVal}{\theta}
\newcommand{\blurrer}{\Xi}
\newcommand{\blurElem}{\xi}
\newcommand{\sblurElem}{\blurElem_{\mathsf{tw}}}
\newcommand{\invsblurElem}{\inv{\blurElem\mathrlap{_{\mathsf{tw}}}}}
\newcommand{\blurElemFix}{\blurElem_{\mathsf{twst}}}

\newcommand{\countVects}[3]{\#_{#1,#2}(#3)}

\newcommand{\occIndices}[1]{\textop{Occ}(#1)}
\newcommand{\fixTwistType}{\tau}
\newcommand{\fixTwistTypeAut}[2]{\autoA^{#1,#2}_{\tau}}

\newcommand{\pot}{q}

\newcommand{\gadgetiso}[2]{#2(#1)}
\newcommand{\pathiso}[2]{\vec{\pi}[#1, #2]}
\newcommand{\stariso}[2]{\pi^*[#1, #2]}

\newcommand{\idmat}{\mymathbb{1}}
\newcommand{\zeromat}{\mymathbb{0}}

\newcommand{\scenter}{z}
\newcommand{\stip}{t}

\newcommand{\ncomp}[2]{C_{#1}(#2)}
\newcommand{\remZ}[1]{#1^{\shortminus z}}

\section{Preliminaries}
	\label{sec:preliminaries}
	We denote the set $\set{1,\dots, k}$ by $[k]$.
	Let~$N$ and~$I$ be finite sets.
	The set of $I$-indexed tuples over~$N$ is denoted by $N^I$.
	For a tuple $\tup{t} \in N^I$ the entry for index $i \in I$ is written as $\tup{t}(i)$.
	For $k \in \nat$, $\tup{t} \in N^k=N^{[k]}$, and $i \leq k$ we also write~$t_i$
	for the $i$-th entry.
	The concatenation of two tuples  $\tup{s}\in N^k$ and $\tup{t}\in N^\ell$ is denoted by $\tup{s}\tup{t} \in N^{k + \ell}$.
	The restriction of $\tup{t} \in N^I$ to $K\subseteq I$ is denoted by $\restrictVect{\tup{t}}{K} \in N^{K}$.
	
	For two finite index sets~$I$ and~$J$,
	an $I \times J$ matrix~$M$ over~$N$ is a map $M \colon I \times J \to N$.
	We write $M(i,j)$ for the entry at position $(i,j)$.
	The identity matrix is denoted by~$\idmat$
	and the zero matrix by~$\zeromat$.
	
	We write~$\ZZ_j$ for the ring of integers modulo~$j$.
	Its elements are $\set{0,\dots, j-1}$.
	For a tuple $\tup{a} \in \ZZ_j^I$
	we set $\sum \tup{a} := \sum_{i \in I} \tup{a}(i)$
	and likewise for a function $f \colon I \to \ZZ_j$.
	
	A (relational) signature
	$\sig = \set{\rel_1, \dots, \rel_\ell}$ is a set of relation symbols
	with associated arities $r_i \in \nat$ for each $i \in [\ell]$.
	A~$\sig$-structure~$\Struct$ is a tuple
	${\Struct = (\StructV, \rel_1^\Struct, \dots, \rel_\ell^\Struct)}$
	where $\rel_i^\Struct \subseteq\StructV^{r_i}$ for all $i \in [\ell]$.
	The universe of~$\Struct$ is always denoted by~$\StructVA$.
	In this article, we only consider finite structures.
	A \defining{pebbled structure} is a pair $(\Struct, \tupA)$
	of a relational structure and a tuple $\tupA \in \StructV^k$.
	Two pebbled structures $(\StructA, \tupA)$ and $(\StructB, \tupB)$
	are isomorphic,
	if there is an isomorphism $\autoA \colon \StructVA \to \StructVB$
	such that $\autoA((\StructA, \tupA)) = (\StructB, \tupB)$.
	That is, every isomorphism has to map~$\tupA$ to~$\tupB$.
	An automorphism of $(\StructA, \tupA)$ is an isomorphism $(\StructA, \tupA) \to (\StructA, \tupA)$.
	The automorphism group of $(\StructA, \tupA)$ is denoted by $\autgrp{(\StructA, \tupA)}$.
	
	Let $G = (V,E)$ be a simple graph.
	For vertices $\bvertA,\bvertB \in V$
	we denote their distance in~$G$ by $\distance{G}{\bvertA}{\bvertB}$.
	For two sets $X,Y \subseteq V$ we set
	$\distance{G}{X}{Y} := \min_{\bvertA \in X,\bvertB \in Y} \distance{G}{\bvertA}{\bvertB}$
	and likewise $\distance{G}{\bvertA}{Y} := \distance{G}{\set{\bvertA}}{Y}$
	for a vertex~$\bvertA$ and a set $Y \subseteq V$.
	The set of neighbors of a vertex $\bvertA \in V$
	is denoted by $\neighbors{G}{\bvertA}$.
	The $k$-neighborhood of~$\bvertA$ in~$G$ is
	$\neighborsK{G}{k}{\bvertA} := \setcond{\bvertB \in V}{\distance{G}{\bvertA}{\bvertB} \leq k}$.
	The induced subgraph of $G$ by $W\subseteq V$ is $G[W]$.
	The graph~$G$ is \defining{$k$-connected},
	if $|V| \geq k$ and for every $V' \subseteq V$ of size at most $k-1$,
	$G \setminus V'$ is connected.
	That is, after removing $k-1$ vertices,~$G$ is still connected.
	The \defining{girth} of~$G$ is the length of the shortest cycle in~$G$.
	
	Let~$\group$ be a finite permutation group with domain~$N$
	and let~$p$ be a prime.
	If for every $\perm \in \group$ there is an~$\ell$
	such that~$\perm$ is of order $p^\ell$,
	i.e, $\perm^{(p^\ell)}$ is the identity,
	then~$\group$ is called a \defining{$p$-group}.
	The \defining{orbit} of $n \in N$ is the set $\setcond{\perm(n)}{\perm \in \group}$.
	In this way,~$N$ is partitioned into orbits.
	This notation generalizes to $k$\nobreakdash-tuples.
	A \defining{$k$-orbit} is a maximal set $\BlockA \subseteq N^k$,
	such that for every $\tup{n},\tup{m} \in \BlockA$,
	there is a $\perm \in \group$ such that $\perm(\tup{n}) = \tup{m}$.
	We write $\orbs{k}{\group}$ for the set of $k$-orbits of $\group$.
	The group $\group$ is \defining{transitive}
	if $|\orbs{1}{\group}|=1$.
	If additionally $|\group| = |N|$, then~$\group$ is called \defining{regular}.
	The $k$-orbits of a pebbled relational structure $(\StructA, \tupA)$
	are $\orbs{k}{(\StructA, \tupA)} := \orbs{k}{\autgrp{(\StructA, \tupA)}}$.

\section{Rank Logic}
\label{sec:rank-logc}
In this section we consider rank logic,
an extension of inflationary fixed-point logic with counting
by a rank operator.
Let $\Struct = (\StructV, \rel_1^\Struct, \dots, \rel_\ell^\Struct)$ be a relational $\sig$\nobreakdash-structure.
We set $\countingExt{\sig} := \sig \disunion \{\cdot, +, 0,1\}$
and
$\countingExt{\Struct} := (\StructV, \rel_1^\Struct, \dots, \rel_\ell^\Struct, \nat, \cdot, +, 0, 1)$
to be the two-sorted $\countingExt{\sig}$\nobreakdash-structure\footnote{This is the only non-finite structure in this article.}
that is the disjoint union of~$\Struct$ and~$\nat$.

\paragraph{Fixed-Point Logic with Counting.}
We introduce \IFPC{}, the fixed-point logic with counting (proposed in~\cite{immerman87}, also see~\cite{Otto1997}).
Let $\sig$ be a signature.
\IFPC{} is a two-sorted logic using the signature~$\countingExt{\sig}$
with \emph{element} variables ranging over the universe of the input structure
and \emph{number} variables ranging over the natural numbers.
We use the letters~$\uniVarA$ and~$\uniVarB$ for element variables,
Greek letters~$\numVarA$ and~$\numVarB$ for numeric variables,~%
$\formA$ and~$\formB$ for formulas,
and letters~$\termA$ and~$\termB$ for numeric terms.
For a tuple of variables or terms we write~$\uniVarVA$,~$\numVarVA$,
and~$\termVA$ respectively.

\IFPC{}-formulas are built from first-order formulas, a fixed-point operator, and counting terms.
To ensure polynomial-time evaluation, quantification over numeric variables
needs to be bounded:
Whenever~$\formA$ is an \IFPC{}-formula,~$\numVarA$ is a numeric, possibly free variable in~$\formA$ and~$\termA$ is a closed numeric term,
then 
\[Q \numVarA \leq \termA.\qspace\formA\]
is an \IFPC{}-formula, where $Q \in\set{\forall, \exists}$.
We now consider (inflationary) fixed-points. Let~$\rel$ be a relation symbol
to define using a fixed-point.
The relation~$\rel$ can contain both elements of the universe and numbers as follows.
Let~$\formA$ be an \IFPC{}-formula,~$\uniVarVA$ and~$\numVarVA$ be variables,
and~$\termVA$ be a tuple of~$|\numVarVA|$ many closed numeric terms that bound the values of~$\numVarVA$. Then
\[\left[\ifp \rel \uniVarVA\numVarVA \leq \termVA.\qspace\formA \right](\uniVarVA\numVarVA)\]
is an \IFPC{}-formula.
To relate element and numeric variables, \IFPC{} possesses counting terms
that count the number of different values for some variables satisfying a formula.
As before, let~$\formA$ be an \IFPC{}-formula,~$\uniVarVA$ and~$\numVarVA$ be variables,
and~$\termVA$ be a tuple of~$|\numVarVA|$ many closed numeric terms which bound the values of~$\numVarVA$. Then
\[ \#\uniVarVA\numVarVA \leq \termVA.\qspace\formA \]
is a numeric \IFPC{}-term.

Let~$\Struct$ be a $\sig$\nobreakdash-structure.
An \IFPC{}-formula (or term) is evaluated over $\countingExt{\Struct}$.
For a numeric term $\termA(\uniVarVA\numVarVA)$ we denote by
$\eval{\Struct}{\termA} \colon \StructV^{|\uniVarVA|} \times \nat^{|\numVarVA|} \to \nat$
the function that maps the possible values of the free variables of~$\termA$
to the value that~$\termA$ takes in~$\countingExt{\Struct}$.
Similarly, for a formula $\formA(\uniVarVA\numVarVA)$
we write $\eval{\Struct}{\formA} \subseteq \StructV^{|\uniVarVA|} \times \nat^{|\numVarVA|}$ for the set of values for the free variables
satisfying $\formA$.
Then, e.g., the evaluation of a counting term for a formula $\formA(\uniVarVB\uniVarVA\numVarVB\numVarVA)$ is defined as
\[\eval{\Struct}{(\#\uniVarVA\numVarVA \leq \termVA.\qspace\formA)}(\tupA\varVec{m}) := \left|\setcond*{\tupC\varVec{n} \in \StructV^{|\uniVarVA|} \times \nat^{|\numVarVA|}}{n_i \leq \termA_i^\StructA \text{ for all } i \in [|\numVarVA|] \text{ and } \tupA\tupC\varVec{m}\varVec{n} \in \eval{\Struct}{\formA}}\right|.\]

\paragraph{Rank  Logic.}
We now consider the extension of \IFPC{} by
the \emph{uniform} rank operator. We follow the definition in~\cite{Holm11}.
Let $\termA(\uniVarVA, \uniVarVB)$ be a numeric term such that $k := |\uniVarVA| = |\uniVarVB|$
and~$\termB$ be a closed numeric term.
Then
\[
	\rk (\uniVarVA, \uniVarVB).\qspace(\termA,\termB)
\]
is a numeric term.
We say for convenience that~$k$ is the \defining{arity} of the operator,
although it is actually~$2k$.
The logic \IFPR{} is the extension of \IFPC{} by the uniform rank operator~$\rk$.
We restricted the definition to square matrices,
but this does not limit the expressive power.
The rank operator is evaluated as follows.
Let~$\Struct$ be a $\sig$\nobreakdash-structure.
The term~$\termA$ defines an $\StructV^k \times \StructV^k$ matrix $\definedMatrix{\Struct}{\termA}$ over $\nat$:
\[\definedMatrix{\Struct}{\termA}(\tupA,\tupB) := \eval{\Struct}{\termA}(\tupA,\tupB).\]
Finally, we define 
$\eval{\Struct}{(\rk (\uniVarVA, \uniVarVB).\qspace(\termA,\termB))}$:
The rank operator evaluates the
rank of $(\definedMatrix{\Struct}{\termA}\bmod p)$ over~$\FF_p$
if $p := \eval{\Struct}{\termB}$ is prime
and to~$0$ otherwise.
We omitted parameters for readability.

For a set of prime numbers~$\Omega$,
we set $\IFPRP{\Omega} $
to be the variant of $\IFPR{}$,
in which we have instead of the uniform rank operator~$\rk$
a different rank operator~$\rk_p$
for every~$\FF_p$ such that $p \in \Omega$.
That is, we have to fix the field in the formula independently of the structure.
This is not the case for the operator~$\rk$,
where we can determine the value for~$p$
by another term that evaluates differently for different structures.

\paragraph{Choiceless Polynomial Time.}
Choiceless Polynomial Time (CPT) is a logic different from \IFPC{}.
CPT-formulas manipulate hereditarily finite sets.
They are choiceless in the sense that they either process all elements of such sets or none.
It is not possible to pick an arbitrary element from a set.
By these conditions, all sets constructed by a \CPT{}-term are closed under automorphisms of the input structure.
Evaluation in polynomial time is guaranteed by explicit polynomial bounds
on the number of steps and sizes of the constructed sets.
We omit a formal definition of \CPT{} here
because it is not needed in this article.
For a formal definition we refer to~\cite{GraedelGrohe15,pakusa2015}.

We review two results for \CPT{}:
A relational $\sig$\nobreakdash-structure~$\Struct$ has \defining{$q$-bounded colors},
if one relation ${\spleq} \in \sig$ is a total preorder
partitioning the universe into $\spleq$-equivalence classes,
called \defining{color classes}, of size at most~$q$.
The structure~$\Struct$ has \defining{abelian colors},
if the induced substructure of every color class 
has an abelian automorphism group.
\begin{theorem}[\cite{AbuZaidGraedelGrohePakusa2014}]
	\label{thm:CPT-bounded-abelian-colors}
	CPT captures \PTime{} on $q$-bounded relational structures
	with abelian colors.
\end{theorem}
This result can be strengthened
from bounded color class size
to ordered colors.
A $\sig$-structure with \defining{ordered colors}
is a tuple $(\StructA, \mymathbb{\Gamma})$,
where~$\StructA$ is a relational $\sig$-structure
with color classes $C_1, \dots, C_n$
and $\mymathbb{\Gamma} = \setcond{(\Gamma_i, \leq_i)} {i \in [n]}$
is a family of ordered permutations groups
such that~$\Gamma_i$ is a transitive group with domain~$C_i$ for every $i \in[n]$.
Note that structures with ordered colors are, without further encoding,
not relational structures because~$\mymathbb{\Gamma}$ is a higher-order object.
However, we only define given a relational $\sig$-structure~$\Struct$
the family~$\mymathbb{\Gamma}$ of ordered permutation groups
in \CPT{} and thus can represent~$\mymathbb{\Gamma}$ as a hereditarily finite set.
\begin{theorem}[\cite{pakusa2015}]
	\label{thm:CPT-ordered-abelian-colors}
	\CPT{} captures \PTime{} on structures with ordered abelian colors.
\end{theorem}

\section{The Invertible-Map Game}
\label{sec:invertible-map-game}
Undefinability results for~\IFPC{}
are often achieved by the embedding of \IFPC{} into infinitary bounded variable counting logic (see~\cite{Otto1997})
and exploiting an equally expressive Ehrenfeucht-Fraïssé-like pebble game.
This game, called the bijective $k$-pebble game~\cite{Hella96},
is a game between two players called Spoiler and Duplicator
played on two pebbled structures.
The aim of Spoiler is to prove that the two structures can be distinguished 
in infinitary $k$-variable counting logic,
where Duplicator tries to show the converse.
Such a game also exists for rank logic (called matrix-equivalence game in~\cite{DawarHolm12}).
It extends the bijective pebble game with ranks.
Instead of looking at this pebble game,
we consider the invertible-map game~\cite{DawarHolm12}.
Its distinguishing power is at least as strong as the one of the rank-pebble game
in the sense that if Duplicator has a winning strategy in the invertible-map game,
then Duplicator has a winning strategy in the rank-pebble game, too.
Hence, to show that rank logic cannot distinguish two structures,
it suffices to show that Duplicator has a winning strategy in the invertible-map game.
The game is defined as follows:

Let~$k$ and~$m$ be two positive integers such that $2k \leq m$ and
let~$\Omega$ be a finite and nonempty set of primes.
The \defining{invertible-map game} $\invertMapGame{m}{k}{\Omega}$  is played on two pebbled structures
$(\StructA, \tupA)$ and $(\StructB, \tupB)$
with $|\tupA| = |\tupB| \leq m$
of the same signature. For each structure there are~$m$ pebbles labeled
with $1,\dots,m$, where on~$\vertA_i$ and~$\vertB_i$
there are pebbles with the same label for all $i \in[|\tupA|]$.
That is, if $|\tupA| < m$, some of the pebbles are not used.
There are two players called Spoiler and Duplicator.
If $|\StructVA| \neq |\StructVB|$, then Spoiler wins the game.
Otherwise a round of the game proceeds as follows:
\begin{enumerate}
	\item Spoiler chooses a prime $p \in \Omega$ and picks up~$2k$ many pebbles
	from~$\StructA$ and the corresponding pebbles (with the same labels) from~$\StructB$.
	\item Duplicator picks a partition~$\PartA$ of $\StructVA^k \times \StructVA^k$ and another one~$\PartB$ of $\StructVB^k \times \StructVB^k$
	such that $|\PartA| = |\PartB|$.
	Furthermore, Duplicator picks an invertible $\StructVA^k \times \StructVB^k$ matrix~$S$ over~$\FF_p$,
	such that the matrix induces a total and bijective  map $\invMapBij\colon\PartA \to \PartB$
	defined by $\BlockA \mapsto \BlockB$ if and only if $ \charMat{\BlockA} = S \cdot \charMat{\BlockB} \cdot \inv{S}$.
	Here~$\charMat{\BlockA}$ (respectively~$\charMat{\BlockB}$) is the characteristic $\StructVA^k \times \StructVA^k$ matrix over~$\FF_p$ of~$\BlockA$
	(respectively the $\StructVB^k \times \StructVB^k$ matrix over~$\FF_p$ of~$\BlockB$)
	which satisfies that $\charMat{\BlockA}(\tupA',\tupB') = 1$ if $\tupA'\tupB' \in \BlockA$ and $\charMat{\BlockA}(\tupA',\tupB') = 0$ otherwise.
	To say it differently, 
	Duplicator has to pick a bijection $\invMapBij \colon \PartA \to \PartB$
	and an invertible $\StructVA^k \times \StructVB^k$ matrix~$S$
	satisfying $\charMat{\BlockA}  = S \cdot \charMat{\invMapBij(\BlockA)} \cdot \inv{S}$
	for all $\BlockA \in \PartA$,
	i.e., the characteristic matrices of~$\PartA$ and~$\PartB$ are simultaneously similar.
	\item Spoiler chooses a block $\BlockA \in \PartA$, a tuple $\tupC \in \BlockA$, and a tuple $\tupC' \in \invMapBij(\BlockA)$.
	Then for each $i \in [2k]$ Spoiler
	places a pebble on~$\vertC_i$
	and the corresponding pebble on~$\vertC'_i$.
\end{enumerate}
After a round, Spoiler wins the game if the pebbles do not define a local isomorphism or if Duplicator was not able to respond with a matrix satisfying the conditions above.
Duplicator wins the game if Spoiler forever fails to win.
Spoiler has a winning strategy if Spoiler can win the game starting at $(\StructA, \tupA)$ and $(\StructB, \tupB)$ in any case independently of the actions of Duplicator.
Likewise, Duplicator has a winning strategy, if Duplicator can always win the game.
In that case, we write $(\StructA, \tupA) \invertibleMapEquiv{m}{k}{\Omega} (\StructB, \tupB)$.
Finally,
we consider the game with a bounded number of rounds:
The $\ell$-round invertible-map game $\invertMapGameBR{m}{k}{\Omega}{\ell}$
proceeds exactly as $\invertMapGame{m}{k}{\Omega}$
but stops after~$\ell$ rounds.
Duplicator wins, if Spoiler did not win in these rounds.
In the following, we use the invertible-map game
instead of the rank-pebble game
because it allows us to prove a stronger result
and simplifies proofs.

\begin{lemma}[\cite{DawarHolm12}]
	\label{lem:invertible-map-refines-rank-logic}
	Let $\GraphClass$ be a class of finite $\sig$-structures
	and $P \subseteq \GraphClass$.
	If for every $k,m \in \nat$ with $2k \leq m$
	and every finite and nonempty set of primes~$\Omega$,
	there is a pair of structures $(\StructA, \StructB)$,
	such that $\StructA \in P$,
	$\StructB \notin P$,
	and $\StructA \invertibleMapEquiv{m}{k}{\Omega} \StructB$,
	then~$P$ is not $\IFPRP{\mathbb{P}}$ definable,
	where~$\mathbb{P}$ is the set of all primes.
\end{lemma}
If we fix a finite set of primes~$\Omega$ in Lemma~\ref{lem:invertible-map-refines-rank-logic},
then~$P$ is not $\IFPRP{\Omega}$ definable~(see~\cite{DawarHolm12})
because a $P$-defining $\IFPRP{\Omega}$-formula
implies a winning strategy of Spoiler in the $\invertMapGame{m}{k}{\Omega}$ game.
Lemma~\ref{lem:invertible-map-refines-rank-logic}
is proved in~\cite{DawarHolm12} for the
$(m,k,\Omega)$-rank-pebble game,
which induces the equivalence~$\rankPebbleEquiv{m}{k}{\Omega}$.
Then the authors show
that~$\invertibleMapEquiv{m}{k}{\Omega}$ refines $\rankPebbleEquiv{m}{k}{\Omega}$.
It is an open problem
whether the equivalence~$\invertibleMapEquiv{m}{k}{\Omega}$
strictly refines $\RanFinArEquiv{m}{k}{\Omega}$.

\section{CFI Structures}
\label{sec:cfi-structures}
In this section we define a variant of the well-known CFI graphs.
Starting from a so-called base graph,
for every vertex in the base graph a gadget is constructed.
In the seminal paper of Cai,~Fürer,~and Immerman~\cite{CaiFI1992}
these gadgets consist of inner and outer vertices,
where the latter are pairs of vertices.
Each outer vertex pair induces the automorphism group~$\ZZ_2$
and the inner vertices realize
the automorphism group $\setcond{\tup{a} \in \ZZ_2^d}{\sum \tup{a} = 0}$
of its~$d$ adjacent outer vertex pairs.
Whenever two vertices in the base graph are adjacent,
the two corresponding outer vertex pairs of the two gadgets are connected.
Such a connection can either be ``straight'' or ``twisted''.
This construction generalizes to other groups than~$\ZZ_2$.
In~\cite{NeuenSchweitzer17} a construction of gadgets for
general abelian groups can be found.
We are interested in cyclic groups~$\ZZ_{2^\pot}$.
The following construction only uses the inner vertices
and directly connects the inner vertices of two gadgets.
For~$\ZZ_2$, this approach is given in \cite{Furer2001}.

A \defining{base graph} is a simple, connected, and totally ordered graph.
Let $G=(V,E,\leq)$ be a base graph.
Consider the additive group of $\ZZ_{2^\pot}$.
For each vertex $\bvertA \in V$  we define a gadget
consisting of vertices~$\StructV_\bvertA$ and two families of relations:
\begin{align*}
	\StructV_\bvertA &:= \setcond*{\tup{a} \in \ZZ_{2^\pot}^{\neighbors{G}{\bvertA}}}{\sum \tup{a} = 0}, & \bvertA \in V,\\
	I_{\bvertA,\bvertB} &:= \setcond*{
	(\tup{a},\tup{b}) \in \StructV_\bvertA^2} {\tup{a}(\bvertB) = \tup{b}(\bvertB)}, &\bvertA \in V, \bvertB\in \neighbors{G}{\bvertA},\\
	C_{\bvertA,\bvertB} &:= \setcond*{(\tup{a},\tup{b}) \in \StructV_\bvertA^2}{\tup{a}(\bvertB)+1=\tup{b}(\bvertB)}, &\bvertA \in V, \bvertB\in \neighbors{G}{\bvertA}.
\end{align*}
Consider the sets $\StructV_{\bvertA,\bvertB,c} := \setcond{\tup{a} \in \StructV_\bvertA}{\tup{a}(\bvertB) = c}$ for $\bvertB \in \neighbors{G}{\bvertA}$
and $c \in \ZZ_{2^\pot}$.
The relation~$I_{\bvertA,\bvertB}$ realizes these sets
by disjoint cliques, one for each $\StructV_{\bvertA,\bvertB,c}$.
The relation $C_{\bvertA,\bvertB}$ induces a directed cycle 
$\StructV_{\bvertA,\bvertB,c}, \StructV_{\bvertA,\bvertB,c+1}, \dots,  \StructV_{\bvertA,\bvertB,c+2^\pot-1}$
on these sets for a fixed~$\bvertB$
by adding directed complete bipartite graphs between subsequent cliques.
In that way, the relation $C_{\bvertA,\bvertB}$ realizes the group~$\ZZ_{2^q}$
on the sets $\StructV_{\bvertA,\bvertB,c}$.
By the condition $\sum \tup{a} = 0$ on the vertices in~$\StructV_\bvertA$,
a gadget thereby has an automorphism group isomorphic to
$\setcond{\tup{a} \in \ZZ_{2^q}^d}{\sum\tup{a} = 0}$
where~$d$ is the degree of~$\bvertA$.

Now we connect gadgets.
We first extend the order~$\leq$ to the lexicographical order on tuples of vertices of~$G$ and further to sets of such tuples.
Let $g\colon E \to \ZZ_{2^\pot}$ be a function
defining the values by which the edges are twisted.
For every edge $\set{\bvertA,\bvertB} \in E$ we connect the gadgets of the incident vertices.
We obtain the CFI structure
\begin{align*}
	\CFIgraph{2^\pot}{G}{g} &:= (\StructV, \rel_I, \rel_C, \rel_{E,0},\dots,\rel_{E,2^\pot-1}, \spleq)
\intertext{as follows:}
	E_{\set{\bvertA,\bvertB},c} &:= \setcond*{\set{\tup{a},\tup{b}}}{\tup{a} \in \StructV_\bvertA, \tup{b} \in \StructV_\bvertB, \tup{a}(\bvertB) + \tup{b}(\bvertA) =  c },\qquad\qquad \set{\bvertA,\bvertB} \in E, c \in \ZZ_{2^\pot},\\
	\spleq &:= \setcond*{(\tup{a},\tup{b})}{\tup{a} \in \StructV_\bvertA, \tup{b} \in \StructV_\bvertB, \bvertA \leq \bvertB},\\
	\rel_I &:= \setcond*{(\tup{a},\tup{b}, \tup{a}',\tup{b}')}{
		\setcond{(\bvertA,\bvertB)}{(\tup{a},\tup{b}) \in I_{\bvertA,\bvertB}}
		\leq 
		\setcond{(\bvertA',\bvertB')}{(\tup{a}',\tup{b}') \in I_{\bvertA',\bvertB'}}
	},\\
	\rel_C &:= \setcond*{(\tup{a},\tup{b}, \tup{a}',\tup{b}')}{
		\setcond{(\bvertA,\bvertB)}{(\tup{a},\tup{b}) \in C_{\bvertA,\bvertB}}
		\leq 
		\setcond{(\bvertA',\bvertB')}{(\tup{a}',\tup{b}') \in C_{\bvertA',\bvertB'}}
},
\end{align*}\vspace{-20pt}
\begin{align*}
	\StructV &:= \bigcup_{\bvertA \in V} \StructV_\bvertA,&
	\rel_{E,c} &:= \bigcup_{e \in E} E_{e,c + g(e)}.
\end{align*}
The unions above are meant to be disjoint.
The relations~$I_{\bvertA,\bvertB}$ (and similarly~$C_{\bvertA, \bvertB}$)
are encoded by~$\rel_I$ (and $\rel_C$) as follows:
All edges $(\tup{a}, \tup{b}) \in \StructVA_{\bvertA}^2$
are partitioned according to the set of base vertices~$\bvertB$
such that $(\tup{a}, \tup{b}) \in I_{\bvertA,\bvertB}$.
The partition is given by
the equivalence classes of~$\rel_I$
(seen as equivalence on pairs).
The relations $I_{\bvertA,\bvertB}$ (respectively $C_{\bvertA, \bvertB}$)
are unions of $\rel_I$-equivalence classes (respectively $\rel_C$-equivalence classes).

\begin{definition}[Origin]We say that the vertices $\tup{a} \in \StructV_\bvertA$
\defining{originate} from~$\bvertA$
or that their \defining{origin} is~$\bvertA$
and write $\orig{\tup{a}} := \bvertA$.
We extend this to tuples
and define the \defining{origin} of $(\tup{a}_1,\dots,\tup{a}_j) \in \StructV^j$ as
$\orig{(\tup{a}_1,\dots,\tup{a}_j)} := (\orig{\tup{a}_1},\dots,\orig{\tup{a}_j})$.
We will often view $\orig{(\tup{a}_1,\dots,\tup{a}_j)}$ as the set $\set{\orig{\tup{a}_1}, \dots, \orig{\tup{a}_j}}$
and write $\bvertA \in \orig{(\tup{a}_1, \dots, \tup{a}_j)}$.
If~$M$ is a set of tuples of the same origin,
we set $\orig{M} := \orig{\tup{a}}$ for some (and thus all) $\tup{a}\in M$.
For a set $W \subseteq V$,
we define the \defining{origin induced substructure}
\[\CFIgraph{2^\pot}{G}{g}[W] := \CFIgraph{2^\pot}{G}{g}[\setcond{\tup{a}}{\orig{\tup{a}} \in W}]\]
to be the substructure induced by all vertices whose origin is contained in~$W$.
\end{definition}
It will be always clear from the context
whether we refer to the origin induced substructure (or just a standard induced substructure).
In that case~$W$ and the universe of the CFI structure are disjoint.

For CFI structures it is well-known
that $\CFIgraph{2^\pot}{G}{g} \iso \CFIgraph{2^\pot}{G}{f}$
if and only if $\sum_{e \in E} g(e) = \sum_{e \in E} f(e)$.
That is, there are up to isomorphism~$2^\pot$ many CFI structures
of the base graph~$G$.

\begin{lemma}
	\label{lem:cfi-autgroup-abelian-two-group}
	The automorphism group $\autgrp{\CFIgraph{2^\pot}{G}{g}}$\hspace{-1pt} of $\CFIgraph{2^\pot}{G}{g}$\hspace{-1pt}
	is an abelian \hspace{-0.3pt}$2$\nobreakdash-group.
\end{lemma}
\begin{proof}
	Every automorphism of  $\CFIgraph{2^\pot}{G}{g}$ is origin-respecting, i.e.,
	it maps a vertex to a vertex of the same origin,
	because the preorder~$\spleq$ on the vertices is obtained from the total order~$\leq$ on~$G$.
	It follows that $\autgrp{\CFIgraph{2^\pot}{G}{g}}$
	is a subgroup of the direct product of the automorphism groups of every gadget.
	Because the automorphism group of each degree $d$ gadget
	is the abelian $2$-group $\setcond{\tup{a} \in \ZZ_{2^q}^d}{\sum\tup{a} = 0}$ as argued before,
	so is the direct product of them and in particular $\autgrp{\CFIgraph{2^\pot}{G}{g}}$.
\end{proof}

\paragraph{Other CFI Constructions.}
We compare our CFI construction to other version in the literature.
The classical construction in \cite{CaiFI1992} uses inner and outer vertices,
while we only use inner ones.
The sets $\StructV_{\bvertA,\bvertB,c}$ in our construction
correspond to the outer vertices in the classical construction.
We only use one type of vertices
to avoid case distinctions between inner and outer vertices in the following.

Another approach to avoid this case distinction is to only use outer vertices
and to replace the inner vertices by relations of higher arity (see, e.g., \cite{Hella96}).
The arity of the relations corresponds to the degree of a vertex in the base graph.
While this construction is more elegant,
it is restricted to base graphs of bounded degree to obtain structures of a fixed signature.
However, our argument separating rank logic and \CPT{}
requires base graphs of unbounded degree.
Our construction always yields structures of arity~$4$,
but the number of relations varies with the group~$\ZZ_{2^\pot}$.
Of course, we could use a single relation to encode the relations~$\rel_{E,c}$.
But in fact,
it suffices only to use~$\rel_{E,0}$
to obtain a structure with the same automorphism group.
Then all~$\rel_{E,c}$ are actually definable in $3$-variable counting logic.
We include all~$\rel_{E,c}$ in the structure for convenience.

In general, most properties of the structures transfer
between the different constructions (with some quite obvious adaptations).

\subsection{Isomorphisms of CFI Structures}
\label{sec:cfi-isomorphisms}
In this section we consider two classes of isomorphisms between
CFI structures.
They get important later in Section~\ref{sec:higher-arities}.
Let $q \in \nat$ and $G = (V,E,\leq)$ be a base graph.
In the following,
we denote for every $f \colon E \to \ZZ_{2^q}$
by~$\StructA_f$ the CFI structure $\CFIgraph{\ZZ_{2^\pot}}{G}{f}$.
These structures have by definition the same universe~$\StructVA$ for every 
$f \colon E \to \ZZ_{2^q}$.

\begin{definition}[Twisted Edge]
	Two functions  $f,g \colon E \to \ZZ_{2^\pot} $ \defining{twist} an edge $e \in E$ if $f(e) \neq g(e)$.
	We also say that~$e$ is
	\defining{twisted} by~$f$ and~$g$.
	For a set $W \subseteq V$ we say
	that~$f$ and~$g$ \emph{do not twist}~$W$
	if no edge in $G[W]$ is twisted by~$f$ and~$g$.
\end{definition}
We omit~$f$ and~$g$ if they are clear from the context.
Let $\bvertA \in V$
and $\tup{a} \in \ZZ_{2^\pot}^{\neighbors{G}{\bvertA}}$
satisfy $\sum \tup{a} = 0$.
We identify~$\tup{a}$ with a permutation of vertices with origin~$\bvertA$
as follows:
if~$\vertA$ has origin~$\bvertA$ (in some CFI structure over~$G$),
then $\gadgetiso{\vertA}{\tup{a}} := \vertB$ such that $\vertB(\bvertB) = \vertA(\bvertB) + \tup{a}(\bvertB)$ for all $\bvertB \in \neighbors{G}{\bvertA}$.
Because $\sum \tup{a} = 0$, $\gadgetiso{\vertA}{\tup{a}}$ is indeed
a vertex with origin $\bvertA$, too.

\begin{definition}[Path Isomorphism]
	\label{def:pathiso}
Let $c \in \ZZ_{2^\pot}$ and $\spath = (\bvertA_1, \dots, \bvertA_n)$ be a simple path in $G$.
For every $1 < i < n$, let $\tup{a}_i \in \ZZ_{2^\pot}^{\neighbors{G}{\bvertA_i}}$ 
such that $\tup{a}_i(\bvertA_{i-1}) = c$, $\tup{a}_i(\bvertA_{i+1}) = -c$,
and $\tup{a}_i(\bvertB) = 0$
for all other $\bvertB \in \neighbors{G}{\bvertA_i}$.
The \defining{path isomorphism} $\pathiso{c}{\spath}$ is defined by
\[
\pathiso{c}{\spath}(\vertA) := \begin{cases}
	\gadgetiso{\vertA}{\tup{a}_i} & \text{if } \orig{\vertA} = \bvertA_i \text{ and } 1 < i < n\\
	\vertA & \text{otherwise.}
\end{cases}
\]
\end{definition}

\begin{lemma}
	\label{lem:pathiso}
	Let $f,g\colon E \to \ZZ_{2^\pot}$,
	$\spath = (\bvertA_1, \dots, \bvertA_n)$ be a simple path in~$G$, 
	$e_1 = \set{\bvertA_1,\bvertA_2}$, and $e_2 = \set{\bvertA_{n-1},\bvertA_{n}}$.
	If no edge apart from~$e_1$ and~$e_2$ is twisted by~$f$ and~$g$,
	$g(e_1) = f(e_1) + c$,
	and $g(e_2) = f(e_2) - c$,
	then $\pathiso{c}{\spath}$ is an isomorphism $(\StructA_f,\pTupA) \to (\StructA_g,\pTupA)$
	for every tuple $\pTupA \in \StructVA^m$ satisfying $\distance{G}{\orig{\pTupA}}{\set{\bvertA_1, \dots, \bvertA_{n}}} > 1$.
\end{lemma}
The proof of Lemma~\ref{lem:pathiso}
is an obvious adaptation of the proof of Lemma~3.11 in~\cite{GradelPakusa19}.
This lemma uses a variant of CFI structures with outer vertices and relations,
but the arguments are similar.
We additionally require that the tuple~$\pTupA$ is fixed,
but because its distance to the path~$\spath$ is greater than~$1$,
it is not affected by the path isomorphism at all,
i.e., $\pathiso{c}{\spath}(\pTupA) = \pTupA$.
Isomorphisms between CFI structures satisfying $\sum f = \sum g$,
in which more than two edges are twisted,
can be composed out of multiple path isomorphisms.
The following special case of such isomorphisms will play an important role later:

\begin{definition}[Star Isomorphism]
	\label{def:stariso}
Let $\bvertC \in V$ be of degree $d$, $\ell \leq d$,
$\spath_1, \dots, \spath_\ell$ be simple paths,
$\spath_i = (\bvertA^i_1,\dots,\bvertA^i_{\ell_i})$,
$\bvertA^i_{\ell_i}=\bvertC$ for all $i \in [\ell]$,
and the~$\spath_i$ be disjoint apart from~$\bvertC$.
We call $\spath_1, \dots, \spath_\ell$ a \defining{star}
and $\bvertC$ the \defining{center of the star}.
For $\tup{c} \in \ZZ_{2^\pot}^\ell$ satisfying $\sum \tup{c} = 0$,
we define the \defining{star-isomorphism} $\stariso{\tup{c}}{\spath_1, \dots, \spath_\ell}$ via
\[
\stariso{\tup{c}}{\spath_1, \dots, \spath_\ell}(\vertA) := \begin{cases}
	\gadgetiso{\vertA}{\tup{c}'} &\text{if } \orig{\vertA} = \bvertC,\\
	\pathiso{c_i}{\spath_i}(\vertA) & \text{if } \orig{\vertA} \neq \bvertC \text{ and } \orig{\vertA} \text{ is contained in } \spath_i,\\
	u & \text{otherwise,}
\end{cases}
\]
where $\tup{c}' \in \ZZ_{2^\pot}^{\neighbors{G}{\bvertC}}$ such that
$\tup{c}'(\bvertA^i_{\ell_i -1}) = c_i$ for all $i \in [\ell]$
and $\tup{c}'(\bvertB) = 0$ for all other $\bvertB \in \neighbors{G}{\bvertC}$.
\end{definition}

\begin{lemma}
	\label{lem:stariso}
	Let $f,g\colon E \to \ZZ_{2^\pot}$,
	$\spath_1, \dots, \spath_\ell$ be a star in~$G$,
	$\spath_i = (\bvertA^i_1,\dots,\bvertA^i_{\ell_i})$ for all $i \in [\ell]$,
	and $\tup{c} \in \ZZ_{2^\pot}^\ell$ such that $\sum \tup{c} = 0$.
	If no edge apart from the edges $e_i = \set{x^i_1,x^i_2}$ for every $i \in [\ell]$ is twisted by~$f$ and~$g$ and
	$g(e_i) = f(e_i) + c_i$ for all $i \in [\ell]$,
	then
	$\stariso{\tup{c}}{\spath_1, \dots, \spath_\ell}$ is an isomorphism $(\Struct_f, \pTupA) \to (\Struct_g, \pTupA)$
	for every tuple $\pTupA\in \StructVA^m$ satisfying $\distance{G}{\orig{\pTupA}}{\setcond{\bvertA^i_j}{i \in [\ell],j \in [\ell_i]}} > 1$.
\end{lemma}
\begin{proof}
	Let $\pTupA\in \StructVA^m$ satisfy $\distance{G}{\orig{\pTupA}}{\setcond{\bvertA^i_j}{i \in [\ell],j \in [\ell_i]}} > 1$
	and let $\bvertC$ be the center of the star $\spath_1, \dots, \spath_\ell$.
	For every $i \in [\ell-1]$,
	let~$\spath_i'$ be 
	the $\bvertA^i_1-\bvertA^{i+1}_1$-path
	obtained by stitching~$\spath_i$ and~$\spath_{i+1}$
	together at $\bvertC := \bvertA^i_{\ell_i}$
	(that is, the path~$\spath_{i+1}$ is attached in reversed direction).
	Furthermore, for every $i \in [\ell-1]$,
	set $\autoA_i := \pathiso{\sum _{j \in [i]} c_j}{\spath_{i}'}$,
	and let $f_i \colon E \to \ZZ_{2^\pot}$
	be the function defined via $f_i(e_j) = f(e_j) + c_j$ for every $j \in [i]$,
	$f_i(e_{i+1}) = f(e_{i+1}) - \sum_{j \in [i]} c_j$, 
	and $f_i(e) = f(e)$ otherwise.
	Applying Lemma~\ref{lem:pathiso} inductively shows that
	$\autoA_1 \circ \cdots \circ \autoA_i$ is an isomorphism
	$(\StructA_f, \pTupA) \to (\StructA_{f_i}, \pTupA)$:
	For $i = 1$,
	the only twisted edges are~$e_1$ and~$e_2$ satisfying
	$f_1(e_1) = f(e_1)+c_1$
	and $f_1(e_2) = f(e_2) - c_1$
	and $\autoA_1 \colon (\Struct_f, \pTupA) \to (\StructA_{f_1}, \pTupA)$ is an isomorphism  by Lemma~\ref{lem:pathiso}.
	For every $2 \leq i  \leq \ell-2$,
	exactly the edges~$e_{i+1}$ and~$e_{i+2}$ are twisted by~$f_i$ and~$f_{i+1}$.
	It holds that
	\begin{align*}
		f_{i+1}(e_{i+1}) &= f(e_{i+1}) + c_{i+1} =  f_i(e_{i+1}) + \sum_{j \in [i+1]} c_j,\\
		f_{i+1}(e_{i+2}) &= f(e_{i+2}) - \sum_{j \in [i+1]} c_i = f_i(e_{i+2}) - \sum_{j \in [i+1]} c_i.
	\end{align*}
	Thus,~$\autoA_{i+1}$ is an isomorphism $(\StructA_{f_i}, \pTupA) \to (\StructA_{f_{i+1}}, \pTupA)$ by Lemma~\ref{lem:pathiso}
	and $\autoA_1 \circ \cdots \circ \autoA_{i+1}$ is an isomorphism
	$(\StructA_{f}, \pTupA) \to (\StructA_{f_{i+1}}, \pTupA)$ by induction.

	Now let $\autoB =\autoA_1 \circ \cdots \circ \autoA_{\ell}$.
	To prove the claim it suffices to show that $f_{\ell-1} = g$
	and that $\autoB = \stariso{\tup{c}}{\spath_1, \dots, \spath_\ell}$.
	The former holds because $\sum \tup{c} = 0$.
	To show the latter,
	first consider the path~$\spath_1$.
	On vertices with origin in~$\spath_1$ different from~$\bvertC$
	the action of~$\autoB$ is equal to the action of~$\autoA_1$.
	This exactly equals the definition of $\stariso{\tup{c}}{\spath_1, \dots, \spath_\ell}$.
	For vertices with origin in~$\spath_\ell$ different from~$\bvertC$
	the argument is similar
	and the action of~$\autoB$ is equal to the action of~$\autoA_{\ell-1}$.
	The isomorphism~$\autoA_{\ell-1}$ twists the edge $\set{\bvertA^\ell_1, \bvertA^\ell_2}$ by $-\sum_{j \in [\ell-1]} c_j$,
	which by assumption is equal to~$c_\ell$ because $\sum \tup{c} = 0$.
	Now consider vertices with origin in~$\spath_i$ different from~$\bvertC$ for~$i \notin \set{1, \ell}$.
	Here the action of~$\autoB$ equals the action of $\autoA_{i-1}\circ\auto_{i}$.
	The isomorphism~$\autoA_{i-1}$ twists the edge $\set{\bvertA^i_1, \bvertA^i_2}$
	by $-\sum_{j \in [i-1]} c_j$ and~%
	$\auto_{i}$ twists the same edge by $\sum_{j \in [i]} c_j$.
	Note that~$\spath_i'$ contains the vertices of~$\spath_{i+1}$ in reversed order, so on all the vertices with origin different from $\bvertC$
	the action of $\autoA_{i-1}\circ\auto_{i}$ becomes equal
	to the action of the path isomorphism $\pathiso{c_i}{\spath_i}$.
	Finally, by a similar argument,
	the action of~$\autoB$ on vertices with origin~$\bvertC$ equals the action of~$\tup{c}'$ defined as in Definition~\ref{def:stariso}.
\end{proof}

\subsection{Orbits of CFI Structures}
\label{sec:cfi-orbits}

In this section we analyze the structure of $k$-orbits of CFI structures
for highly connected base graphs.
Let $\pot, k, m \in \nat$ and
$G=(V,E,\leq)$ be a $(k + m + 1)$-connected base graph.
We denote again for every $f \colon E \to \ZZ_{2^q}$
by~$\StructA_f$ the CFI structure $\CFIgraph{\ZZ_{2^\pot}}{G}{f}$
with universe~$\StructVA$.
Let $\pTupA \in \StructVA ^ m$ be arbitrary but fixed.
We consider the $k$-orbits of 
pebbled structures $(\StructA_f, \pTupA)$,
i.e.,~orbits of $k$\nobreakdash-tuples.
Recall that $\autgrp{(\StructA_f,\pTupA)}$ is the automorphism group of $(\StructA_f,\pTupA)$
and that $\orbs{k}{(\StructA_f, \pTupA)}$
is the set of all $k$-orbits (cf.~Section~\ref{sec:preliminaries}).

\begin{definition}[Type of a Tuple]
	\label{def:tuple-type}
	The \defining{isomorphism type} of a pebbled structure
	is the class of all isomorphic structures.
	For $f \colon E \to \ZZ_{2^\pot}$
	the \defining{type} of a tuple $\tupA \in \StructVA^k$ in $(\StructA_f,\pTupA)$
	is the pair $(\orig{\tupA}, T)$,
	where $T$ is the isomorphism type of 
	$(\StructA_f[\orig{\pTupA\tupA}], \pTupA\tupA)$.
\end{definition}
We omit the pebbled structure $(\StructA_f,\pTupA)$ 
if it is clear from the context.
Including $\orig{\tupA}$ in the type is needed
because the isomorphism type~$T$
respects the relative order of
the gadgets in~$\spleq$ only.
If~$\StructA_f$ was vertex-colored instead, this would not be a problem.
We have to consider the origin induced substructure of $\orig{\pTupA\tupA}$
and not of~$\pTupA\tupA$ because only then 
the relations~$I_{\bvertA, \vertB}$ and~$C_{\bvertA, \vertB}$
can be recovered from~$\rel_I$ and~$\rel_C$.
Here, an edge coloring would resolve this issue.

\begin{lemma}
	\label{lem:cfi-same-type-automorphism}
	For every $f \colon E \to \ZZ_{2^\pot}$
	and every $\tupA,\tupB \in \StructVA^k$
	there is an automorphism $\autoA \in \autgrp{(\StructA_f,\pTupA)}$
	such that $\autoA(\tupA) = \tupB$
	if and only if~$\tupA$ and~$\tupB$ have the same type.
\end{lemma}
\begin{proof}
	A similar argument to the following can be found in~Lemma~3.15 in~\cite{GradelPakusa19}.
	Let $f \colon E \to \ZZ_{2^\pot}$ and $\tupA,\tupB \in \StructVA^k$.
	If $\autoA(\tupA) = \tupB$ for some automorphism $\autoA \in \autgrp{(\StructA_f,\pTupA)}$,
	then surely~$\tupA$ and~$\tupB$ have the same type.
	
	For the other direction, assume that~$\tupA$ and~$\tupB$ have the same type.
	Then, by definition, there is an isomorphism 
	$\autoA\colon(\StructA_f[\orig{\pTupA\tupA}], \pTupA\tupA) \to (\StructA_f[\orig{\pTupA\tupB}], \pTupA\tupB)$.
	Because~$\tupA$ and~$\tupB$ have the same type,
	it follows that $\orig{\pTupA\tupA}=\orig{\pTupA\tupB}$
	and in particular
	that~$\autoA$ is an automorphism of $(\StructA_f[\orig{\pTupA\tupA}],\pTupA)$.
	We show that this local automorphism 
	extends to an automorphism of $(\StructA_f, \pTupA)$.
	
	We extend~$\autoA$ by the identity map on all vertices
	with origin not in $\orig{\pTupA\tupA}$.
	Then~$\autoA$ is an isomorphism between~$(\StructA_f,\pTupA)$
	and another CFI structure,
	where all twisted edges $e_1, \dots, e_\ell$
	leave $\orig{\tupA}$ and are not incident to $\orig{\pTupA}$
	(edges incident to $\orig{\pTupA}$ cannot be twisted because~$\autoA$ fixes~$\pTupA$).
	Let~$N$ be the neighborhood of $\orig{\tupA}$
	(and thus of $\orig{\tupB}$).
	Because~$G$ is $(k+m+1)$-connected,
	there is an $\bvertA$-$\bvertB$-path
	not using $\orig{\pTupA\tupA}$
	for every $\bvertA,\bvertB \in N$
	because $G \setminus \orig{\pTupA\tupA}$ is still connected
	when removing at most $|\pTupA\tupA| = k + m < k + m + 1$ many vertices.
	Hence, we can use path isomorphisms
	to move the twists at every~$e_i$ all to~$e_1$.
	But because~$\autoA$ was an automorphism of $(\StructA_f[\orig{\pTupA\tupA}],\pTupA)$,
	the sum of the twists is~$0$.
	Hence, composing~$\autoA$ and the mentioned path isomorphisms
	forms an automorphism $\autoB \in \autgrp{(\StructA_f,\pTupA)}$.
	Because the selected paths do not use $\orig{\pTupA\tupA}$,
	we still have $\autoB(\pTupA\tupA) = \pTupA\tupB$.
\end{proof}
\begin{corollary}
	\label{cor:orbit-types-unique}
	For every $f \colon E \to \ZZ_{2^\pot}$
	and every $\BlockA \in \orbs{k}{(\StructA_f, \pTupA)}$
	there is a type such that~$\BlockA$ contains exactly the tuples of that type.
\end{corollary}
\begin{definition}[Type of an Orbit]
	For $f \colon E \to \ZZ_{2^\pot}$ the \defining{type} of a $k$-orbit in $(\StructA_f,\pTupA)$
	is the type of its contained tuples.
\end{definition}
\begin{corollary}
	\label{cor:cfi-same-autgroup-and-orbits}
	For every pair $f,g \colon E \to \ZZ_{2^\pot}$
	that does not twist $\orig{\pTupA}$,
	it holds that
	\begin{align*}
		\orbs{k}{(\StructA_f,\pTupA)} &= \orbs{k}{(\StructA_g,\pTupA)} \text{ and}\\
		\autgrp{(\StructA_f,\pTupA)} &= 
		\autgrp{(\StructA_g,\pTupA)}.
	\end{align*}
\end{corollary}
While the orbit partitions of $(\StructA_f,\pTupA)$
and $(\StructA_g,\pTupA)$ are equal,
it is in general not true that 
an orbit $\BlockA \in \orbs{k}{(\StructA_f,\pTupA)}$
has the same type in 
$(\StructA_f,\pTupA)$
and in $(\StructA_g,\pTupA)$.

\begin{lemma}
	\label{lem:cfi-occuring-orbit-types-equal}
	Suppose the functions $f,g \colon E \to \ZZ_{2^\pot}$
	do not twist $\orig{\pTupA}$.
	Then for every $k$-orbit
	$\BlockA \in \orbs{k}{(\StructA_f,\pTupA)}$
	there is a
	$\BlockB \in \orbs{k}{(\StructA_g,\pTupA)}$
	that has the same type.
\end{lemma}
\begin{proof}
	It suffices to consider the case that exactly one edge $e=\set{\bvertA,\bvertB}$ is twisted
	because isomorphisms preserve types
	and because no edge contained in $\orig{\pTupA}$ is twisted,
	all twists can be moved to a single edge using isomorphisms.
	
	Let $\BlockA \in \orbs{k}{(\StructA_f,\pTupA)}$.
	If $\set{\bvertA,\bvertB} \not\subseteq \orig{\BlockA}$,
	then $\BlockA$ has the same type in $(\StructA_f,\pTupA)$
	and in $(\StructA_g,\pTupA)$.
	Otherwise, let $\set{\bvertA,\bvertB} \subseteq \orig{\BlockA}$
	and assume w.l.o.g.~that $\bvertB \notin \orig{\pTupA}$
	(if $\set{\bvertA,\bvertB} \subseteq \orig{\pTupA}$,
	then $\set{\bvertA,\bvertB} \not\subseteq \orig{\BlockA}$
	because the twisted edge is not contained in $\orig{\pTupA}$).
	Furthermore, choose a path $\spath = (\bvertA,\bvertB,\dots,\bvertC)$,
	such that $\bvertC \notin \orig{\BlockA}$ and the path, possibly apart from $\bvertA$,
	is disjoint from $\orig{\pTupA}$.
	Such a path exists,
	because  $G \setminus \orig{\pTupA} \setminus \orig{\BlockA}$ is connected 
	(at most $m + k < m + k + 1$ many vertices are removed)
	and $\bvertB \notin \orig{\pTupA}$ by assumption.
	So we can pick some vertex~$\bvertC$ not contained in $\orig{\BlockA}$
	and in $\orig{\pTupA}$.
	Now, we move the twist to an edge incident to~$\bvertC$
	with the path isomorphism $\autoA := \pathiso{g(e)-f(e)}{\spath}$.
	Then~$\BlockA$ has the same type in $(\StructA_f,\pTupA)$
	as in $\autoA((\StructA_g,\pTupA)) = (\autoA(\StructA_g),\pTupA)$
	because $\StructA_f[\orig{\pTupA}\cup\orig{\BlockA}] = \autoA(\StructA_g)[\orig{\pTupA}\cup\orig{\BlockA}]$.
	Because isomorphisms preserve types,
	there is an orbit $\BlockB \in \orbs{k}{(\StructA_g,\pTupA)}$
	with the same type in $(\StructA_g,\pTupA)$ as~$\BlockA$ has in $(\StructA_f,\pTupA)$.
\end{proof}

\begin{lemma}
	\label{lem:cfi-orbit-auto-group-regular}
	Let $f \colon E \to \ZZ_{2^\pot}$ and
	$\BlockA \in \orbs{k}{(\StructA_f, \pTupA)}$.
	Then the permutation group~$\group$ on~$\BlockA$
	induced by $\autgrp{(\StructA_f, \pTupA)}$
	is a regular and abelian $2$-group.
\end{lemma}
\begin{proof}
	We first argue that the automorphism group of a gadget is a regular abelian $2$\nobreakdash-group.
	Recall that the vertices of a gadget for the vertex $\tup{\bvertA} \in V$
	are defined as ${\StructV_\bvertA = \setcond{a \in \ZZ_{2^\pot}^{\neighbors{G}{\bvertA}}}{\sum a = 0}}$.
	So $|\StructV_\bvertA| = (2^\pot)^{d-1}$,
	where~$d$ is the degree of~$\bvertA$.
	We saw in Section~\ref{sec:cfi-isomorphisms}
	that the automorphism group of a gadget is transitive.
	We already argued that the automorphism group is isomorphic to
	$\setcond{\tup{a} \in \ZZ_{2^\pot}^d}{\sum{\tup{a}} = 0}$.
	Thus, the automorphism group is a $2$-group and has order $(2^\pot)^{d-1}$.
	Hence, it is a regular abelian $2$-group.
	
	The claim for $k$-orbits
	follows from the case of a gadget:~%
	$\group$ is a subgroup of the direct product of the automorphism groups
	of the gadgets of $\orig{\BlockA}$.
	That is,~$\group$ is an abelian $2$-group.
	By definition of a $k$-orbit,~$\group$ is transitive.
	For regularity, note that a gadget is partitioned into singleton orbits
	once one vertex of the gadget is fixed (cf.~Lemma~3.13 in~\cite{GradelPakusa19}).
	So if we fix a $\tupA \in \BlockA$,
	all gadgets in the origin of~$\tupA$ are fixed.
	So, if an automorphism~$\autoA$ maps~$\tupA$ to~$\tupB$,
	then its action on~$\BlockA$ is fixed,
	i.e., there is exactly one permutation in~$\group$
	that
	maps~$\tupA$ to~$\tupB$.
	Hence, $|\group| = |\BlockA|$ and~$\group$ is regular.
\end{proof}

\subsection{Composition of Orbits}
Composing $k$-orbits out of $k'$-orbits for $k' < k$ plays a special role later.
We further analyze the structure of $k$-orbits and identify cases
in which such a composition in possible.
As in the previous section,
let $\pot, k, m \in \nat$ and
$G=(V,E,\leq)$ be a $(k + m + 1)$-connected base graph,
denote for $f \colon E \to \ZZ_{2^q}$
by $\StructA_f$ the CFI structure $\CFIgraph{\ZZ_{2^\pot}}{G}{f}$
with universe~$\StructV$,
and let $\pTupA \in \StructVA ^ m$.

Let $\tupA \in \StructVA^k$ and
$\orig{\tupA}$ (viewed as a set)
be partitioned into~$M$ and~$N$.
We now introduce notation for splitting~$\tupA$
into its parts belonging to~$N$ and~$M$
and for recovering $\tupA$ from these two parts again.
\begin{enumerate}
\item The tuple $\tupA_{N}$ obtained from~$\tupA$
by deleting all entries whose origin is not in~$N$ (respectively for~$M$), is \[\tupA_{N} := \tupA_{\setcond{i \in [k]} {\orig{\vertA_i} \in N}}.\]
\item We define  a concatenation operation for a permutation~$\perm$ of $[k]$ as follows:
\[\tupA_{N} \cdot_\perm \tupA_{M} := \perm(\tupA_{N} \tupA_{M}).\]
For a suitable~$\perm$ we have
$\tupA = \tupA_{N} \cdot_\perm \tupA_{M}$.
In this article we are only interested in permutations
satisfying the former equation.
Then~$\perm$ is almost always fixed by the context and we use juxtaposition
$\tupA_{N}\tupA_{M}$.
\textit{It is never the case that we refer with $\tupA_{N}\tupA_{M}$
to ordinary concatenation.}
\item We define similar operations for orbits:
For $\BlockA \in \orbs{k}{(\StructA,\pTupA)}$
we set
\begin{align*}
	\restrictVect{\BlockA}{N} &:=\setcond[\big]{\tupA_{N}}{\tupA\in \BlockA},\\
	\restrictVect{\BlockA}{N} \times_\perm \restrictVect{\BlockA}{M} &:= \setcond[\big]{\tupA_N \cdot_\perm \tupA_M}{\tupA_N \in \restrictVect{\BlockA}{N}, \tupA_M \in \restrictVect{\BlockA}{M}},
\end{align*}
and leave out~$\perm$ if clear from the context.
This intentionally overloads notation. Because the tuples in~$\BlockA$
	are indexed by $[k]$,
	$\restrictVect{\BlockA}{N}$ and $\restrictVect{\BlockA}{K}$
	for  $N \subseteq \orig{\BlockA}\subseteq V$ and $K \subseteq [k]$
	can always be distinguished.
\end{enumerate}
We also use this notation 
if~$N$ and~$M$ are sets of sets,
such that $\orig{\tupA}$ is partitioned into $\bigcup N$ and $\bigcup{M}$.

\begin{definition}[Components of Tuples and Orbits]
	\label{def:component}
	Let $f\colon E \to \ZZ_{2^\pot}$,
	$\tupA \in \StructVA ^k$,
	and $N \subseteq \orig{\tupA}$.
	We call~$N$ a \defining{component} of~$\tupA$
	if~$N$ is a connected component of $G[\orig{\tupA}]$.
	We call~$\tupA$ \defining{disconnected}
	if it has more than one component.
	
	Likewise, a $k$-orbit $\BlockA \in \orbs{k}{(\StructA_f, \pTupA)}$ is disconnected
	if~$\BlockA$ contains some (and thus only) disconnected tuples.
	A set $N \subseteq \orig{\BlockA}$  is a component of~$\BlockA$
	if~$N$ is a connected component of $G[\orig{\BlockA}]$.
\end{definition}
If a $k$-orbit~$\BlockA$ is disconnected,
then we can split~$\BlockA$ into multiple $k'$-orbits for $k' < k$ as follows.
\begin{lemma}
	\label{lem:split-discon-orbits}
	Let $f\colon E \to \ZZ_{2^\pot}$, $\BlockA \in \orbs{k}{(\StructA_f, \pTupA)}$,
	and the components of~$\BlockA$ be partitioned into~$M$ and~$N$.
	Then $\BlockA =  \restrictVect{\BlockA}{M} \times \restrictVect{\BlockA}{N}$,
	$\restrictVect{\BlockA}{M} \in \orbs{k_M}{(\StructA_f, \pTupA)}$,
	and $\restrictVect{\BlockA}{N} \in \orbs{k_N}{(\StructA_f, \pTupA)}$
	for suitable $k_M$ and $k_N$ such that $k_M + k_N = k$.
\end{lemma}
\begin{proof}
	This can easily be seen by Corollary~\ref{cor:orbit-types-unique}.
	Because~$M$ and~$N$ are sets of components,
	the type of $\tupA \in \BlockA$
	is given by the disjoint union of the types of~$\tupA_M$
	and~$\tupA_N$ (even if $\orig{\pTupA}$ overlaps with~$M$ and~$N$ because~$\pTupA$ has to be fixed by every automorphism).
\end{proof}

Next, we show how to obtain $k'$-orbits from $k$-orbits with $k'<k$ 
by fixing a vertex.
\begin{lemma}
	\label{lem:k-orbit-fix-vertex}
	Let $f \colon E\to\ZZ_{2^\pot}$,
	$\BlockA \in \orbs{k}{(\StructA_f, \pTupA)}$,
	$K \subseteq [k]$,
	and $\orig{\restrictVect{\BlockA}{K}} = \set{\bvertC}$.
	For every $\tupB \in \StructVA^{|K|}$ and $\vertC \in \StructVA$
	such that $\orig{\tupB} = \set{\bvertC}$ and $\orig{\vertC} = \bvertC$,
	the set
	\begin{align*}
		\BlockB &:=\setcond*{\restrictVect{\tupA}{[k]\setminus K}}{\tupA \in \BlockA, \restrictVect{\tupA}{K} =\tupB } \text{ satisfies }\\
		\BlockB &\in \orbs{k-|K|}{(\StructA_f,\pTupA\vertC)} \cup  \set{\emptyset}.
	\end{align*}
	If~$\tupB$ has the same type as  $\restrictVect{\tupA}{K}$ for some (and thus every) $\tupA \in \BlockA$,
	then $\BlockB \neq \emptyset$.
\end{lemma}
\begin{proof}
	We assume w.l.o.g.~up to reordering that $K = [|K|]$.
	Let $\tupB\in \StructVA^{|K|}$ such that $\orig{\tupB} = \set{\bvertC}$.
	Every vertex~$\vertB_i$ forms a singleton orbit in $\orbs{1}{(\StructA_f, \vertC)}$
	and in particular in $\orbs{1}{(\StructA_f, \pTupA\vertC)}$
	because~$\vertB_i$ and~$\vertC$ have the same origin~$\bvertC$
	(all vertices with origin~$\bvertC$ can be distinguished by their distances to~$\vertC$ in the~$C_{\vertA,\vertB}$ relation, cf.~Lemma~3.13 in~\cite{GradelPakusa19}).
	So it holds that
	$\autgrp{(\StructA_f,\pTupA\tupB)}= \autgrp{(\StructA_f,\pTupA\vertC)}$.
	Assume that $\BlockB \neq \emptyset$.
	Because~$\BlockA$ is an orbit,
	if $\tupB\tupA, \tupB\tupA' \in \BlockA$,
	then there is an automorphism $\autoA \in \autgrp{(\StructA_f, \pTupA)}$
	such that $\autoA(\tupB\tupA) = \tupB\tupA'$.
	That is, $\autoA \in \autgrp{(\StructA_f, \pTupA\tupB)} = \autgrp{(\StructA_f, \pTupA\vertC)}$ and thus~$\BlockB$ is a subset of an orbit in $\orbs{k-|K|}{(\StructA_f,\pTupA\vertC)}$.
	To show that~$\BlockB$ is indeed an orbit,
	assume that $\tupA \in \BlockB$ and $\autoA \in \autgrp{(\StructA_f,\pTupA\vertC)} = \autgrp{(\StructA_f,\pTupA\tupB)}$.
	Because $\tupA \in \BlockB$, $\tupB\tupA \in \BlockA$
	and $\autoA(\tupB\tupA) = \tupB\autoA(\tupA)  \in \BlockA$.
	Hence, $\autoA(\tupA) \in \BlockB$
	and so $\BlockB \in \orbs{k-|K|}{(\StructA_f,\pTupA\vertC)}$.
	
	Now assume that there is some $\tupA \in \BlockA$
	such that  $\restrictVect{\tupA}{K}$ has the same type as~$\tupB$.
	That is, there is an automorphism $\autoA \in \autgrp{(\StructA_f,\pTupA)}$
	such that $\restrictVect{\autoA(\tupA)}{K} = \tupB$ (Lemma~\ref{lem:cfi-same-type-automorphism}).
	Hence, $\autoA(\tupA) \in \BlockA$ 
	and $\restrictVect{\autoA(\tupA)}{[k]\setminus K} \in \BlockB$.
\end{proof}
Note that~$\BlockB$ is independent of~$\vertC$,
but not the type of~$\BlockB$ in $(\StructA_f,\pTupA\vertC)$.

\begin{corollary}
	\label{cor:k-orbit-fix-vertex}
	Let $f \colon E \to\ZZ_{2^\pot}$,
	$\BlockA \in \orbs{k}{(\StructA_f, \pTupA)}$,
	$i \in [k]$,
	$\orig{\restrictVect{\BlockA}{\set{i}}} = \set{\bvertC}$,
	and let $\distance{G}{\bvertC}{\orig{\pTupA}} > 1$.
	For every $\vertB, \vertC \in \StructVA$ 
	such that $\orig{\vertB} = \orig{\vertC} = \bvertC$
	it holds that $\setcond{\restrictVect{\tupA}{[k]\setminus \set{i}}}{\tupA \in \BlockA, \vertA_i =\vertB } \in \orbs{k-1}{(\StructA_f,\pTupA\vertC)}$.
\end{corollary}
\begin{proof}
	We apply Lemma~\ref{lem:k-orbit-fix-vertex}:
	Because $\distance{G}{\bvertC}{\orig{\pTupA}} > 1$,
	the type of~$\vertC$ is the same as the type of every~$\vertB$
	with origin~$\bvertC$,
	in particular the same as~$\vertB_i$ for every $\tupB \in \BlockA$.
\end{proof}

\subsection{Rank Logic on CFI Structures}
In this section we refine a result of~\cite{GradelPakusa19}
and show that on CFI structures over~$\ZZ_{2^\pot}$
the uniform rank logic \IFPR{} has the same expressiveness as 
the rank logic \IFPRP{\set{2}} only with rank operators over $\FF_2$.

\begin{definition}
	\label{def:cfi-graph-class}
	For a class of base graphs $\GraphClass$,
	\[\CFITwoClass{\GraphClass} := \setcond[\big]{\CFIgraph{2^\pot}{G}{f}}{\pot\in \nat, G=(V,E,\leq) \in \GraphClass, f \colon E \to \ZZ_{2^\pot}}\]
	is the class of all CFI structures over $\GraphClass$.
\end{definition}

\begin{lemma}
	\label{lem:IFPR-to-IFPR2}
	Let~$\GraphClass$ be a class of base graphs.
	For every \IFPR{}-formula~$\formA$
	there is an $\IFPRP{\set{2}}$-formula~$\formB$
	that is equivalent to~$\formA$ on $\CFITwoClass{\GraphClass}$.
\end{lemma}
\begin{proof}
	Let \defining{solvability logic} $\IFPS$ be the extension by $\IFPC$
	by the uniform solvability quantifier~$\slv$~\cite{GradelPakusa19}.
	If~$\termA(\uniVarVA, \uniVarVB)$ is a numeric term
	and~$\termB$ is a closed numeric term, then
	\[\slv (\uniVarVA, \uniVarVB).\qspace(\termA,\termB)\]
	is a formula.
	Similar to the rank operator~$\rk$, the numeric term~$\termB$ defines a number~$p$.
	If~$p$ is prime, then the solvability quantifier is satisfied if the linear system  $\definedMatrix{\Struct}{\termA}x=1$ is solvable over~$\FF_p$.
	If otherwise~$p$ is not prime, then the operator is not satisfied.
	Let $\IFPSP{\Omega}$ be the extension of $\IFPC$ with
	solvability quantifiers~$\slv_p$ for each fixed field~$\FF_p$ with $p \in \Omega$ similar.
	We again left out parameters for readability.
	
	Grädel and Pakusa~\cite{GradelPakusa19} give a translation of
	$\IFPRP{\Omega}$-formulas to $\IFPC$-formulas
	equivalent on CFI structures over~$\FF_2$ for every set of primes~$\Omega$ satisfying $2 \notin \Omega$.
	The crucial point in their proofs is that the automorphism groups of these CFI structures are abelian $2$-groups
	and that their $k$-orbits can be defined and ordered in \IFPC{},
	that is, there is an \IFPC{}-definable total preorder on all $k$\nobreakdash-tuples
	whose equivalence classes coincide with the $k$-orbits
	(their construction is not specific to~$\FF_2$ but generally for~$\FF_p$ whenever $p \notin \Omega$ and~$p$ is the characteristic of the CFI structures).
	These assumptions are made explicit in Section~3.2 in~\cite{GradelPakusa19}.
	Hence, the arguments work for CFI structures over $\ZZ_{2^\pot}$
	instead of~$\FF_2$, too.
	In~\cite{GradelPakusa19}, the authors use solvability logic as an intermediate step
	and first show that for all sets of primes~$\Omega$
	(even with $2 \in \Omega$) it holds that
	$\IFPRP{\Omega} = \IFPSP{\Omega}$ on $\CFITwoClass{\GraphClass}$ (Lemma~3.7 in~\cite{GradelPakusa19}).
	This reduction works as well for the uniform case and shows
	$\IFPR = \IFPS$ on $\CFITwoClass{\GraphClass}$.
	
	The second step in~\cite{GradelPakusa19} is a recursive translation
	of $\IFPRP{\Omega}$-formulas to $\IFPC$-formulas if $2 \notin \Omega$ (Lemmas~3.4 to~3.6 in~\cite{GradelPakusa19}).
	For every \IFPC{}-term~$\termA$ the solvability quantifier $\formB = \slv_p (\uniVarVA, \uniVarVB).\qspace\termA$
	over~$\FF_p$ can be simulated in \IFPC{}
	by computing the rank of the matrix $M := \definedMatrix{\Struct}{\termA}$ orbit-wise.
	This is expressible in \IFPC{} because the automorphism group is a $2$-group and $p \neq 2$. This process works as follows:
	There is an \IFPC{}-formula
	that for every prime~$p$ and every term~$\termA$ exploits the orbits of the structure
	to define a matrix~$E$
	such that $Mx = 1$ is solvable if and only if $(M \cdot E)x = 1$ is solvable
	(Lemma~3.6 in~\cite{GradelPakusa19}).
	Now,~$E$ is defined such that the columns of $M \cdot E$ are totally ordered
	and thus the solution can be obtained in $\IFPC$.
	
	Now, we translate an $\IFPSP{\Omega}$-formula (respectively term) with $2 \in \Omega$
	recursively into an $\IFPSP{\set{2}}$-formula (respectively term).
	Again consider a solvability quantifier $\formB = \slv_p (\uniVarVA, \uniVarVB).\qspace\termA$.
	If $p = 2$, then we recurse on~$\termA$ but do not replace the solvability quantifier.
	If otherwise $p \neq 2$,
	then we recurse on~$\termA$
	and obtain an $\IFPSP{\set{2}}$-term equivalent to~$\termA$,
	define the matrix~$E$ with the $\IFPC{}$-formula from above,
	and construct a formula defining whether 
	$M \cdot E = 1$ is solvable.
	Because this check can be done in $\IFPC$
	and~$M$ is defined by an $\IFPSP{\set{2}}$-term,
	we obtain an $\IFPSP{\set{2}}$-formula equivalent to~$\formB$.
	
	We finally deal with the case of an $\IFPS$ formula,
	where the prime is defined by a numeric term~$\termB$.
	Checking an ordered equation system for solvability
	is $\IFPC$-definable when the prime is given by a term, too.
	Let $\formB = \slv (\uniVarVA, \uniVarVB).~(\termA,\termB)$ be a uniform solvability quantifier.
	Let $\formB_2$ be the formula obtained for $\slv_2 (\uniVarVA, \uniVarVB).~\termA$ in the former case
	and $\formB_{\neq2}$ be the formula for the case $p \neq 2$,
	where we already use~$\termB$ to obtain the prime.
	Indeed, $\formB_{\neq2}$ is independent of~$p$
	because defining the matrix~$E$ is independent of~$p$ and
	checking the linear equation system for consistency is already done using the prime-defining term~$\termB$.
	Then the uniform solvability quantifier~$\formB$ is
	equivalent to the $\IFPSP{\set{2}}$-formula
	$(\termB = 2 \rightarrow \formB_2)  \land  (\termB \neq 2 \rightarrow \formB_{\neq2})$.
	Obviously, an $\IFPSP{\set{2}}$-formula can be translated
	back into an  $\IFPRP{\set{2}}$-formula.
\end{proof}

\section{Matrices over CFI Structures}
\label{sec:matrices-cfi}
In the invertible-map game,
Duplicator has to partition the $2k$\nobreakdash-tuples of CFI structures
and to provide a similarity matrix.
For our arguments, we want that Duplicator plays with the
$2k$-orbit partitions.
To construct the required similarity matrices,
we develop a criterion for invertibility
of matrices over~$\FF_2$
and show that this criterion is preserved by matrix multiplication.

Let $\pot, k, m \in \nat$ and 
$G=(V,E,\leq)$ be a $(k + m + 1)$-connected base graph.
The connectivity is needed to apply the lemmas of Section~\ref{sec:cfi-orbits}.
Again, we denote for a function $f \colon E \to \ZZ_{2^q}$
by~$\StructA_f$ the CFI structure $\CFIgraph{\ZZ_{2^\pot}}{G}{f}$
with universe~$\StructVA$ (which is equal for every $f \colon E \to \ZZ_{2^q}$).
Let $\pTupA \in \StructVA^m$ be arbitrary but fixed in this section.

\begin{definition}[Blurring the Twist]
	\label{def:blur-twist}
	For $f,g \colon E \to \ZZ_{2^\pot}$ not twisting $\orig{\pTupA}$,
	an $\StructVA^k \times \StructVA^k$ matrix~$S$ over~$\FF_2$
	\defining{$k$-blurs the twist} between $(\StructA_f,\pTupA)$ and $(\StructA_g,\pTupA)$ if~$S$ is invertible and
	$\charMat{\BlockA} \cdot S = S \cdot \charMat{\BlockB}$
	for every $\BlockA \in \orbs{2k}{(\StructA_f,\pTupA)}$ and
	$\BlockB \in \orbs{2k}{(\StructA_g,\pTupA)}$
	that are of the same type.
\end{definition}
Note that by Corollary~\ref{cor:orbit-types-unique}
two different orbits have different types
and that by Lemma~\ref{lem:cfi-occuring-orbit-types-equal}
for each $\BlockA \in \orbs{2k}{(\StructA_f,\pTupA)}$
there is a $\BlockB \in \orbs{2k}{(\StructA_g,\pTupA)}$
of the same type.
So we indeed get a bijection between the orbits
and Duplicator can use the matrix~$S$ in the invertible-map game.
Because~$S$ is invertible,
$\charMat{\BlockA} \cdot S = S \cdot \charMat{\BlockB}$
is equivalent to
$\charMat{\BlockA} = S \cdot \charMat{\BlockB} \cdot \inv{S}$.
Showing the former has the benefit that we do not need the inverse~$\inv{S}$.

\begin{lemma}\label{lem:blur-twist-mult}
	Let $f,g,h\colon E \to \ZZ_{2^\pot}$ pairwise not twist $\orig{\pTupA}$
	and $S,T$ be $\StructVA^k \times \StructVA^k$ matrices over~$\FF_2$.
	If~$S$ blurs the twist between
	$(\StructA_f, \pTupA)$ and  $(\StructA_g, \pTupA)$
	and~$T$ blurs the twist between
	$(\StructA_g, \pTupA)$ and  $(\StructA_h, \pTupA)$,
	then $S \cdot T$ blurs the twist
	between $(\StructA_f, \pTupA)$ and $(\StructA_h, \pTupA)$.
\end{lemma}
\begin{proof}
	Let $\BlockA \in \orbs{2k}{(\StructA_f,\pTupA)}$, $\BlockB \in \orbs{2k}{(\StructA_g,\pTupA)}$, and $\BlockC \in \orbs{2k}{(\StructA_h,\pTupA)}$ be of the same type.
	Recall that given~$\BlockA$, the orbits~$\BlockB$ and~$\BlockC$ are determined uniquely (Corollary~\ref{cor:orbit-types-unique}).
	Then $\charMat{\BlockA} \cdot S \cdot T = S \cdot \charMat{\BlockB} \cdot T = S \cdot T \cdot \charMat{\BlockC}$.
\end{proof}

Now we want to develop combinatorial conditions for an $\StructVA^k \times \StructVA^k$ matrix~$S$ over~$\FF_2$,
which guarantee that~$S$ is invertible.
The $k$-orbits (for given $f,g\colon E\to\ZZ_{2^\pot}$) partition~$S$ into a block matrix.
Each $\BlockA \in \orbs{k}{(\StructA_f,\pTupA)}$ corresponds to a subset of the rows of~$S$
and each $\BlockB \in \orbs{k}{(\StructA_g,\pTupA)}$ corresponds to a subset of the columns of~$S$.
We denote by $S_{\BlockA\times\BlockB}$ the corresponding submatrix of~$S$.

\begin{definition}[Orbit-Diagonal Matrix]
	For $f,g \colon E \to \FF_2$ not twisting $\orig{\pTupA}$,
	we call an $\StructVA^k \times \StructVA^k$ matrix $S$ over $\FF_2$ \defining{orbit-diagonal} over $(\StructA_f, \pTupA)$ and $(\StructA_g, \pTupA)$,
	if for every $\BlockA\in \orbs{k}{(\StructA_f,\pTupA)} $ and every
	$\BlockB\in\orbs{k}{(\StructA_g,\pTupA)}$
	it holds that if $S_{\BlockA\times\BlockB} \neq \zeromat$,
	then~$\BlockA$ has the same type in $(\StructA_f, \pTupA)$
	as~$\BlockB$ has in $(\StructA_g, \pTupA)$.
\end{definition}
We have seen that for every $\BlockA \in \orbs{k}{(\StructA_f,\pTupA)}$
there is exactly one $\BlockB \in \orbs{k}{(\StructA_g,\pTupA)}$
of the same type.
So orbit-diagonal matrices are block-diagonal matrices,
where orbits of the same type form the nonzero blocks.
A permutation~$\perm$ of~$\StructVA^k$
is applied to an $\StructVA^k \times \StructVA^k$ matrix~$S$ in the natural way:
$(\perm(S))(\tupA,\tupB) = S(\perm(\tupA), \perm(\tupB))$.
Of particular interest are automorphisms.
\begin{definition}[Orbit-Invariant Matrix]
	\label{def:orbit-invariant-matrix}
	For $f,g \colon E \to \ZZ_{2^\pot}$ that do not twist $\orig{\pTupA}$,
	an $\StructVA^k\times\StructVA^k$ matrix~$S$ over~$\FF_2$
	is called \defining{orbit-invariant}
	over $(\StructA_f, \pTupA)$ and $(\StructA_g, \pTupA)$,
	if for every $\BlockA \in \orbs{k}{(\StructA_f,\pTupA)}$, $\BlockB \in \orbs{k}{(\StructA_f,\pTupA)}$,
	and $\autoA \in \autgrp{(\StructA_f,\pTupA)} = \autgrp{(\StructA_g,\pTupA)}$ (cf.~Corollary~\ref{cor:cfi-same-autgroup-and-orbits})
	the matrix~$S$ satisfies
	$\autoA(S_{\BlockA \times \BlockB}) = S_{\BlockA \times \BlockB}$.
\end{definition}

\begin{lemma}
	\label{lem:orbit-invariant-mult}
	Let $f,g,h \colon E\to \ZZ_{2^\pot}$ not twist $\orig{\pTupA}$
	and $S,T$ be $\StructVA^k \times \StructVA^k$ matrices over~$\FF_2$.
	If~$S$ is orbit-diagonal and orbit-invariant over $(\StructA_f, \pTupA)$ and $(\StructA_g, \pTupA)$
	and~$T$ is orbit-diagonal and orbit-invariant over $(\StructA_g, \pTupA)$ and $(\StructA_h, \pTupA)$,
	then $S\cdot T$ is
	orbit-diagonal and orbit-invariant over $(\StructA_f, \pTupA)$ and $(\StructA_h, \pTupA)$.
\end{lemma}
\begin{proof}
	It is clear that $S \cdot T$ is orbit-diagonal over $(\StructA_f, \pTupA)$ and $(\StructA_h, \pTupA)$.
	For $k$-orbits ${\BlockA \in \orbs{k}{(\StructA_f,\pTupA)}}$,
	${\BlockB \in \orbs{k}{(\StructA_g,\pTupA)}}$, and
	${\BlockC \in \orbs{k}{(\StructA_h,\pTupA)}}$ of the same type it holds that
	$(S\cdot T)_{\BlockA \times \BlockC}(\tupA,\tupC) = \sum_{\tupB \in \BlockB} S_{\BlockA \times \BlockB} (\tupA,\tupB) \cdot T_{\BlockB\times \BlockC}(\tupB,\tupC)$.
	Let $\autoA \in \autgrp{(\StructA, \pTupA)}$.
	Then
	\begin{align*}
		(\autoA(S\cdot T))_{\BlockA \times \BlockC}(\tupA,\tupC)
		&= (S\cdot T)_{\BlockA \times \BlockC}(\autoA(\tupA),\autoA(\tupC))\\
		&= \sum_{\tupB \in \BlockB} S_{\BlockA \times \BlockB} (\autoA(\tupA),\tupB) \cdot T_{\BlockB\times \BlockC}(\tupB,\autoA(\tupC))\\
		&= \sum_{\tupB \in \BlockB} S_{\BlockA \times \BlockB} (\autoA(\tupA),\autoA(\tupB)) \cdot T_{\BlockB\times \BlockC}(\autoA(\tupB),\autoA(\tupC))\\
		&= \sum_{\tupB \in \BlockB} S_{\BlockA \times \BlockB} (\tupA,\tupB) \cdot T_{\BlockB\times \BlockC}(\tupB,\tupC)\ifmultiroweq{\\
		&}= (S\cdot T)_{\BlockA \times \BlockC}(\tupA,\tupC).
	\end{align*}
	Applying~$\autoA$ to~$\tupB$ is valid because~$\autoA$ is a permutation of~$\BlockB$ and thus only permutes the summands.
	Then $S_{\BlockA \times \BlockB} (\autoA(\tupA),\autoA(\tupB)) = S_{\BlockA \times \BlockB} (\tupA,\tupB)$ because~$S$ is orbit-invariant
	(and likewise for~$T$).
\end{proof}

\begin{definition}[Odd-Filled Matrix]
	A matrix over~$\FF_2$ is called \defining{odd-filled}
	if every row contains an odd number of ones.
\end{definition}

\begin{lemma}\label{lem:odd-filled-product}
	If two $\StructVA^k \times \StructVA^k$ matrices~$S$ and~$T$ over~$\FF_2$ 
	are odd-filled,
	then so is $S\cdot T$.
\end{lemma}
\begin{proof}
	Let $R = S\cdot T$ and
	denote by~$r_{\tupA}$ and~$t_{\tupB}$ the rows of~$R$ and~$T$
	indexed by $\tupA \in \StructVA^k$ and $\tupB \in \StructVA^k$.
	Then 
	\[r_{\tupA} = \sum_{\tupB \in \StructVA^k} S(\tupA,\tupB) \cdot t_{\tupB}.\]
	The number of ones modulo~$2$ is given by
	\[\sum r_{\tupA} = \sum_{\tupB \in \StructVA^k} S(\tupA,\tupB) \cdot \sum t_{\tupB}.\]
	Now $S(\tupA,\tupB) = 1$ for an odd number of $\tupB \in \StructVA^k$,
	because~$S$ is odd-filled.
	Hence, $\sum r_{\tupA}$ is the sum of an odd number of $\sum t_{\tupB}$,
	of which each is odd because~$T$ is odd-filled.
	So $\sum r_{\tupA} = 1$
	and~$r_{\tupA}$ contains an odd number of ones.
\end{proof}

\begin{lemma}
	\label{lem:block-diagonal-odd-filled-orbit-invariant-columns}
	Let $f,g\colon E \to \ZZ_{2^\pot}$ not twist $\orig{\pTupA}$
	and~$S$ be an $\StructVA^k \times \StructVA^k$ matrix over~$\FF_2$.
	If~$S$ is odd-filled
	and both orbit-diagonal and orbit-invariant
	over $(\StructA_f, \pTupA)$
	and $(\StructA_g, \pTupA)$,
	then every column of~$S$ contains an odd number of ones.
\end{lemma}
\begin{proof}
	Consider the block $S_{\BlockA\times \BlockB}$
	for arbitrary $\BlockA\in \orbs{k}{(\StructA_f,\pTupA)}$ and $\BlockB\in\orbs{k}{(\StructA_g,\pTupA)}$ of the same type.
	Let $\BlockA= \set{\tupA_1,\dots, \tupA_n}$ and
	$\BlockB= \set{\tupB_1,\dots,\tupB_n}$.
	Then consider automorphisms~$\autoA_i$
	such that $\autoA_i(\tupA_1) = \tupA_i$.
	Because the induced action of $\autgrp{(\StructA, \pTupA)}$ on~$\BlockA$
	(and on~$\BlockB$) is regular (Lemma~\ref{lem:cfi-orbit-auto-group-regular}),
	the action of~$\autoA_i$ on~$\BlockA$ (and so~$\BlockB$) is
	uniquely determined.
	Now we consider w.l.o.g.~the column indexed by~$\tupB_1$:
	\[S(\tupA_i,\tupB_1) = \inv{\autoA}_i(S)(\tupA_i,\tupB_1)= S(\tupA_1,\inv{\autoA_i}(\tupB_1))\]
	because $S$ is orbit-invariant.
	So the column indexed by~$\tupB_1$
	contains exactly the entries of the row indexed by~$\tupA_1$.
	That is, the number of ones in every column is odd.
\end{proof}

\begin{lemma}
	\label{lem:odd-ones-F2-full-space}
	Let $\tup{a} \in \FF_2^N$ for some finite set~$N$
	and $\group < \SymSetGroup{N}$
	be a regular and abelian $2$-group.
	If the number of ones in~$\tup{a}$ 
	is odd,
	then the set
	$B := \setcond{\perm(\tup{a})}{\perm \in \group}$
	is a basis of~$\FF_2^N$.
\end{lemma}
\begin{proof}
	Assume w.l.o.g.~that $N = [\ell]$
	and let $W \subseteq \FF_2^N$ be the linear space spanned by~$B$.
	Because~$\group$ is regular, it consists of~$\ell$ many permutations 
	$\group = \set{\perm_1, \dots, \perm_\ell}$ such that
	$\perm_i(1) = i$ for all $i \in [\ell]$.
	By definition,~$W$ is invariant under permutations of~$\group$.
	In coding theory, such a linear space is called an abelian code.
	It is known that~$W$ can be identified with an ideal of
	the group algebra $\FF_2[\group]$~\cite{Berman67},
	which is the set of formal sums
	\[\setcond[\bigg]{\sum_{g \in \group} b_gg}{b_g \in \FF_2}.\]
	This set is naturally an~$\FF_2$-vector space indexed by~$\group$.
	To turn it into a $\FF_2$-algebra, multiplication is defined via
	\[\bigg(\sum_{g \in \group} b_gg\bigg)\cdot\bigg(\sum_{g \in \group} c_gg\bigg) := \sum_{g,h \in \group} (b_g\cdot c_h)(g\cdot h).\]
	A set $I \subseteq \FF_2[\group]$ is a (left) ideal of the algebra $\FF_2[\group]$ if $g \cdot h \in I$
	for every $g \in \FF_2[\group]$ and $h \in I$,
	i.e., $ \FF_2[\group]\cdot I = I$.
	The abelian code~$W$ is identified with an ideal of $\FF_2[\group]$
	via the bijection
	$(b_1, \dots, b_\ell) \mapsto \sum_{i = 1}^\ell b_i \perm_i$ for every $\tup{b} \in W$.
	
	Let $I \subseteq \FF_2[\group]$ be the corresponding ideal of~$W$
	and let the number of ones of $\tup{a} \in W$ be odd.
	Because~$\group$ is a $2$-group, there is a~$k$ such that
	$\perm_i^{(2^{\smash{k}})} = 1_\group$ for all $i \in [\ell]$.
	Because~$\group$ is abelian and we consider~$\FF_2$,
	$ (b \perm_i)(c \perm_j)  + (c \perm_j)(b \perm_i) = 2(b \perm_i)(c \perm_j) = 0$.
	So
	$(b \perm_i + c \perm_j)^2 = (b \perm_i)^2 + (c \perm_j)^2$ and $(b \perm_i + c \perm_j)^{(2^k)} = (b \perm_i)^{(2^k)} + (c \perm_j)^{(2^k)}$.
	It follows that
	\[\left( \sum_{i=1}^\ell a_i \perm_i \right)^{(2^k)} = 
		\sum_{i=1}^\ell (a_i \perm_i)^{(2^k)} =
		\sum_{i=1}^\ell a_i^{(2^k)} 1_\group =
		\sum_{i=1}^\ell a_i 1_\group = 1_\group.
 	\]
 	The last step holds because the number of ones in~$\tup{a}$ is odd.
 	So $\sum_{i=1}^\ell a_i \perm_i$ is a unit with inverse 
 	$( \sum_{i=1}^\ell a_i \perm_i)^{2^k-1}$.
 	First, $\sum_{i=1}^\ell a_i \perm_i \in I$ because $\tup{a} \in W$.
 	Second, $1_\group \in I$ because the inverse of $\sum_{i=1}^\ell a_i \perm_i \in I$ is clearly contained in $\FF_2[\group]$ and $\FF_2[\group] \cdot I = I$.
 	Thus, $I = \FF_2[\group]$ and $W = \FF_2^N$.
 	Finally,~$B$ must be a basis of~$W$ because $|B| = |N|$.
\end{proof}

\begin{lemma}
	\label{lem:orbit-diagonal+orbit-invariant+odd-filled-implies-invertible}
	Let $f,g\colon E \to \ZZ_{2^\pot}$ not twist $\orig{\pTupA}$ and~$S$ be an $\StructVA^k \times \StructVA^k$ matrix over~$\FF_2$.
	If~$S$ is odd-filled and both
	orbit-diagonal and orbit-invariant over $(\StructA_f, \pTupA)$
	and $(\StructA_g, \pTupA)$,
	then~$S$ is invertible.
\end{lemma}
\begin{proof}
	It suffices to show that each block on the diagonal of~$S$ is invertible
	because~$S$ is orbit-diagonal.	
	Let $\BlockA \in \orbs{k}{(\StructA_f,\pTupA)}$ and $\BlockB \in \orbs{k}{(\StructA_g,\pTupA)}$ be of the same type.
	Because~$S$ is odd-filled and orbit-diagonal, $S_{\BlockA\times\BlockB}$ is also odd-filled.
	By Lemma~\ref{lem:cfi-orbit-auto-group-regular},
	the action of $\autgrp{(\StructA_f,\pTupA)}$ on~$\BlockA$
	induces a regular and abelian $2$-group~$\group$.
	By Corollary~\ref{cor:cfi-same-autgroup-and-orbits},
	the action of $\autgrp{(\StructA_g,\pTupA)}$ on~$\BlockB$
	yields the same group~$\group$.
	Let $n := |\BlockA|$,
	$\BlockA = \set{\tupA_1, \dots, \tupA_n}$,
	and~$s_i$ be the row of
	$S_{\BlockA\times\BlockB}$ indexed by~$\tupA_i$.
	We want to show that $s_i = \autoA_i(s_1)$
	for a unique $\autoA_i \in \group$.
	Each $\autoA \in \group$ acts as a permutation on the entries of each~$s_i$,
	that is $(\autoA(s_i))(\tupB) = s_i(\autoA(\tupB))$.
	Let $\group = \set{\autoA_1,\dots,\autoA_n}$ such that $\inv{\autoA_i}(\tupA_1) = \tupA_i$ for every $i \in[n]$
	(this is possible because~$\group$ is regular).
	Then 
	\begin{align*}
		(\autoA_i(s_1))(\tupB) &= S_{\BlockA\times\BlockB}(\tupA_1,\autoA_i(\tupB)) = S_{\BlockA\times\BlockB}(\inv{\autoA_i}(\tupA_1),\tupB)
	\intertext{
		because~$S$ is orbit-invariant.	Hence,}
		(\autoA_i(s_1))(\tupB) &=
	S_{\BlockA\times\BlockB}(\inv{\autoA_i}(\tupA_1),\tupB) = s_i(\tupB),
	\end{align*}
	i.e.,~$\autoA_i(s_1) = s_i$.
	Finally,
	$\setcond{\autoA_i(s_1)}{i \in [n]} = \set{s_1, \dots, s_n}$
	forms a basis of~$\FF_2^n$ by Lemma~\ref{lem:odd-ones-F2-full-space}.
	That is, $S_{\BlockA\times\BlockB}$ has full rank
	and is invertible.
\end{proof}

\section{The Arity 1 Case}
\label{sec:1-ary-case}

To separate rank logic from \PTime{},
we want to show that for every arity~$k$ and number of pebbles $2k+m$,
there are two non-isomorphic CFI structures over $\ZZ_{2^\pot}$ for a suitable $\pot \in \nat$
for which Duplicator has a winning strategy in the invertible-map game $\invertMapGame{2k+m}{k}{\set{2}}$.
This implies $\IFPRP{\set{2}}$-undefinability of the CFI query by Lemma~\ref{lem:invertible-map-refines-rank-logic}
and $\IFPR{}$-undefinability by Lemma~\ref{lem:IFPR-to-IFPR2}.
The most challenging part of constructing winning strategies for Duplicator
in the invertible-map game is to provide similarity matrices.
Indeed, our goal is to construct matrices blurring the twist.
Once we achieve this, it suffices to ensure that the pebbled tuples in both structures always have the same type. This final step is made formal in Section~\ref{sec:separating}.
Constructing matrices blurring the twists for an arbitrary arity~$k$
turns out to be formally intricate
and is in particular recursive on the arity.
In this section, we start with constructing matrices for arity~$1$,
which serve as a base case for the recursion.
We introduce basic techniques that we generalize to higher arities later in Section~\ref{sec:higher-arities}.

Let $\pot \geq 2$, $m \in \nat$,
 $G=(V,E, \leq )$ be an $(m+3)$-connected base graph,
 $\scenter \in V$ be a vertex of degree $d$, and
$\set{\scenter, \stip} \in E$.
Let $f, g \colon E \to \ZZ_{2^\pot}$ 
such that $\set{\scenter,\stip}$ is the only twisted edge and
$g(\set{\scenter,\stip}) = f(\set{\scenter,\stip}) + 2^{\pot-1}$.
The number~$m$ is the number of pebbles
remaining on the structure when Spoiler picks up the $2=2k$ many pebbles 
before Duplicator needs to provide the similarity matrix (we consider arity $k=1$ in this section).
From another perspective,~$m$ corresponds to the number of free variables
of a rank operator.
Set $\StructA_f := \CFIgraph{2^\pot}{G}{f}$ and $\StructA_g := \CFIgraph{2^\pot}{G}{g}$,
both with universe~$\StructVA$.
Let $\pTupA \in \StructVA^m $ 
such that $\distance{G}{\scenter}{\orig{\pTupA}} \geq 3$,
in particular $g$ and $f$ do not twist $\orig{\pTupA}$.
The tuple~$\pTupA$ is the tuple of vertices on which the pebbles remain.
It suffices to consider only a single tuple~$\pTupA$
for both structures
because we will ensure that the pebbled tuples always have the same type in both structures.
Whenever the pebbled tuples have the same type but are not equal,
we can consider an isomorphic structure
in which we moved the twist to an edge far apart from the pebbled tuples.
Then the tuples are equal and we
ensured that the distance between $\orig{\pTupA}$ and the twisted edge is sufficiently large (details in Section~\ref{sec:separating}).

For $\bvertA \in V$, let $\StructVA_\bvertA$
be the set of vertices originating from~$\bvertA$,
i.e.,~the vertices of the gadget for~$\bvertA$.
The key idea is to ``distribute'' the twist among multiple edges,
such that it cannot be detected by Spoiler.
For this, we introduce blurrers,
the key ingredient to define the desired similarity matrix.
\begin{definition}
	\label{def:blurrer-1-dim}
	Let $\blurrer \subseteq \ZZ_{2^\pot}^{d}$.
	For $b \in \ZZ_{2^\pot}$ and $j \in [d]$ we define
	\[\countVects{j}{b}{\blurrer} \coloneqq \left|\setcond[\big]{\tup{a} \in \blurrer}{a_j=b}\right| \bmod 2.\]
	The set~$\blurrer$ is called a \defining{$(\pot,d)$-blurrer}
	if it satisfies
	\begin{enumerate}
		\item $\sum \tup{a} = 0$ for all $a \in \blurrer$,
		\item $\countVects{1}{2^{\pot-1}}{\blurrer} = 1$,
		\item $\countVects{j}{0}{\blurrer} = 1$ for all $1 < j \leq d$, and
		\item $\countVects{j}{b}{\blurrer} = 0$ for all other pairs of $b\in \ZZ_{2^\pot}$ and $j \in [d]$.
	\end{enumerate}
\end{definition}
From now on, we use the letter~$\blurElem$ for elements of a blurrer~$\blurrer$. 
Note that~$\blurrer$ consists solely of tuples satisfying $\sum \blurElem = 0$,
i.e.,~we can later turn every $\blurElem \in \blurrer$ into an automorphism.
But intuitively, when looking at a single index and summing over all $\blurElem \in \blurrer$,
it looks like there is a twist at index~$1$
and no twist at all other indices.

\begin{lemma}
	\label{lem:1ary-blurrer}
	The size $|\blurrer|$ of every $(\pot,d)$-blurrer is odd.
	For every $d\geq 3$, there is a $(\pot,d)$-blurrer.
\end{lemma}
\begin{proof}
	By~Conditions~2 and~4 it holds that \[|\blurrer|= \sum_{b \in \ZZ_{2^q}} \countVects{1}{b}{\blurrer} = \countVects{1}{2^{\pot-1}}{\blurrer} + \sum_{b \in \ZZ_{2^q} \setminus \set{2^{\pot-1}}} \countVects{1}{b}{\blurrer} = 1  \bmod 2.\]
	
	For $d \geq 3$, set $\blurrer := 2^{\pot-2}  \cdot \set{(3,0,1, 0,\dots, 0), (3,1,0, 0,\dots, 0), (2,1,1, 0,\dots, 0)}$.
\end{proof}

Let $\PartA_j = \orbs{j}{(\StructA_f, \pTupA)}$
and $\PartB_j = \orbs{j}{(\StructA_g, \pTupA)}$ for every $j \in [2]$.
For $\BlockA \in \PartA_2$ we set $\BlockA_i := \restrictVect{\BlockA}{i}$
for every $i \in [2]$ and likewise for a $\BlockB \in \PartB_2$.
By Corollary~\ref{cor:orbit-types-unique},
$\BlockA_i$ satisfies $\BlockA_i = \StructVA_\bvertA$
if $\bvertA = \orig{\BlockA_i}$ and $\distance{G}{\bvertA}{\orig{\pTupA}} > 1$.
Moreover, every $\BlockA \in \PartA_1$
is also in $\PartB_1$ and has the same type in $(\StructA_f, \pTupA)$ as in $(\StructA_g, \pTupA)$.

Let~$\blurrer$ be a $(\pot,d)$-blurrer (note that~$\scenter$ is of degree $d\geq 3$ because $G$ is $(m+3)$\nobreakdash-connected) and 
$\neighbors{G}{\scenter} = \set{\stip_1, \dots, \stip_d}$
such that $\stip_1 = \stip$.
Then we can view $\blurElem \in \blurrer$ also as a tuple $\blurElem \in \ZZ_{2^\pot}^{\neighbors{G}{\scenter}}$.
Thus,~$\blurElem$ acts on vertices~$\vertA$ originating from~$\scenter$ 
and we denote this action by 
$\gadgetiso{\vertA}{\blurElem}$
(cf.~Section~\ref{sec:cfi-isomorphisms} for a definition of the action).
Note that every $\blurElem\in \blurrer$ extends to an automorphism of
$(\StructA_f, \pTupA)$ (and so of $(\StructA_g, \pTupA)$):
By Corollary~\ref{cor:orbit-types-unique},
the gadget of~$\scenter$ consists of a single orbit
because $\distance{G}{\scenter}{\pTupA} \geq 3$,
i.e.,~$\StructVA_\scenter \in \PartA_1$.
We define an $\StructVA \times \StructVA$ matrix~$S$ over~$\FF_2$,
which is orbit-diagonal over $(\StructA_f, \pTupA)$ and $(\StructA_g, \pTupA)$.
We set $S_\BlockA := S_{\BlockA \times \BlockA}$ and define
\[
	S_\BlockA(\vertA, \vertB) := \begin{cases}
		1 & \text{if }  \orig{\BlockA} \neq \scenter \text{ and } \vertA =\vertB,\\
		1 & \text{if } \orig{\BlockA} = \scenter \text{ and } \gadgetiso{\vertA}{\blurElem} = \vertB \text { for some } \blurElem \in \blurrer,\\
		0 & \text{otherwise.}
	  	\end{cases}
\]
Of particular interest is the unique $1$-orbit~$\BlockA_\scenter$ with origin~$\scenter$.
We have already seen that $\BlockA_\scenter = \StructVA_\scenter \in \PartA_1$,
because $\distance{G}{\scenter}{\orig{\pTupA}} \geq 3$ by assumption.
For all other orbits $\BlockA \in \PartA_1$,
it is easy to see that $S_\BlockA = \idmat$.
\begin{lemma}
	\label{lem:1ary-orbit-invariant}
	The matrix~$S$ is orbit-invariant over $(\StructA_f, \pTupA)$ and $(\StructA_g, \pTupA)$.
\end{lemma}
\begin{proof}
	Let $\BlockA \in \PartA_1$,
	$\autoA \in \autgrp{(\StructA_f, \pTupA)}$,
	$\vertA \in \BlockA$, and
	$\vertB \in \BlockB = \BlockA \in \PartB_1$.
	If $\BlockA \neq \BlockA_\scenter$,
	then clearly $\autoA(S_\BlockA) =  \autoA(\idmat) = \idmat = S_\BlockA$.
	Otherwise, $\BlockA = \BlockA_\scenter$.
	Because the automorphism group of~$\StructA_f$ is abelian (Lemma~\ref{lem:cfi-autgroup-abelian-two-group})
	and every $\blurElem \in \blurrer$ extends to an automorphism,
	it holds that $\gadgetiso{\autoA(\vertA)}{\blurElem} = \autoA(\gadgetiso{\vertA}{\blurElem})$.
	So $S_{\BlockA_\scenter}(\autoA(\vertA),\autoA(\vertB)) = 1$
	if and only if $\gadgetiso{\autoA(\vertA)}{\blurElem} = \autoA(\gadgetiso{\vertA}{\blurElem}) = \autoA(\vertB)$
	for some $\blurElem \in \blurrer$
	if and only if $\gadgetiso{\vertA}{\blurElem} = \vertB$
	for some $\blurElem \in \blurrer$, i.e,  $S_{\BlockA_\scenter}(\vertA,\vertB) = 1$.
\end{proof}

\begin{lemma}
	\label{lem:1ary-odd-filled}
	The matrix~$S$ is odd-filled.
\end{lemma}
\begin{proof}
	Let $\BlockA \in \PartA_1$.
	For $\BlockA \neq \BlockA_\scenter$,
	the number of ones in a row of  $S_\BlockA = \idmat$ is one and thus odd.
	In~$S_{\BlockA_\scenter}$, the number of ones in a row is 
	$|\blurrer|$ because $\gadgetiso{\vertA}{\blurElem} \neq \gadgetiso{\vertA}{\blurElem'}$ if $\blurElem \neq \blurElem'$
	(Lemma~\ref{lem:cfi-orbit-auto-group-regular})
	and if $\vertA \in \BlockA_\scenter$, then $\gadgetiso{\vertA}{\blurElem} \in\BlockA_\scenter$ for every
	$\blurElem \in \blurrer$.
	From Lemma~\ref{lem:1ary-blurrer} it follows that~$|\blurrer|$ is odd.
\end{proof}

\begin{corollary}
	\label{cor:1ary-matrix-invertible}
	The matrix~$S$ is invertible.
\end{corollary}
\begin{proof}
	Apply Lemmas~\ref{lem:orbit-diagonal+orbit-invariant+odd-filled-implies-invertible},~\ref{lem:1ary-orbit-invariant}, and%
	~\ref{lem:1ary-odd-filled}.
\end{proof}

\begin{figure}
	\centering
	\begin{tikzpicture}
		\tikzstyle{vertex} = [circle, fill=black, inner sep=0.5mm]
		\begin{scope}[ scale=1,every node/.style={transform shape}]
			\begin{scope}
				\begin{scope}
					\node[vertex] (At0) at (0,0) {};
					\node[vertex] (At1) at (0,-0.5) {};
					\node[vertex] (At2) at (0,-1) {};
					\node[vertex] (At3) at (0,-1.5) {};
					\node (At0l) at (-0.75,0) {$\StructVA_{\scenter,\stip,0}$};
					\node (At1l) at (-0.75,-0.5) {$\StructVA_{\scenter,\stip,1}$};
					\node (At2l) at (-0.75,-1) {$\StructVA_{\scenter,\stip,2}$};
					\node (At3l) at (-0.75,-1.5) {$\StructVA_{\scenter,\stip,3}$};
					\draw[black] (-0.5,-0.75) ellipse (1 and 1.3);
					\node (Atl) at (-0.5, -2.4) {\large \strut$\StructVA_\scenter$};
				\end{scope}
				\node (Af) at (1, -2.4) {\large \strut$\StructA_f$};
				\begin{scope}[xshift=2cm]
					\node[vertex] (Atp0) at (0,0) {};
					\node[vertex] (Atp1) at (0,-0.5) {};
					\node[vertex] (Atp2) at (0,-1) {};
					\node[vertex] (Atp3) at (0,-1.5) {};
					\node (Atp0l) at (0.75,0) {$\StructVA_{\stip,\scenter,0}$};
					\node (Atp1l) at (0.75,-0.5) {$\StructVA_{\stip,\scenter,1}$};
					\node (Atp2l) at (0.75,-1) {$\StructVA_{\stip,\scenter,2}$};
					\node (Atp3l) at (0.75,-1.5) {$\StructVA_{\stip,\scenter,3}$};
					\draw[black] (0.5,-0.75) ellipse (1 and 1.3);
					\node (Atl) at (0.5, -2.4) {\large \strut$\StructVA_{\stip}$};
				\end{scope}
				\path[blue, thick]
				(At0) edge (Atp0)
				(At1) edge (Atp3)
				(At2) edge (Atp2)
				(At3) edge (Atp1);
				
				\path[red, dashed ,thick]
				(At0) edge (Atp2)
				(At1) edge (Atp1)
				(At2) edge (Atp0)
				(At3) edge (Atp3);
			\end{scope}
			\begin{scope}[xshift = 7cm]
				\begin{scope}
					\node[vertex] (Bt0) at (0,0) {};
					\node[vertex] (Bt1) at (0,-0.5) {};
					\node[vertex] (Bt2) at (0,-1) {};
					\node[vertex] (Bt3) at (0,-1.5) {};
					\node (Bt0l) at (-0.75,0) {$\StructVA_{\scenter,\stip,0}$};
					\node (Bt1l) at (-0.75,-0.5) {$\StructVA_{\scenter,\stip,1}$};
					\node (Bt2l) at (-0.75,-1) {$\StructVA_{\scenter,\stip,2}$};
					\node (Bt3l) at (-0.75,-1.5) {$\StructVA_{\scenter,\stip,3}$};
					\draw[black] (-0.5,-0.75) ellipse (1 and 1.3);
					\node (Btl) at (-0.5, -2.4) {\large \strut$\StructVA_\scenter$};
				\end{scope}
				\node (Bf) at (1, -2.4) {\large \strut$\StructA_g$};
				\begin{scope}[xshift=2cm]
					\node[vertex] (Btp0) at (0,0) {};
					\node[vertex] (Btp1) at (0,-0.5) {};
					\node[vertex] (Btp2) at (0,-1) {};
					\node[vertex] (Btp3) at (0,-1.5) {};
					\node (Btp0l) at (0.75,0) {$\StructVA_{\stip,\scenter,0}$};
					\node (Btp1l) at (0.75,-0.5) {$\StructVA_{\stip,\scenter,1}$};
					\node (Btp2l) at (0.75,-1) {$\StructVA_{\stip,\scenter,2}$};
					\node (Btp3l) at (0.75,-1.5) {$\StructVA_{\stip,\scenter,3}$};
					\draw[black] (0.5,-0.75) ellipse (1 and 1.3);
					\node (Btl) at (0.5, -2.4) {\large \strut$\StructVA_{\stip}$};
				\end{scope}
				\path[dashed, red, thick]
				(Bt0) edge (Btp0)
				(Bt1) edge (Btp3)
				(Bt2) edge (Btp2)
				(Bt3) edge (Btp1);
				
				\path[blue, thick]
				(Bt0) edge (Btp2)
				(Bt1) edge (Btp1)
				(Bt2) edge (Btp0)
				(Bt3) edge (Btp3);
			\end{scope}
		\end{scope}
	\end{tikzpicture}
	\caption{
		This figure shows the twist between~$\StructA_f$ and~$\StructA_g$.
		It assumes that we consider~$\ZZ_4$ and that $f(\set{\scenter,\stip}) = 0$ and $g(\set{\scenter,\stip}) = 2$.
		It shows the twisted connection for the edge $\set{\scenter,\stip}$.
		On the left, there are the two gadgets for the vertices~$\scenter$ and~$\stip$ in~$\StructA_f$
		and on the right there are the same gadgets in~$\StructA_g$.
		Every vertex represents a clique corresponding to the
		$\StructVA_{\scenter,\stip,c}$ and every edge a complete bipartite graph (cf.~Section~\ref{sec:cfi-structures}).
		The relation~$\rel_{E,0}$ is drawn in blue
		and~$\rel_{E,2}$ in red and dashed style.
		Restricted to the connection between $\set{\scenter,\stip}$,
		we have in~$\StructA_f$ that $\rel_{E,0} = E_{\set{\scenter,\stip},0}$
		and $\rel_{E,2} = E_{\set{\scenter,\stip},2}$.
		In~$\StructA_g$ we have that $\rel_{E,0} = E_{\set{\scenter,\stip},2}$
		and $\rel_{E,2} = E_{\set{\scenter,\stip},0}$.
	}
	\label{fig:twist}
\end{figure}

We want to define a function $\invMapBij \colon \PartA_2 \to \PartB_2$
such that it maps an orbit to another orbit of the same type.
By Corollary~\ref{cor:cfi-same-autgroup-and-orbits},
we know that $\PartA_2 = \PartB_2$
and by Lemma~\ref{lem:cfi-occuring-orbit-types-equal}
that a type-preserving bijection exists.
Let $\BlockA \in \PartA_2$ with origin $(\bvertA, \bvertB)$.
If $\set{\scenter, \stip} \neq \set{\bvertA, \bvertB}$,
we set $\invMapBij(\BlockA) := \BlockA$.
Otherwise if $(\stip, \scenter) = (\bvertA, \bvertB)$,
then~$\BlockA$ has a different type
in $(\StructA_f,\pTupA)$ than in $(\StructA_g, \pTupA)$:
Every vertex in~$\BlockA_1$ is related with every vertex in~$\BlockA_2$
via some~$\rel_{E,c}$.
By Corollary~\ref{cor:orbit-types-unique},
we have that $\BlockA = E_{\set{\scenter,\stip},a}$ for some $a \in \ZZ_{2^\pot}$
(recall our assumption $\distance{G}{\scenter}{\pTupA} \geq 3$
and thus~$a$ determines the type of~$\BlockA$).
We set $\invMapBij(\BlockA) :=  E_{\set{\scenter,\stip},a+2^{\pot-1}}$,
which then has the same type in $(\StructA_g,\pTupA)$
because of the twist (cf.~Figure~\ref{fig:twist}).
The case of $(\scenter,\stip)$ is analogous.

\begin{lemma}
	\label{lem:1-ary-matrices-similar}
	$\charMat{\BlockA} \cdot S = S \cdot \charMat{\invMapBij(\BlockA)}$ for every $\BlockA \in \PartA_2$.
\end{lemma}
\begin{proof}
	Let $\BlockA \in \PartA_2$ and $\orig{\BlockA} = (\bvertA,\bvertB)$
	and set $\BlockB:= \invMapBij(\BlockA)$.
	Clearly $\BlockA\subseteq \BlockA_1 \times \BlockA_2$.
	We also have $\BlockA_1 = \BlockB_1$ and $\BlockA_2 = \BlockB_2$
	(as seen earlier by Corollary~\ref{cor:orbit-types-unique}).
	Then the $\BlockA_1 \times \BlockA_2$ block is the only nonzero block of $\charMat{\BlockA}$.
	Because~$S$ is orbit-diagonal,
	$\charMat{\BlockA} \cdot S$ has only one nonzero block,
	namely the $\BlockA_1 \times \BlockB_2$ block, which satisfies
	$(\charMat{\BlockA} \cdot S)_{\BlockA_1 \times \BlockB_2} =
	\charMat{\BlockA}_{\BlockA_1 \times \BlockA_2} \cdot S_{\BlockA_2 \times \BlockB_2}$.
	Likewise,
	$(S \cdot \charMat{\BlockB})_{\BlockA_1 \times \BlockB_2} = S_{\BlockA_1\times\BlockB_1} \cdot \charMat{\BlockB}_{\BlockB_1 \times \BlockB_2}$.
	Recall that we have set $S_{\BlockA_2} = S_{\BlockA_2 \times \BlockB_2}$.
	We identify~$\charMat{\BlockA}$ with $\charMat{\BlockA}_{\BlockA_1 \times \BlockA_2}$ and likewise for~$\charMat{\BlockB}$.
	So we are left to show that
	$\charMat{\BlockA} \cdot S_{\BlockA_2} = S_{\BlockB_1} \cdot \charMat{\BlockB}$.
	\begin{itemize}
		\item Case $\scenter \notin \set{\bvertA,\bvertB}$: Then
		$\BlockB = \invMapBij(\BlockA) = \BlockA$
		and
		$\charMat{\BlockA} \cdot S_{\BlockA_2} = 
		\charMat{\BlockA} \cdot \idmat =
		\idmat \cdot \charMat{\BlockB} =
		S_{\BlockB_1} \cdot \charMat{\BlockB}$.
		
		\item Case $\bvertA = \bvertB = \scenter$:
		Then $\BlockB= \invMapBij(\BlockA) = \BlockA$.
		As already seen, $\BlockA_1 = \BlockA_2 = \BlockB_1 = \BlockB_2 =\BlockA_\scenter$. So if $\vertA \in \BlockB_2$, then $\inv{\blurElem}(\vertA) \in \BlockB_2$ for every $\blurElem \in \blurrer$.
		We obtain 
		\begin{align*}
		 	(\charMat{\BlockA} \cdot S_{\BlockA_2}) (\vertA, \vertB)
			&= \sum_{\vertC \in \BlockA_2}
			\charMat{\BlockA}(\vertA, \vertC) \cdot S_{\BlockA_\scenter}(\vertC,\vertB)\\
			&=\sum_{\blurElem \in \blurrer} 
			\charMat{\BlockA}(\vertA, \gadgetiso{\vertB}{\inv{\blurElem}})\\
			&=\sum_{\blurElem \in \blurrer} 
			\charMat{\BlockA}(\gadgetiso{\vertA}{\blurElem}, \vertB).
		\intertext{
		The last step uses that~$\blurElem$ extends to an automorphism
		and thus
		$\charMat{\BlockA}(\vertA, \gadgetiso{\vertB}{\inv{\blurElem}}) =
		\blurElem(\charMat{\BlockA})(\vertA, \gadgetiso{\vertB}{\inv{\blurElem}}) = \charMat{\BlockA}(\gadgetiso{\vertA}{\blurElem}, \vertB)$.
		The reverse direction is similar:}
			&\hiddenEq \sum_{\blurElem \in \blurrer} \charMat{\BlockA}(\gadgetiso{\vertA}{\blurElem},\vertB)\\
			&=\sum_{\blurElem \in \blurrer} \charMat{\BlockB}(\gadgetiso{\vertA}{\blurElem},\vertB)\\
			&=\sum_{\vertC \in \BlockB_1}
			S_{\BlockB_1}(\vertA, \vertC) \cdot
			\charMat{\BlockB}(\vertC,\vertB)\\
			&=(S_{\BlockB_1} \cdot \charMat{\BlockB}) (\vertA,\vertB).
		\end{align*}

		\item Case $\bvertB = \scenter$ and $\set{\bvertA, \scenter} \in E$ (the case $\bvertA = \scenter$ and $\set{\scenter, \bvertB} \in E$ is analogous):
		Again $\BlockA_2 = \BlockA_\scenter$.
		 We have
		\begin{align*}
			&\hiddenEq (\charMat{\BlockA} \cdot S_{\BlockA_2}) (\vertA, \vertB)\\
			&= \sum_{\vertC \in \BlockA_\scenter}
			\charMat{\BlockA}(\vertA, \vertC) \cdot S_{\BlockA_\scenter}(\vertC,\vertB)\\
			&= \sum_{\blurElem \in \blurrer} \charMat{\BlockA}(\vertA,\gadgetiso{\vertB}{\inv{\blurElem}})\\
			&=\begin{cases}
			1 & \text{if } |\setcond{\blurElem \in \blurrer}{(\vertA,\gadgetiso{\vertB}{\inv{\blurElem}}) \in \BlockA}| \text{ is odd,}\\
			0 & \text{otherwise.}
		\end{cases}	
		\end{align*}
		Let $\BlockA = E_{\set{\bvertA, \scenter},a}$ for some $a \in \ZZ_{2^\pot}$
		(cf.~the definition of~$\invMapBij$)
		and $(\vertA, \vertB) \in  E_{\set{\bvertA,\scenter},b}$ for some $b \in \ZZ_{2^\pot}$.
		Then, by definition of $E_{\set{\bvertA,\scenter},b}$, it holds that
		$\vertA(\scenter) + \vertB(\bvertA) = b$.
		Let $i \in [d]$ such that $\bvertA = \stip_i$
		(recall that $\neighbors{G}{\scenter} = \set{\stip_1, \dots, \stip_d}$
		and $\stip_1 = \stip$).
		For every $\blurElem \in \blurrer$ it holds that $(\vertA,\gadgetiso{\vertB}{\inv{\blurElem}}) \in \BlockA  = E_{(\bvertA, \scenter),a}$
		if and only if $\vertA(\scenter)+\inv{\blurElem}(\vertB)(\bvertA) = a$
		if and only if $\blurElem(i) = b - a$
		because
		\[\vertA(\scenter) + \gadgetiso{\vertB}{\inv{\blurElem}}(\bvertA)  = \vertA(\scenter) +\vertB(\bvertA) - \blurElem(i) =  b - \blurElem(i).\]
		We see that  
		\[\left|\setcond*{\blurElem \in \blurrer}{(\vertA,\gadgetiso{\vertB}{\inv{\blurElem}}) \in \BlockA}\right| = \countVects{i}{b-a}{\blurrer}.\]
		
		Set $c := 2^{\pot-1}$ if $i = 1$ (and so $\bvertA = \stip$)
		and $c:=0$ otherwise.
		Then $\countVects{i}{b-a}{\blurrer} = 1$
		if and only if $b-a = c$
		by the properties of a blurrer.
		It follows that 
		\begin{align*}
			(\charMat{\BlockA} \cdot S_{\BlockA_2}) (\vertA, \vertB)
			&=\begin{cases}
				1 & \text{if } b - a = c,\\
				0 &\text{otherwise.}
			\end{cases}
		\end{align*}
		\begin{itemize}
			\item If $i \neq 1$ (so $\bvertA \neq t'$),
			then $c = 0$ and $(\charMat{\BlockA} \cdot S_{\BlockA_2}) (\vertA, \vertB) = 1$ if and only if $b = a$,
			but that holds if and only if $(\vertA,\vertB) \in \BlockA$.
			So
			\[\charMat{\BlockA} \cdot S_{\BlockA_2} = \charMat{\BlockA} = \idmat \cdot \charMat{\BlockB} = S_{\BlockB_1} \cdot  \charMat{\BlockB}\]
			because $\BlockB= \invMapBij(\BlockA) =\BlockA$.
			\item If $i = 1$ (so $\bvertA = \stip$), then
			$(\charMat{\BlockA} \cdot S_{\BlockA_2}) (\vertA, \vertB) = 1$
			if and only if
			$b - a = 2^{\pot-1}$, i.e., $a +  2^{\pot-1} = b$.
			But that holds by definition of $\invMapBij$ if and only if
			\begin{align*}
				(\vertA, \vertB) \in \BlockB &= \invMapBij(\BlockA) = E_{(\bvertA,\scenter),a+2^{\pot-1}}
			\intertext{and so}
				\charMat{\BlockA} \cdot S_{\BlockA_2} &= \idmat \cdot \charMat{\BlockB} = S_{\BlockB_1} \cdot  \charMat{\BlockB}.
		\end{align*}
		\end{itemize}

		\item Case $\bvertB = \scenter$ and $\set{\bvertA, \scenter} \notin E$
		(the case $\bvertA = \scenter$ and $\set{\scenter, \bvertB} \notin E$ is analogous):
		By the assumption that $\distance{G}{\scenter}{\orig{\pTupA}} \geq 3$
		the type of $(\vertA,\vertB)$ and $(\vertA,\vertB')$
		for $\vertA \in \StructVA_\bvertA$ and $\vertB,\vertB' \in \StructVA_\scenter$
		is equal.
		So $(\vertA,\vertB) \in \BlockA$ if and only if  $(\vertA,\vertB') \in \BlockA$ by Corollary~\ref{cor:orbit-types-unique}.
		In particular,
		$(\vertA,\vertB) \in \BlockA$
		if and only if 
		$(\vertA, \gadgetiso{\vertB}{\inv{\blurElem}}) \in \BlockA$
		for every $\blurElem \in \blurrer$.
		Set
		\begin{align*}
			D &:= \setcond*{\blurElem \in \blurrer}{(\vertA,\gadgetiso{\vertB}{\inv{\blurElem}}) \in \BlockA}.
		\intertext{Then we have}
			 (\charMat{\BlockA} \cdot S_{\BlockA_2}) (\vertA,\vertB)
			&=\begin{cases}
				1 & \text{if } |D| \text{ is odd,}\\
				0 & \text{otherwise}
				\end{cases}\\
			&= \charMat{\BlockA}(\vertA,\vertB).
		\end{align*}
		The last step holds because
		if $(\vertA,\vertB) \in \BlockA$, then
		$D = \blurrer$ and
		$|D| = |\blurrer|$ is odd (Lemma~\ref{lem:1ary-blurrer}), and
		if $(\vertA,\vertB) \not\in \BlockA$,
		then $D = \emptyset$ and $|D| = 0$.
		As seen before,
		\[\charMat{\BlockA} \cdot S_{\BlockA_2}=\charMat{\BlockA} = \charMat{\BlockB} = S_{\BlockB_2} \cdot \charMat{\BlockB}\]
		because $\BlockB = \invMapBij(\BlockA) = \BlockA$.
		\qedhere
	\end{itemize}

\end{proof}
\begin{corollary}
	The matrix~$S$ $1$-blurs the twist between $(\StructA_f, \pTupA)$
	and $(\StructA_g, \pTupA)$.
\end{corollary}

We summarize the results of this section:
\begin{lemma}
	\label{lem:1ary-summary}
	For every $\pot \geq 2, m \in \nat$,
	every $(m+3)$-connected base graph $G=(V,E,\leq)$,
	every $f,g \colon E \to \ZZ_{2^\pot}$
	twisting exactly one edge $\set{\scenter,\stip}$
	such that $f(\set{\scenter,\stip}) = g(\set{\scenter,\stip}) + 2^{\pot-1}$,
	and every $m$-tuple $\pTupA \in \StructVA^m$
	of $\StructA_f := \CFIgraph{2^\pot}{G}{f}$
	and $\StructA_g := \CFIgraph{2^\pot}{G}{f}$,
	there is an odd-filled matrix~$S$,
	both orbit-diagonal and orbit-invariant over $(\StructA_f, \pTupA)$ and $(\StructA_g, \pTupA)$,
	that
	$1$-blurs the twist between $(\StructA_f, \pTupA)$
	and $(\StructA_g, \pTupA)$
	and satisfies that $S_{\BlockA,\BlockB} = \idmat$
	for every $k$-orbits $\BlockA \in \orbs{1}{(\StructA_f, \pTupA)}$
	and $\BlockB \in \orbs{1}{(\StructA_g, \pTupA)}$
	of the same type with $\orig{\BlockA} = \orig{\BlockB} \neq \scenter$.
\end{lemma}

Constructing matrices blurring the twist for
higher arities is more difficult:
First, we have to generalize our notion of a blurrer to arity~$k$.
Second, we are faced with disconnected orbits,
which do not pose a problem in the $1$-ary case, but complicate matters in the general case.
To deal with these orbits, we need to establish more technical lemmas for matrices over CFI structures.

\section{The Active Region of a Matrix}
\label{sec:active-region}
In this section, we consider the part of a matrix~$S$, where~$S$ ``has a non-trivial effect''.
Intuitively, this means that~$S$ is locally not the identity matrix.
We will call these parts the active region.
Conversely, for parts where~$S$ is not active,~$S$ is locally the identity matrix.
We now make this idea formal.

As in Section~\ref{sec:matrices-cfi},
let $\pot, k, m \in \nat$ and
$G=(V,E,\leq)$ be a $(k + m + 1)$-connected base graph.
We again denote for every $f \colon E \to \ZZ_{2^\pot}$ by~$\Struct_f$
the CFI structure $\CFIgraph{\ZZ_{2^\pot}}{G}{f}$
and by~$\StructVA$ the universe of these CFI structures.
Let $\pTupA \in \StructVA^m$ be arbitrary but fixed in this section.
For $N \subseteq V$,
the \defining{$N$-components $\ncomp{N}{\BlockA}$} of an orbit~$\BlockA$
is the set of components~$\comp$ of~$\BlockA$ satisfying $\comp \subseteq N$.

\begin{definition}[Active Region]
	Let $f,g \colon E \to \ZZ_{2^\pot}$ not twist $\orig{\pTupA}$,~%
	$S$ be an ${\StructVA^k \times \StructVA^k}$ matrix over~$\FF_2$,
	and $\PartA_k = \orbs{k}{(\Struct_f, \pTupA)}$ and
	$\PartB_k = \orbs{k}{(\Struct_g, \pTupA)}$.
	For $\BlockA \in \PartA_k$ and $\BlockB \in \PartB_k$
	of the same type,
	the matrix~$S$ is \defining{active} (with respect to  $(\StructA_f,\pTupA)$ and $(\StructA_g,\pTupA)$) on a component~$\comp$ of~$\BlockA$ (and so of~$\BlockB$),
	if there are
	$\tupA \in \BlockA$ and
	$\tupB \in \BlockB$
	such that ${\tupA_\comp \neq \tupB_\comp}$
	and $S(\tupA,\tupB) = 1$.
	We write $\activeRegionBIdx{f,g}{\pTupA}{S}{\BlockA} = \activeRegionBIdx{f,g}{\pTupA}{S}{\BlockB}$
	for the set of components of~$\BlockA$, on which~$S$ is active,
	and $\inactiveRegionBIdx{f,g}{\pTupA}{S}{\BlockA} = \inactiveRegionBIdx{f,g}{\pTupA}{S}{\BlockB}$ for the remaining components.
	The \defining{active region $\activeRegionIdx{f,g}{\pTupA}{S} \subseteq V$ of~$S$} is the 
	inclusion-wise smallest set satisfying the following:
	\begin{enumerate}
		\item \label{itm:active-region-contaied} $\comp \subseteq \activeRegionIdx{f,g}{\pTupA}{S}$
		for every $\comp \in \activeRegionBIdx{f,g}{\pTupA}{S}{\BlockA}$
		and every $\BlockA \in \PartA_k$.
		\item \label{itm:active-region-replace}
		For every $\BlockA, \BlockA' \in \PartA_k$
		and $\BlockB, \BlockB' \in \PartB_k$ 
		such that
		\[\ncomp{\activeRegionIdx{f,g}{\pTupA}{S}}{\BlockA} = \ncomp{\activeRegionIdx{f,g}{\pTupA}{S}}{\BlockA'} = \ncomp{\activeRegionIdx{f,g}{\pTupA}{S}}{\BlockB} = \ncomp{\activeRegionIdx{f,g}{\pTupA}{S}}{\BlockB'} =:\mathsf{A},\]
		both~$\BlockA$ and~$\BlockB$ (respectively~$\BlockA'$ and~$\BlockB'$)
		have the same type, and thus 
		$\inactiveRegionBIdx{f,g}{\pTupA}{S}{\BlockA} = \inactiveRegionBIdx{f,g}{\pTupA}{S}{\BlockB} =: \mathsf{N}$
		(respectively $\inactiveRegionBIdx{f,g}{\pTupA}{S}{\BlockA'} = \inactiveRegionBIdx{f,g}{\pTupA}{S}{\BlockB'}=: \mathsf{N}'$),
		and every $\tupA \in \BlockA$, $\tupA' \in \BlockA'$,
		$\tupB \in \BlockB$, and $\tupB' \in \BlockB'$,
		it holds that 
		if $\tupA_{\mathsf{A}} = \primeSub{\tupA}{\mathsf{A}}$,
		$\tupB_{\mathsf{A}} = \primeSub{\tupB}{\mathsf{A}}$,		
		$\tupA_{\mathsf{N}} = \tupB_{\mathsf{N}}$, and
		$\primeSub{\tupA}{\mathsf{N}'} = \primeSub{\tupB}{\mathsf{N}'}$,
		then $S(\tupA,\tupB) = S(\tupA',\tupB')$.		
	\end{enumerate}
\end{definition}
The active region is well-defined:
Clearly~$V$ itself satisfies Conditions~\ref{itm:active-region-contaied} and~\ref{itm:active-region-replace}.
If two sets~$X\subseteq V$ and~$Y \subseteq V$ satisfy the two conditions,
then also $X \cap Y$.
Note that $\ncomp{\activeRegionIdx{f,g}{\pTupA}{S}}{\BlockA}$
and $\inactiveRegionBIdx{f,g}{\pTupA}{S}{\BlockA}$
are not necessarily disjoint,
but $\inactiveRegionBIdx{f,g}{\pTupA}{S}{\BlockA}$ contains all components
of~$\BlockA$ not contained in $\ncomp{\activeRegionIdx{f,g}{\pTupA}{S}}{\BlockA}$.
Condition~\ref{itm:active-region-replace} equivalently can be stated 
only with the remaining components apart from $\ncomp{\activeRegionIdx{f,g}{\pTupA}{S}}{\BlockA}$ 
instead of $\inactiveRegionBIdx{f,g}{\pTupA}{S}{\BlockA}$.

Although Condition~\ref{itm:active-region-replace} is rather technical,
it ensures that the ``non-identity-part'' of~$S$ only depends on the active region:
$S(\tupA,\tupB)$
only depends on the components of~$\tupA$ and~$\tupB$,
on which~$S$ is active,
as long as the entries for the other components are equal
(otherwise $S(\tupA,\tupB)=0$ anyway by Condition~\ref{itm:active-region-contaied}).
That is $S(\tupA,\tupB)$ only depends on whether $\tupA_{\inactiveRegionBIdx{f,g}{\pTupA}{S}{\BlockA}} = \tupB_{\inactiveRegionBIdx{f,g}{\pTupA}{S}{\BlockA}}$
but not on e.g. the type of $\tupA_{\inactiveRegionBIdx{f,g}{\pTupA}{S}{\BlockA}}$.

We first consider the matrix blurring the twist defined in Section~\ref{sec:1-ary-case}:
\begin{lemma}
	\label{lem:1ary-active-region}
	The matrix $S$ given in the setting of Lemma~\ref{lem:1ary-summary}
	satisfies $\activeRegionIdx{f,g}{\pTupA}{S} = \set{\scenter}$.
\end{lemma}
\begin{proof}
	Let $\BlockA \in \orbs{1}{(\StructA_f, \pTupA)}$
	and $\BlockB \in \orbs{1}{(\StructA_g, \pTupA)}$
	be of the same type with origin $\bvertA = \orig{\BlockA} = \orig{\BlockB} \neq \scenter$.
	Then $S_{\BlockA \times \BlockB} = \idmat$ by Lemma~\ref{lem:1ary-summary},
	i.e.,~$S$ is clearly not active on $\set{\bvertA}$.
	The matrix~$S$ has to be active on $\set{\scenter}$
	because otherwise $S = \idmat$ and the structures would be isomorphic.
	This proves Condition~\ref{itm:active-region-contaied}.
	In the $1$-ary case, a $1$-orbit can only have one component,
	so Condition~\ref{itm:active-region-replace} of the active region is trivially satisfied.
\end{proof}

We now continue in the general case.
The rest of this section establishes
rather technical lemmas needed in Section~\ref{sec:higher-arities}.
It is easy to see that if~$\BlockA$ and~$\BlockB$ have the same type,
whose origins contain no vertex of $\activeRegionIdx{f,g}{\pTupA}{S}$,
then $S_{\BlockA\times\BlockB} = \idmat$.
In the region of a twist,~$S$ has to be active:
\begin{lemma}
	\label{lem:not-equal-components-implies-active}
	Let $f,g\colon E \to \ZZ_{2^\pot}$ not twist $\orig{\pTupA}$,~%
	$S$ be an $\StructVA^k \times \StructVA^k$ matrix over~$\FF_2$,
	and $\BlockA \in \orbs{k}{(\StructA_f, \pTupA)}$ and $\BlockB \in \orbs{k}{(\StructA_g, \pTupA)}$
	have the same type.
	If the block $S_{\BlockA \times \BlockB}$ is nonzero
	and~$\comp$ is a component of~$\BlockA$ (and thus of~$\BlockB$)
	such that $\restrictVect{\BlockA}{\comp} \neq \restrictVect{\BlockB}{\comp}$,
	then $\comp \in \activeRegionBIdx{f,g}{\pTupA}{S}{\BlockA}$.
\end{lemma}
\begin{proof}
	Let $\tupA\in \BlockA$ and $\tupB \in \BlockB$ such that 
	$S(\tupA,\tupB) = 1$.
	Such an entry must exist because $S_{\BlockA \times \BlockB}$ is nonzero.
	If $\tupA_\comp = \tupB_\comp$,
	then $\restrictVect{\BlockA}{\comp} = \restrictVect{\BlockB}{\comp}$ by Lemma~\ref{lem:split-discon-orbits} and Corollary~\ref{cor:orbit-types-unique},
	which contradicts our assumption.
	So $\tupA_\comp \neq \tupB_\comp$
	and $\comp \in \activeRegionBIdx{f,g}{\pTupA}{S}{\BlockA}$.
\end{proof}
The next lemma shows that, as long as~$f$ and~$g$ agree on the edges in $\orig{\pTupA}$,
the actual values~$f$ and~$g$ assign to edges~$e$
are not important but only the difference $f(e) -g(e)$ matters.
\begin{lemma}
	\label{lem:active-region-change-equally}
	Let $f,g\colon E \to \ZZ_{2^\pot}$ not twist $\orig{\pTupA}$
	and~$S$ be an $\StructVA^k\times \StructVA^k$ matrix over~$\FF_2$.
	Furthermore, let $f',g'\colon E \to \ZZ_{2^\pot}$
	such that $f'(e) = f(e)$ and $g'(e) = g(e)$ for every $e\in E$ with 
	$e \cap \orig{\pTupA} \neq \emptyset$
	and 
	$f'(e) - f(e) = g'(e) - g(e)$ for every other $e \in E$.
	Then $\activeRegionIdx{f,g}{\pTupA}{S} = \activeRegionIdx{f',g'}{\pTupA}{S}$
	and
	if~$S$ is orbit-diagonal (respectively orbit-invariant) over $(\StructA_f,\pTupA)$ and $(\StructA_g,\pTupA)$,
	then~$S$ is orbit-diagonal (respectively orbit-invariant) over $(\StructA_{f'},\pTupA)$ and $(\StructA_{g'},\pTupA)$.
\end{lemma}
\begin{proof}
	Note that if $\BlockA \in \orbs{k}{(\StructA_f,\pTupA)}$
	has the same type in $(\StructA_f, \pTupA)$
	as $\BlockB \in \orbs{k}{(\StructA_g,\pTupA)}$ has in $(\StructA_g, \pTupA)$,
	then $\BlockA \in \orbs{k}{(\StructA_{f'},\pTupA)}$ and $\BlockB \in \orbs{k}{(\StructA_{g'},\pTupA)}$ (Corollary~\ref{cor:cfi-same-autgroup-and-orbits})
	 and~$\BlockA$ has the same type in $(\StructA_{f'}, \pTupA)$
	 as~$\BlockB$ has in $(\StructA_{g'}, \pTupA)$.
	So we only change the type of the orbits,
	but not the correspondence between orbits of the same type.
	That is, if~$S$ is orbit-diagonal (respectively orbit-invariant) over $(\StructA_f,\pTupA)$ and $(\StructA_g,\pTupA)$,
	then~$S$ is orbit-diagonal (respectively orbit-invariant) over $(\StructA_{f'},\pTupA)$ and $(\StructA_{g'},\pTupA)$.
	Furthermore,~$S$ is active on the same components with respect to $(\StructA_f, \pTupA)$ and $(\StructA_g, \pTupA)$
	as~$S$ is with respect to $(\StructA_{f'}, \pTupA)$ and $(\StructA_{g'}, \pTupA)$.
	This implies that $\activeRegionIdx{f,g}{\pTupA}{S} = \activeRegionIdx{f',g'}{\pTupA}{S}$.
\end{proof}

We now show that the active region of products $S \cdot T$
is bounded by the active regions of~$S$ and~$T$.
For two $k$\nobreakdash-tuples $\tupA,\tupB \in \StructVA^k$
we use the Kronecker delta $\kroneckerEq{\tupA}{\tupB}$,
which is~$1$ if and only if $\tupA = \tupB$
and~$0$ otherwise.

\begin{lemma}
		\label{lem:active-region-mult}
	Let $f,g,h \colon E \to \ZZ_{2^\pot}$ pairwise not twist $\orig{\pTupA}$
	and $S,T$ be $\StructVA^k \times \StructVA^k$ matrices over~$\FF_2$.
	If~$S$ is orbit-diagonal over $(\StructA_f,\pTupA)$ and $(\StructA_g,\pTupA)$ and $T$ is orbit-diagonal over $(\StructA_g,\pTupA)$ and $(\StructA_h,\pTupA)$,
	then $\activeRegionIdx{f,h}{\pTupA}{S\cdot T} \subseteq \activeRegionIdx{f,g}{\pTupA}{S} \cup \activeRegionIdx{g,h}{\pTupA}{T}$.
\end{lemma}
\begin{proof}
	In this proof we omit the superscripts~$f$,~$g$,~$h$, and~$\pTupA$
	for readability:
	for~$S$ we always refer to~$f$ and~$g$,
	for~$T$ to~$g$ and~$h$, and
	for $S\cdot T$ to~$f$ and~$h$.
	We show that $\activeRegion{S} \cup \activeRegion{T}$ satisfies Conditions~\ref{itm:active-region-contaied} and~\ref{itm:active-region-replace}.
	Because the active region is the inclusion-wise minimal set satisfying the two conditions, it then follows that $\activeRegion{S\cdot T} \subseteq \activeRegion{S} \cup \activeRegion{T}$.
	Let $\PartA_k = \orbs{k}{(\StructA_f,\pTupA)}$, $\PartB_k = \orbs{k}{(\StructA_g,\pTupA)}$, and $\PartC_k = \orbs{k}{(\StructA_h,\pTupA)}$.
	
	We show  Condition~\ref{itm:active-region-contaied} by contraposition.
	Let $\BlockA \in \PartA_k$ and~$\comp$ be a connected component of $G[\orig{\BlockA}]$. 
	We show that if $\comp \not\in \activeRegionB{S}{\BlockA} \cup \activeRegionB{T}{\BlockA}$, then $\comp \not\in \activeRegionB{S\cdot T}{\BlockA}$.
	Let $\comp \not\in \activeRegionB{S}{\BlockA} \cup \activeRegionB{T}{\BlockA}$,
	$\BlockB \in \PartB_k$ and $\BlockC \in \PartC_k$ be of the same type as~$\BlockA$, $\tupA \in \BlockA$, and $\tupC \in \BlockC$.
	Because~$S$ and~$T$ are orbit-diagonal,
	\[(S\cdot T)(\tupA, \tupC) = \sum_{\tupB \in \BlockB } S(\tupA,\tupB) \cdot T(\tupB, \tupC).\]
	If $S(\tupA,\tupB) = 1$ (i.e., $S(\tupA,\tupB) \neq 0$), then $\tupA_{\comp} = \tupB_{\comp}$ because $\comp \notin \activeRegionB{S}{\BlockA}$.
	Similarly, $\tupB_{\comp} = \tupC_{\comp}$ if $T(\tupB,\tupC) = 1$.
	This implies $\tupA_{\comp} = \tupC_{\comp}$
	if $(S\cdot T)(\tupA, \tupC) = 1$
	(so there is at least one nonzero summand).
	Hence, $\comp \not\in \activeRegionB{S\cdot T}{\BlockA}$.

	To show Condition~\ref{itm:active-region-replace}, let $\BlockA, \BlockA' \in \PartA_k$
	and $\BlockC, \BlockC' \in \PartC_k$ be arbitrary $k$-orbits,
	such that 
	\[\mathsf{A}:=\ncomp{\activeRegion{S}\cup \activeRegion{T}}{\BlockA} = \ncomp{\activeRegion{S}\cup \activeRegion{T}}{\BlockA'} = \ncomp{\activeRegion{S}\cup \activeRegion{T}}{\BlockC} = \ncomp{\activeRegion{S}\cup \activeRegion{T}}{\BlockC'},\]
	the orbits~$\BlockA$ and~$\BlockC$ have the same type,
	and~$\BlockA'$ and~$\BlockC'$ have the same type.
	Let $\mathsf{N}$ be the set of remaining components of~$\BlockA$ (and so of $\BlockC$) apart from~$\mathsf{A}$.
	Similarly, let~$\mathsf{N}'$ be the set of remaining components of~$\BlockA'$ (and so of~$\BlockC'$) apart from~$\mathsf{A}$ .
	Let $\tupA \in \BlockA$, $\tupA' \in \BlockA'$,
	$\tupC \in \BlockC$, and $\tupC' \in \BlockC'$, such that
	$\tupA_{\mathsf{A}} = \primeSub{\tupA}{\mathsf{A}}$,
	$\tupC_{\mathsf{A}} = \primeSub{\tupC}{\mathsf{A}}$,
	$\tupA_{\mathsf{N}} = \tupC_{\mathsf{N}}$, and
	$\primeSub{\tupA}{\mathsf{N}'} = \primeSub{\tupC}{\mathsf{N}'}$.
	We have to show that $(S\cdot T)(\tupA,\tupC) = (S\cdot T)(\tupA', \tupC')$.
	
	By assumption, $\tupA_\mathsf{A} \in \restrictVect{\BlockA}{\mathsf{A}}$,
	$\primeSub{\tupA}{\mathsf{A}} \in \restrictVect{\BlockA'}{\mathsf{A}} $,
	and $\tupA_\mathsf{A} =\primeSub{\tupA}{\mathsf{A}}$.
	So $\restrictVect{\BlockA}{\mathsf{A}} = \restrictVect{\BlockA'}{\mathsf{A}}$
	by Corollary~\ref{cor:orbit-types-unique}
	because they have the same type and contain the same tuple.	
	Let $\BlockB \in \PartB_k$ be of the same type as~$\BlockA$
	and $\BlockB' \in \PartB_k$ be of the same type as~$\BlockA'$.
	Then $\restrictVect{\BlockB}{\mathsf{A}} = \restrictVect{\BlockB'}{\mathsf{A}}$
	and~$\mathsf{A}$ and~$\mathsf{N}$ (respectively~$\mathsf{N}'$) are sets of components of~$\BlockB$
	(respectively $\BlockB'$).
	We first assume that the blocks $S_{\BlockA \times \BlockB}$ and
	$T_{\BlockB \times \BlockC}$ are nonzero.
	We apply Lemma~\ref{lem:split-discon-orbits}:
	$\BlockB = \restrictVect{\BlockB}{\mathsf{A}} \times \restrictVect{\BlockB}{\mathsf{N}}$,
	$\BlockB' = \restrictVect{\BlockB'}{\mathsf{A}} \times \restrictVect{\BlockB'}{\mathsf{N}'}$,
	and likewise for~$\BlockA$ and~$\BlockA'$.	
	\begin{align*}
		(S\cdot T)(\tupA, \tupC)
		&= \sum_{\tupB \in \BlockB } S(\tupA,\tupB) \cdot T(\tupB, \tupC)\\
		&= \sum_{\tupB_\mathsf{A} \in \restrictVect{\BlockB}{\mathsf{A}} }
		\sum_{\tupB_{\mathsf{N}} \in \restrictVect{\BlockB}{\mathsf{N}} } S(\tupA_\mathsf{A}\tupA_{\mathsf{N}},\tupB_\mathsf{A}\tupB_{\mathsf{N}}) \cdot T(\tupB_\mathsf{A}\tupB_{\mathsf{N}}, \tupC_\mathsf{A}\tupC_{\mathsf{N}}) \tag{$\star$}.
	\end{align*}
	From Lemma~\ref{lem:not-equal-components-implies-active} it
	follows that $\restrictVect{P}{\mathsf{N}} = \restrictVect{Q}{\mathsf{N}} = \restrictVect{R}{\mathsf{N}}$ (recall we assumed that the blocks $S_{\BlockA \times \BlockB}$ and
	$T_{\BlockB \times \BlockC}$ are nonzero),
	in particular $\tupA_{\mathsf{N}},\tupC_{\mathsf{N}} \in \restrictVect{Q}{\mathsf{N}}$.
	We use again that the ${\mathsf{N}}$-components are not in the active region of~$S$ and~$T$. We continue the equation~($\star$):
	\begin{align*}
		(\star) &= \sum_{\tupB_\mathsf{A} \in \restrictVect{\BlockB}{\mathsf{A}} }
		\sum_{\tupB_N \in \restrictVect{\BlockB}{\mathsf{N}} }
		\kroneckerEq{\tupA_{\mathsf{N}}}{\tupB_{\mathsf{N}}} \cdot S(\tupA_\mathsf{A}\tupA_{\mathsf{N}},\tupB_\mathsf{A}\tupA_{\mathsf{N}}) \cdot
		\kroneckerEq{\tupB_{\mathsf{N}}}{\tupC_{\mathsf{N}}} \cdot T(\tupB_\mathsf{A}\tupC_{\mathsf{N}}, \tupC_\mathsf{A}\tupC_{\mathsf{N}})\\
		&= \sum_{\tupB_\mathsf{A} \in \restrictVect{\BlockB}{\mathsf{A}} }
		\kroneckerEq{\tupA_{\mathsf{N}}}{\tupC_{\mathsf{N}}} \cdot S(\tupA_\mathsf{A}\tupA_{\mathsf{N}},\tupB_\mathsf{A}\tupA_{\mathsf{N}}) \cdot T(\tupB_\mathsf{A}\tupC_{\mathsf{N}}, \tupC_\mathsf{A}\tupC_{\mathsf{N}})\\
		&= \sum_{\tupB_\mathsf{A} \in \restrictVect{\BlockB'}{\mathsf{A}} }
		\kroneckerEq{\primeSub{\tupA}{\mathsf{N}'}}{\primeSub{\tupC}{\mathsf{N}'}} \cdot S(\tupA_\mathsf{A}\primeSub{\tupA}{\mathsf{N}'},\tupB_\mathsf{A}\primeSub{\tupA}{\mathsf{N}'}) \cdot T(\tupB_\mathsf{A}\primeSub{\tupC}{\mathsf{N}'}, \tupC_\mathsf{A}\primeSub{\tupC}{\mathsf{N}'})\\
		&=\sum_{\primeSub{\tupB}{\mathsf{A}} \in \restrictVect{\BlockB'}{\mathsf{A}} }
		\sum_{\primeSub{\tupB}{\mathsf{N}'} \in \restrictVect{\BlockB'}{\mathsf{N}'}}
		\kroneckerEq{\primeSub{\tupA}{\mathsf{N}'}}{\primeSub{\tupB}{\mathsf{N}'}} \cdot S(\tupA_\mathsf{A}\primeSub{\tupA}{\mathsf{N}'},\primeSub{\tupB}{\mathsf{A}}\primeSub{\tupA}{\mathsf{N}'}) \cdot
		\kroneckerEq{\primeSub{\tupB}{\mathsf{N}'}}{\primeSub{\tupC}{\mathsf{N}'}} \cdot T(\primeSub{\tupB}{\mathsf{A}}\primeSub{\tupC}{\mathsf{N}'}, \tupC_\mathsf{A}\primeSub{\tupC}{\mathsf{N}'}).
	\end{align*}
	We used, as already seen, $\restrictVect{\BlockB}{\mathsf{A}}=\restrictVect{\BlockB'}{\mathsf{A}}$.
	We also used $\activeRegionB{S\cdot T}{\BlockA} \subseteq \activeRegionB{S}{\BlockA} \cup \activeRegionB{T}{\BlockA}$ as shown for Condition~\ref{itm:active-region-contaied}.
	So~$\tupA_{\mathsf{N}}$ can be exchanged with~$\primeSub{\tupA}{\mathsf{N}'}$
	and~$\tupC_{\mathsf{N}}$ with~$\primeSub{\tupC}{\mathsf{N}'}$.
	In the next step we use that $\tupA_\mathsf{A} = \primeSub{\tupA}{\mathsf{A}}$ and $\tupC_\mathsf{A} = \primeSub{\tupC}{\mathsf{A}}$ (by assumption)
	and again that~$S$ and~$T$ are not active on the ${\mathsf{N}'}$-components.
	\begin{align*}
		(\star) &= \sum_{\primeSub{\tupB}{\mathsf{A}} \in \restrictVect{\BlockB'}{\mathsf{A}} }
		\sum_{\primeSub{\tupB}{\mathsf{N}} \in \restrictVect{\BlockB'}{\mathsf{N}'} } S(\primeSub{\tupA}{\mathsf{A}}\primeSub{\tupA}{\mathsf{N}'},\primeSub{\tupB}{\mathsf{A}}\primeSub{\tupB}{\mathsf{N}'}) \cdot T(\primeSub{\tupB}{\mathsf{A}}\primeSub{\tupB}{\mathsf{N}'}, \primeSub{\tupC}{\mathsf{A}}\primeSub{\tupC}{\mathsf{N}'}) \\
		&= \sum_{\tupB' \in \BlockB' } S(\tupA',\tupB') \cdot T(\tupB', \tupC')\\
		&= (S\cdot T)(\tupA', \tupC').
	\end{align*}
	If $S_{\BlockA \times \BlockB}$ or
	$T_{\BlockB \times \BlockC}$ is zero,
	then $S_{\BlockA' \times \BlockB'}$ or
	$T_{\BlockB' \times \BlockC'}$ is zero
	because $S(\tupA,\tupB) = S(\tupA',\tupB') = 0$ and likewise for~$T$.
	The claim follows
	because $(S\cdot T) _{\BlockA \times \BlockB} = \zeromat$
	and $(S\cdot T) _{\BlockA' \times \BlockB'} = \zeromat$.
\end{proof}

We now consider products $S \cdot T$
in the case that the active regions of~$S$ and~$T$ are disjoint.
Intuitively, our goal is to prove that then  $S \cdot T$
is given by~$S$ on the active region of~$S$
and by~$T$ on the active region of~$T$. 

\begin{lemma}
	\label{lem:orbit-diagonal-multiply}
	Let $f,g,h \colon E \to \ZZ_{2^\pot}$ pairwise not twist $\orig{\pTupA}$
	and $S,T$ be $\StructVA^k \times \StructVA^k$ matrices over~$\FF_2$.
	Let~$S$ be orbit-diagonal over $(\StructA_f,\pTupA)$ and $(\StructA_g,\pTupA)$,~$T$ be orbit-diagonal over $(\StructA_g,\pTupA)$ and $(\StructA_h,\pTupA)$, both be odd-filled,
	$\activeRegionIdx{f,g}{\pTupA}{S} \cap \activeRegionIdx{g,h}{\pTupA}{T} = \emptyset$,
	$\BlockA \in \orbs{k}{(\StructA_f,\pTupA)}$, $\BlockB \in \orbs{k}{(\StructA_g,\pTupA)}$, and $\BlockC \in \orbs{k}{(\StructA_h,\pTupA)}$
	be of the same type,
	and the components of~$\BlockA$ (and thus the components of~$\BlockB$ and~$\BlockC$) be partitioned
	into~$M$ and~$N$
	such that $\ncomp{\activeRegionIdx{f,g}{\pTupA}{S}}{\BlockA} \subseteq M$
	and  $\ncomp{\activeRegionIdx{g,h}{\pTupA}{T}}{\BlockB} \subseteq N$.
	\begin{enumerate}[label = (\alph*)]
		\item \label{itm:disjoint-active-region} For every $\tupA \in \BlockA$ and $\tupC \in \BlockC$
		it holds that \[(S\cdot T)(\tupA,\tupC) = S(\tupA_M\tupA_N,\tupC_M\tupA_N) \cdot
		T(\tupC_M\tupA_N,\tupC_M\tupC_N).\]
		\item \label{itm:sum-left-disjoint-active-region} If~$S$ is orbit-invariant over $(\StructA_f,\pTupA)$ and $(\StructA_g,\pTupA)$,
		then for every $\tupA \in \BlockA$ and $\tupC\in \BlockC$ it holds that
		\[\sum_{\primeSub{\tupA}{M} \in \restrictVect{\BlockA}{M}} (S\cdot T) (\primeSub{\tupA}{M}\tupA_N,\tupC_M\tupC_N) = T(\tupC_M\tupA_N, \tupC_M\tupC_N).\]
	\end{enumerate} 
\end{lemma}
\begin{proof}
	We first show Part~\ref{itm:disjoint-active-region}. Let $\tupA \in \BlockA$ and $\tupC \in \BlockC$.
	\begin{align*}
		 (S\cdot T)(\tupA,\tupC)
		&= \sum_{\tupB \in \BlockB} 
		S(\tupA_M\tupA_N,\tupB_M\tupB_N) \cdot 
		T(\tupB_M\tupB_N, \tupC_M\tupC_N)\\
		&= \sum_{\tupB \in \BlockB} 
		\kroneckerEq{\tupA_N}{\tupB_N} \cdot S(\tupA_M\tupA_N,\tupB_M\tupA_N) \cdot
		\kroneckerEq{\tupB_M}{\tupC_M}  \cdot T(\tupC_M\tupB_N, \tupC_M\tupC_N). \tag{$\star$}
	\end{align*}
	The last step uses that components of~$\tupA_N$ and~$\tupC_N$
	consist only of components not contained in $\activeRegion{S}$
	and likewise for~$\tupB_M$ and~$\tupC_M$.
	\begin{align*}
		(\star)
		&= \sum_{\substack{\tupB \in \BlockB,\\ \tupB= \tupC_M\tupA_N}} 
		S(\tupA_M\tupA_N,\tupC_M\tupA_N) \cdot
		T(\tupC_M\tupA_N, \tupC_M\tupC_N)\\
		&= S(\tupA_M\tupA_N,\tupC_M\tupA_N) \cdot 
		T(\tupC_M\tupA_N, \tupC_M\tupC_N). 
	\end{align*}
	For the last step, we have to argue that $\tupC_M\tupA_N \in \BlockB$.
	From Lemma~\ref{lem:split-discon-orbits} it follows that 
	$\BlockA = \restrictVect{\BlockA}{M} \times \restrictVect{\BlockA}{N}$,
	$\BlockB = \restrictVect{\BlockB}{M} \times \restrictVect{\BlockB}{N}$, and
	$\BlockC = \restrictVect{\BlockC}{M} \times \restrictVect{\BlockC}{N}$.
	Because~$S$ is not active on the components in~$N$
	and~$T$ is not active on the components in~$M$,
	it follows from Lemma~\ref{lem:not-equal-components-implies-active} that
	$\restrictVect{\BlockA}{N} = \restrictVect{\BlockB}{N}$ and that
	$\restrictVect{\BlockB}{M} = \restrictVect{\BlockC}{M}$ (the corresponding blocks of~$S$ and~$T$ are nonzero because~$S$ and~$T$ are odd-filled).
	Hence, $\tupC_M\tupA_N \in \BlockB$
	because $\tupC_M \in \restrictVect{\BlockC}{M}$ and $\tupA_N \in \restrictVect{\BlockA}{N}$.

	We now show Part~\ref{itm:sum-left-disjoint-active-region}.
	We apply Part~\ref{itm:disjoint-active-region}:
	\begin{align*}
		&\hiddenEq\sum_{\tupA_M' \in \restrictVect{\BlockA}{M}}
		(S\cdot T) (\primeSub{\tupA}{M}\tupA_N,\tupC_M\tupC_N) \\
		&= \sum_{\primeSub{\tupA}{M} \in \restrictVect{\BlockA}{M}}
		S(\primeSub{\tupA}{M}\tupA_N,\tupC_M\tupA_N) \cdot
		T(\tupC_M\tupA_N,\tupC_M\tupC_N)\\
		&= T(\tupC_M\tupA_N,\tupC_M\tupC_N) \cdot
		\sum_{\primeSub{\tupA}{M} \in \restrictVect{\BlockA}{M}}
		S(\primeSub{\tupA}{M}\tupA_N,\tupC_M\tupA_N).
	\end{align*}
	It suffices to show that the value of the sum is~$1$.
	We rewrite the sum using $\BlockA = \restrictVect{\BlockA}{M} \times \restrictVect{\BlockA}{N}$ (Lemma~\ref{lem:split-discon-orbits}):
	\begin{align*}
		\sum_{\primeSub{\tupA}{M} \in \restrictVect{\BlockA}{M}}
		S(\primeSub{\tupA}{M}\tupA_N,\tupC_M\tupA_N)
		= \sum_{\primeSub{\tupA}{M}\primeSub{\tupA}{N} \in \BlockA}
		S(\primeSub{\tupA}{M}\primeSub{\tupA}{N},\tupC_M\tupA_N) - 
		\sum_{\substack{\primeSub{\tupA}{M}\primeSub{\tupA}{N} \in \BlockA,\\\primeSub{\tupA}{N} \neq \tupA_N }}
		S(\primeSub{\tupA}{M}\primeSub{\tupA}{N},\tupC_M\tupA_N).
	\end{align*}
	In the right sum
	it always holds that $S(\primeSub{\tupA}{M}\primeSub{\tupA}{N},\tupC_M\tupA_N) = 0$ because $\primeSub{\tupA}{N} \neq \tupA_N$
	and~$N$ is not in the active region of~$S$.
	So the right sum is zero.
	Finally, the left sum $\sum_{\primeSub{\tupA}{M}\primeSub{\tupA}{N} \in \BlockA}
	S(\primeSub{\tupA}{M}\primeSub{\tupA}{N},\tupC_M\tupA_N)$
	sums over a column of~$S$ because~$S$ is orbit-diagonal.
	Because~$S$ is orbit-invariant and odd-filled,
	this sum is~$1$ by Lemma~\ref{lem:block-diagonal-odd-filled-orbit-invariant-columns}.
\end{proof}

Finally, we show the result of Lemma~\ref{lem:orbit-diagonal-multiply}\ref{itm:sum-left-disjoint-active-region}
for a
product of three matrices $S_1 \cdot S_2 \cdot S_3$.

\begin{lemma}
	\label{lem:orbit-diagonal-multiply-sum-middle-disjoint-active-region}
	Let $g_i \colon E \to \ZZ_{2^\pot}$ pairwise not twist $\orig{\pTupA}$
	for every $i \in [4]$.
	Let $S_i$ be an $\StructVA^k \times \StructVA^k$ matrix over~$\FF_2$
	that is odd-filled and both orbit-diagonal and orbit-invariant over $(\StructA_{g_i},\pTupA)$ and $(\StructA_{g_{i+1}},\pTupA)$ for every $i \in [3]$.
	If the active regions $\activeRegionIdx{g_i,g_{i+1}}{\pTupA}{S_i}$ are pairwise disjoint,
	then for all $k$-orbits $\BlockA_i \in \orbs{k}{(\StructA_{g_i},\pTupA)}$ of the same type for all $i\in[4]$,
	every partition of the components of the $\BlockA_i$ into~$M_1$,~$M_2$, and~$M_3$
	such that ${\ncomp{\activeRegionIdx{g_i,g_{i+1}}{\pTupA}{S_i}}{\BlockA_i} \subseteq M_i}$ for every $i \in [3]$,
	and every $\tupA\in \BlockA_1$ and $\tupC \in \BlockA_4$ it holds that
	\[\sum_{\tupA_{M_2}' \in \restrictVect{\BlockA_1}{M_2}}
	(S_1\cdot S_2 \cdot S_3) (\tupA_{M_1}\primeSub{\tupA}{M_2}\tupA_{M_3},\tupC_{M_1}\tupC_{M_2}\tupC_{M_3}) =
	(S_1 \cdot S_3) (\tupA_{M_1}\tupC_{M_2}\tupA_{M_3},\tupC_{M_1}\tupC_{M_2}\tupC_{M_3}).\]
\end{lemma}
\begin{proof}
	By Lemma~\ref{lem:orbit-invariant-mult},
	the matrix $(S_1 \cdot S_2)$ is orbit-diagonal and orbit-invariant over $(\StructA_{g_1},\pTupA)$ and $(\StructA_{g_3},\pTupA)$
	and the matrix $(S_2 \cdot S_3)$ is orbit-diagonal and orbit-invariant over $(\StructA_{g_2},\pTupA)$ and $(\StructA_{g_4},\pTupA)$.
	Both matrices are odd-filled by Lemma~\ref{lem:odd-filled-product}.
	We apply Lemma~\ref{lem:orbit-diagonal-multiply}\ref{itm:disjoint-active-region}
	for the partition of the components of $\BlockA$ into $M_1 \cup M_3$ and~$M_2$:
	\begin{align*}
		&\hiddenEq\sum_{\primeSub{\tupA}{M_2} \in \restrictVect{\BlockA_1}{M_2}}
		(S_1\cdot S_2 \cdot S_3) (\tupA_{M_1}\primeSub{\tupA}{M_2}\tupA_{M_3},\tupC_{M_1}\tupC_{M_2}\tupC_{M_3})\\
		&= \sum_{\primeSub{\tupA}{M_2} \in \restrictVect{\BlockA_1}{M_2}}
		S_1(\tupA_{M_1}\primeSub{\tupA}{M_2}\tupA_{M_3},\tupC_{M_1}\primeSub{\tupA}{M_2}\tupA_{M_3}) \cdot
		(S_2 \cdot S_3)(\tupC_{M_1}\primeSub{\tupA}{M_2}\tupA_{M_3},\tupC_{M_1}\tupC_{M_2}\tupC_{M_3})\\
		&= \sum_{\primeSub{\tupA}{M_2} \in \restrictVect{\BlockA_1}{M_2}}
		S_1(\tupA_{M_1}\tupC_{M_2}\tupA_{M_3},\tupC_{M_1}\tupC_{M_2}\tupA_{M_3}) \cdot
		(S_2 \cdot S_3)(\tupC_{M_1}\primeSub{\tupA}{M_2}\tupA_{M_3},\tupC_{M_1}\tupC_{M_2}\tupC_{M_3}). \tag{$\star$}
	\end{align*}
	The last step uses that~$M_2$ consists only of components not contained in $\activeRegionIdx{g_1,g_2}{\pTupA}{S_1}$.
	We continue the equation by moving~$S_1$ out of the sum and applying Lemma~\ref{lem:orbit-diagonal-multiply}\ref{itm:sum-left-disjoint-active-region} for the partition of the components of~$\BlockA$ into $M_1 \cup M_3$ and~$M_2$:
	\begin{align*}
		(\star) &=
		S_1(\tupA_{M_1}\tupC_{M_2}\tupA_{M_3},\tupC_{M_1}\tupC_{M_2}\tupA_{M_3}) \cdot
		\sum_{\tupA_{M_2}' \in \restrictVect{\BlockA}{M_2}}
		(S_2 \cdot S_3)(\tupC_{M_1}\tupA_{M_2}'\tupA_{M_3},\tupC_{M_1}\tupC_{M_2}\tupC_{M_3})\\
		&=
		S_1(\tupA_{M_1}\tupC_{M_2}\tupA_{M_3},\tupC_{M_1}\tupC_{M_2}\tupA_{M_3}) \cdot
		S_3(\tupC_{M_1}\tupC_{M_2}\tupA_{M_3},\tupC_{M_1}\tupC_{M_2}\tupC_{M_3})\\
		&=(S_1 \cdot S_3) (\tupA_{M_1}\tupC_{M_2}\tupA_{M_3},\tupC_{M_1}\tupC_{M_2}\tupC_{M_3}).
	\end{align*}
	The last step follows from applying Lemmas~\ref{lem:active-region-change-equally} and~\ref{lem:orbit-diagonal-multiply}\ref{itm:disjoint-active-region}
	in the reverse direction by partitioning the components into~$M_1$ and $M_2\cup M_3$.
\end{proof}

\section{The Arity \texorpdfstring{$k$}{k} Case}
\label{sec:higher-arities}
We now construct a similarity matrix for the $k$-ary invertible-map game.
Constructing this matrix and verifying its suitability will be quite technical and intricate.
We first discuss the difficulties we have to overcome and why the approach for arity~$1$ cannot be generalized to arity~$k$ easily.
In the following, we provide high-level intuition for constructing the similarity matrix for arity~$k$.
This prepares us for the lengthy formal definition of this matrix,
which follows subsequently.

\subsection{Overview of the Construction}

\paragraph{Orbits of the Same Type.}
Let~$\StructA_f$ and~$\StructA_g$ be two CFI structures,
such that a single edge $\set{\stip,\stip'}$ of the base graph~$G$ is twisted by~$f$ and~$g$.
Let~$\pTupA$ be parameters, whose origin has sufficiently large distance to the twisted edge.
We have seen in Section~\ref{sec:1-ary-case}
that every $1$\nobreakdash-orbit has the same type in $(\StructA_f, \pTupA)$
as it has in $(\StructA_g,\pTupA)$.
For a $k$-orbit $\BlockA$, this is not the case
whenever $\set{\stip,\stip'}\subseteq \orig{\BlockA}$.
Ultimately, our goal is to construct an orbit-invariant, orbit-diagonal, and odd-filled similarity matrix~$S$ that $k$-blurs the twist.
Because the blocks on the diagonal of~$S$ arise from orbits of the same type
and because the characteristic matrices of orbits of the same type have to  be simultaneously similar,
we first want to define a bijection
$\orbs{k'}{(\StructA_f,\pTupA)} \to \orbs{k'}{(\StructA_g, \pTupA)}$
for every $k' \leq 2k$
that preserves the orbit types.
For this, we want to construct a function $\fixTwistType\colon \StructVA^{\leq 2k} \to \StructVA^{\leq 2k}$
that preserves the type of tuples.
Then $\fixTwistType$ preserves orbit types, too.
To do so, we pick
a vertex $\scenter$ satisfying $\distance{G}{\stip}{\scenter} > 2k$
and a  path $(\scenter,\dots, \stip',\stip)$.
We consider the path-isomorphism~$\autoA_\fixTwistType$
that twists the edge $\set{\stip,\stip'}$ and the edge incident to~$\scenter$ in the chosen path.
That is, between $\autoA_\fixTwistType(\StructA_f)$ and~$\StructA_g$
an edge incident to~$\scenter$ is twisted but the edge $\set{\stip,\stip'}$ is not.
For the moment assume that we only consider connected tuples and thus only connected orbits.
Let~$\fixTwistType$
be the function that applies the path-isomorphism~$\autoA_\fixTwistType$ to every tuple~$\tupA$
with $\set{\stip,\stip'} \subseteq \orig{\tupA}$ and is the identity function on all others.
Let $\tupA \subseteq \StructVA^{\leq 2k}$ be such a tuple with $\set{\stip,\stip'} \subseteq \orig{\tupA}$.
Because $\distance{G}{\stip}{\scenter} > 2k$
and because we consider connected tuples, we have that $\scenter \notin \orig{\tupA}$.
Hence,
$\StructA_g[\orig{\tupA}] = \autoA_\fixTwistType(\StructA_f)[\orig{\tupA}]$
and~$\tupA$ has the same type in $(\StructA_f,\pTupA)$ 
as~$\fixTwistType(\tupA)$ has in $(\StructA_g,\pTupA)$.
Consequently, for every  $k' \leq 2k$ and $\BlockA \in \orbs{k'}{(\Struct_f,\pTupA)}$
it holds that $\fixTwistType(\BlockA) \in \orbs{k'}{(\StructA_g,\pTupA)}$
and $\fixTwistType(\BlockA)$ has the same type in $(\StructA_g,\pTupA)$ as~$\BlockA$ has in $(\StructA_f,\pTupA)$.

\paragraph{Generalized Blurrers.}
Next we transfer the concept of a blurrer to the $k$-ary case.
Definition~\ref{def:blurrer-1-dim} of a $(\pot,d)$-blurrer
requires that there seems to be a twist at index~$1$ but none at the others indices
when considering only one of the~$d$ entries of the blurrer elements.
Although, all tuples $\blurElem$ in a blurrer satisfy $\sum \blurElem = 0$.
We require the same property in the $k$-ary case,
but now not only consider one index at a time but sets of~$k$ many indices.
We will generalize $(\pot,d)$-blurrers to $(k,\pot,a,d)$-blurrers,
where~$k$ is the arity,~$\pot$ specifies the ring~$\ZZ_{2^\pot}$,~$d$ the length of the tuples in the blurrer, 
and $a \in \ZZ_{2^\pot}$ the value of the twist (which was fixed to $2^{\pot-1}$ before).
Showing the existence of such blurrers will be more difficult,
in particular we will have to use, for a given~$k$, the ring~$\ZZ_{2^\pot}$
for a sufficiently large $\pot = \pot(k)$.

In the $1$-ary case,
we identified a tuple $\blurElem \in \blurrer$
with a local automorphism of the gadget of~$\scenter$.
We now describe the approach in the $k$-ary case.
Assume we are given a generalized $(k,\pot,a,d)$-blurrer~$\blurrer$ for arity~$k$
for some suitable~$\pot$,~$a \in \ZZ_{2^\pot}$, and~$d$.
We now require that the base graph~$G$ is regular of degree~$d$.
Recall that in Section~\ref{sec:1-ary-case}
we blurred the twist between the edges incident to~$\scenter$,
of which one was the twisted edge:
We used multiple local automorphisms (one for each $\blurElem \in \blurrer$)
to distribute the twist among these edges.
When considering connected $2k$\nobreakdash\nobreakdash-tuples,
we want to ensure that the origin of every $2k$\nobreakdash-tuple contains
at most one of the edges between which we blur the twist.
So it is not possible to blur the twist between the incident edges of a single vertex.
Instead, we will choose vertices $\stip_1, \dots, \stip_d$ and $\primeSub{\stip}{1}, \dots, \primeSub{\stip}{d}$,
such that $\stip = \stip_1$,  $\stip' = \primeSub{\stip}{1}$, 
and such that there are simple paths $\spath_i= (\scenter, \dots, \primeSub{\stip}{i}, \stip_i)$ of length at least~$2k$
forming a star,
i.e.,~the paths~$\spath_i$ are disjoint apart from~$\scenter$ (cf.~Figure~\ref{fig:star}).
Here it will be important to choose~$\spath_1$ to be the path we used to define
the tuple-type-preserving map~$\fixTwistType$ in the previous paragraph.
We will ensure that such paths exist by requiring that the girth of~$G$ is large enough. 
We will blur the twist between the edges $\set{\stip_i,\primeSub{\stip}{i}}$.
In the $1$-ary case, an element in a blurrer corresponded to an automorphism of the gadget of~$\scenter$, or equivalently to a star-isomorphism, where the paths of the star have length~$1$.
In the $k$-ary case we will identify a $\blurElem \in \blurrer$ with the
star-isomorphism $\autoA_\blurElem := \stariso{\blurElem}{\spath_1,\dots,\spath_d}$.
Again to preserve the type of tuples,
we will only apply~$\blurElem$ to tuples~$\tupA$ satisfying $\set{\stip_i,\primeSub{\stip}{i}}\not\subseteq\orig{\tupA}$ for all $i \in [d]$.
That is, on such a~$\tupA$, the action of~$\blurElem$ could also be defined by an automorphism.
This turns~$\blurElem$ into a ``star-automorphism''.
Using a star in combination with the large girth ensures that the tips of the star, the edges $\set{\stip_i,\primeSub{\stip}{i}}$,
are sufficiently far apart.
If we only had to deal with connected tuples,
this approach would be sufficient to construct a similarity matrix
(and in particular, we could even use easier blurrers).
However, disconnected tuples complicate matters.

\paragraph{Disconnected Tuples and Orbits.}
We have to consider disconnected tuples and orbits.
While with connected tuples the approach just described is local
(we only considered the $2k$-neighborhood of~$\scenter$),
there are disconnected  tuples containing vertices scattered in the structure.
But these vertices belong to different components of the tuple (cf.~Definition~\ref{def:component}).
Lemma~\ref{lem:split-discon-orbits} tells us
that the components of disconnected orbits are independent
whenever the connectivity of~$G$ is sufficiently large.
In a first step, we will salvage the previous approach by applying
the path- and the star-isomorphism not to entire tuples,
but instead to components of tuples.
That is, if a component~$\comp$ contains the twisted edge $\set{\stip, \stip'} = \set{\stip_1, \primeSub{\stip}{1}}$,
then we apply the type-preserving map~$\fixTwistType$ to this component.
If $\set{\stip_i,\primeSub{\stip}{i}} \not\subseteq \comp$ for all $i\in[d]$,
i.e,~$\comp$ contains none of the edges between which we blur the twist,
we apply the star-automorphisms~$\blurElem$ to this component.
(Note that~$\blurElem$ is the identity map unless~$\comp$ intersects non-trivially with at least one path~$\spath_i$.)

This approach fails, when for a $2k$-orbit~$\BlockA$ 
the two $k$-orbits  $\BlockA_1 := \restrictVect{\BlockA}{\set{1,\dots, k}}$
and ${\BlockA_2 := \restrictVect{\BlockA}{\set{k+1,\dots, 2k}}}$ contain the center~$\scenter$ of the star 
and some of the edges $\set{\stip_i,\primeSub{\stip}{i}}$ in their origin.
Because the edges $\set{\stip_i,\primeSub{\stip}{i}}$ are contained in the origin,
we need to argue with the blurrer properties to show that we blur the twist.
This is only possible if for two $k$\nobreakdash-tuples $\tupA \in \BlockA_1$
and  $\tupB \in \BlockA_2$
it only depends on up to~$k$ indices of a $\blurElem \in \blurrer$ whether
$\blurElem(\tupA)\tupB$ is in the same orbit as~$\tupA\tupB$.
But because the center~$\scenter$ is in the origin, this actually depends on all~$d$ entries of~$\blurElem$
and the blurrer properties do not apply.
This is why we will have to distinguish two kinds of $k$-orbits.
We call a $k$-orbit $\BlockA$ \defining{blurrable},
if $\scenter \notin \orig{\BlockA}$.
For non-blurrable orbits, we need another technique as follows.

\paragraph{Recursive Blurring.}
Now consider a $2k$-orbit $\BlockA$,
such that both $\BlockA_1 := \restrictVect{\BlockA}{\set{1,\dots, k}}$
and $\BlockA_2 := \restrictVect{\BlockA}{\set{k+1,\dots, 2k}}$ are non-blurrable
$k$-orbits.
Let us quickly recall the $1$-ary case.
It was possible to blur the twist in Lemma~\ref{lem:1-ary-matrices-similar}
because we summed over the tuples $\blurElem(\tupA)\tupB$ for all $\blurElem \in \blurrer$.
For a $2$-orbit~$\BlockA$, whose origin was the twisted edge $\set{\scenter,\stip}$,
w.l.o.g.~the origin of~$\BlockA_2$ is~$\scenter$ and
for every $\vertB \in \BlockA$ it held that $\blurElem(\vertB) \neq \blurElem'(\vertB)$ for every $\blurElem \neq \blurElem'$ in the blurrer.
But for~$\BlockA_1$ only one index of the blurrer was relevant.
That is, for every  $\vertA \in \BlockA_1$, $\vertB \in \BlockA_2$,
and $\blurElem, \blurElem' \in \blurrer$ such that $\blurElem(1) = \blurElem'(1)$
we had that
$\blurElem(\vertA) = \blurElem'(\vertA)$
and that $\blurElem(\vertA)\vertB$ is in the same orbit as $\blurElem'(\vertA)\vertB$.
So were able to apply the properties of a blurrer,
i.e., when summing over $\blurElem(\vertA)\vertB$ for all $\blurElem \in \blurrer$
and if only one index matters, then  the twist vanishes.
The $2$-orbits for which both~$\BlockA_1$ and~$\BlockA_2$
have origin $\set{\scenter}$ did not cover the twisted edge
and so did not pose a problem in the $1$-ary case.

Now consider the $k$-ary case again.
Here of course there are orbits~$\BlockA$
such that both~$\BlockA_1$ and~$\BlockA_2$ are non-blurrable
and they contain the twisted edge and the center~$\scenter$ in their origins.
Let $\tupA \in \BlockA_1$ and $\tupB \in \BlockA_2$.
Both~$\tupA$ and~$\tupB$ contain a vertex with origin~$\scenter$
and the blurrer properties do not apply
because the orbit of $\blurElem(\tupA)\tupB$
is different for every $\blurElem \in \blurrer$ 
(fixing one vertex of origin~$\scenter$ separates the gadget of~$\scenter$ into singleton orbits).

That is, when summing over all $\blurElem \in \blurrer$,
we map every $\tupA\tupB$ to the tuple $\blurElem(\tupA)\tupB$,
whose type in $(\autoA_\blurElem(\StructA_g), \pTupA)$ 
is the same as the type of $\tupA\tupB$ in $(\StructA_f, \pTupA)$.
But in $(\StructA_g, \pTupA)$ the tuple $\blurElem(\tupA)\tupB$ has a different type.
Between $(\StructA_g, \pTupA)$ and $(\autoA_\blurElem(\StructA_g), \pTupA)$
the edges $\set{\stip_i,\primeSub{\stip}{i}}$ are twisted additionally (the values of the twists depend on~$\blurElem)$.
This, in some sense, introduces other twists,
but only for said $2k$-orbits~$\BlockA$,
where the origins of both~$\BlockA_1$ and~$\BlockA_2$
are not blurrable.

The idea is to fix an arbitrary vertex $\pVertA_\scenter$ with origin~$\scenter$
and consider $(k-1)$-orbits of $(\StructA,\pTupA\pVertA_\scenter)$
(justified by Corollary~\ref{cor:k-orbit-fix-vertex}).
This can be done because all non-blurrable orbits contain~$\scenter$ in its origin.
For every $\blurElem \in \blurrer$, we will recursively obtain  a matrix~$S^\blurElem$
that $(k-1)$-blurs the twist between $(\StructA_f,\pTupA\pVertA_\scenter)$
and $(\inv{\autoA_\blurElem}(\StructA_g),\pTupA\pVertA_\scenter)$.
We use  the inverse~$\inv{\autoA_\blurElem}$ of the star-isomorphism~$\autoA_\blurElem$ because we want to revert the twists introduced by~$\autoA_\blurElem$.
Here the need arises to blur a twist of value $a \neq 2^{\pot-1}$.
Combining the blurrer~$\blurrer$ with the matrices~$S^\blurElem$
to a matrix~$S$ that $k$-blurs the twist will become formally tedious.
In particular, we will need to ensure that the~$S^\blurElem$ act ``independently''
on the $\set{\stip_i,\primeSub{\stip}{i}}$,
which we discuss next.

\paragraph{Active Region and Blurrers.}
The matrix~$S$ is defined for blocks of $k$-orbits.
Blocks for blurrable $k$-orbits will be defined using the blurrer~$\blurrer$,
blocks for non-blurrable $k$-orbits will be defined using~$\blurrer$ and the matrices~$S^\blurElem$.
With this approach we will show that~$S$
is a similarity matrix for all orbits~$\BlockA$,
for which either~$\BlockA_1$ and~$\BlockA_2$
are both blurrable or~$\BlockA_1$ and~$\BlockA_2$ are both non-blurrable.
In the former case, we will use the blurrer property,
in the latter case, we will use induction.
The case that~$\BlockA_1$ is blurrable and~$\BlockA_2$ is not 
or vice versa remains.
We have to show that $\charMat{\BlockA} \cdot S = S \cdot \charMat{\BlockB}$
(for $\BlockB = \fixTwistType(\BlockA)$, which has the same type as~$\BlockA$).
Assume that~$\BlockA_1$ is blurrable and~$\BlockA_2$ is not.
For $S \cdot \charMat{\BlockB}$ solely the block $S_{\BlockA_1 \times \BlockB_1}$ of~$S$ is relevant. This block is defined using the blurrer~$\blurrer$
because~$\BlockA_1$ is blurrable.
Similarly, for $\charMat{\BlockA} \cdot S$ solely the block $S_{\BlockA_2\times\BlockB_2}$ is relevant.
This block is defined using the blurrer~$\blurrer$ and the matrices~$S^\blurElem$
because~$\BlockA_2$ is non-blurrable.

To use the blurrer properties also for~$\BlockA_2$,
we will define the matrices~$S^\blurElem$,
which blur multiple twists at the edges $\set{\stip_i, \stip_i'}$,
as $S^\blurElem := \Sblurelem{\blurElem}{1}\cdot \ldots \cdot \Sblurelem{\blurElem}{d}$,
where each $\Sblurelem{\blurElem}{i}$ only blurs a single twist at the edge $\set{\stip_i,\primeSub{\stip}{i}}$.
We will ensure that the active region of $\Sblurelem{\blurElem}{i}$
is bound by the $\minRad(k)$-neighborhood of~$\stip_i$
for some suitable~$\minRad(k)$.
We then enlarge the star such that the paths~$\spath_i$
have length greater than $\max\set{2k, \minRad(k)}$.
Now, the active regions of the~$\Sblurelem{\blurElem}{i}$
are disjoint and we can use Lemma~\ref{lem:orbit-diagonal-multiply-sum-middle-disjoint-active-region}
to show that indeed all except~$k$ many of the $\Sblurelem{\blurElem}{i}$ cancel out.
So finally, we can use the blurrer properties to show that~$S$ is a similarity matrix
for orbits where~$\BlockA_1$ is blurrable and~$\BlockA_2$ is not.  
We now start with generalizing blurrers
and then show the existence of the required similarity matrix.

\subsection{Blurrer}
When dealing with arity~$k$,
the properties of a blurrer must be generalized from a single index
to sets of indices of size at most~$k$.
Let $\pot,d \in \nat$ and  $\blurrer \subseteq \ZZ_{2^\pot}^{d}$.
For $K \subseteq [d]$ and $\tup{b} \in \ZZ_{2^\pot}^{|K|}$
we count the tuples contained in~$\blurrer$
whose restriction to~$K$ equals~$\tup{b}$.
We define
\[\countVects{K}{\tup{b}}{\blurrer} \coloneqq \left|\setcond*{\tup{c} \in \blurrer}{\restrictVect{\tup{c}}{K}= \tup{b}}\right| \bmod 2.\]
\begin{definition}[Blurrer]
	\label{def:blurrer-kary}
	Let $d\geq k$, $\blurrer \subseteq \ZZ_{2^\pot}^{d}$, and $a \in \ZZ_{2^\pot}$.
	The set $\blurrer \subseteq \ZZ_{2^\pot}^{d}$ is called a \defining{$(k,\pot,a,d)$-blurrer}
	if it satisfies the following for all $K \subseteq [d]$ with $|K| = k$:
	\begin{enumerate}
		\item \label{itm:blurrer-sum} $\sum \blurElem = 0$ for all $\blurElem \in \blurrer$.
		\item \label{itm:blurrer-twist-first}
		If $1 \in K$, then 
		$\countVects{K}{(a,0,\dots,0)}{\blurrer} = 1$.
		\item \label{itm:blurrer-no-twist-other} If $1 \notin K$, then
		$\countVects{K}{\tup{0}}{\blurrer} = 1$.
		\item \label{itm:blurrer-remaining} $\countVects{K}{\tup{b}}{\blurrer} = 0$
		for all other pairs of~$K$ and~$\tup{b}$.
	\end{enumerate}
\end{definition}
\noindent The crucial property of a blurrer is the following:
\begin{lemma}
	\label{lem:blurrer-sum}
	Let $\blurrer$ be a $(k,\pot,a,d)$-blurrer,
	$K \subseteq [d]$ such that $|K| = k$,
	and define ${\blurElemFix := (a,0,\dots, 0) \in \ZZ_{2^\pot}^d}$.
	Every function $f\colon \restrictVect{\blurrer}{K} \to \FF_2$ satisfies 
	\[\sum_{\blurElem \in \blurrer} f(\restrictVect{\blurElem}{K}) =
	f(\restrictVect{\blurElemFix}{K}).\]
	In particular, there is a $\sblurElem \in \blurrer$
	such that $\restrictVect{\sblurElem}{K} = \restrictVect{\blurElemFix}{K}$.
\end{lemma}
\begin{proof}
	Because~$f$ takes $k$\nobreakdash-tuples as input,
	it cannot distinguish whether it is applied to
	$\restrictVect{\blurElem}{K}$ or to 
	$\restrictVect{\blurElemFix}{K}$ because
	by Conditions~\ref{itm:blurrer-twist-first}, and~\ref{itm:blurrer-no-twist-other}
	$\restrictVect{\blurElemFix}{K} \in \restrictVect{\blurrer}{K}$.
	Conditions~\ref{itm:blurrer-twist-first},~\ref{itm:blurrer-no-twist-other}, and~\ref{itm:blurrer-remaining} ensure that
	when summing over all $\blurElem \in \blurrer$,
	all summands $f(\restrictVect{\blurElem}{K})$ apart from
	$f(\restrictVect{\blurElemFix}{K})$ cancel out
	(by Condition~\ref{itm:blurrer-remaining} and~\ref{itm:blurrer-twist-first} if $1 \in K$ or Condition~\ref{itm:blurrer-no-twist-other} if $1 \notin K$).
	The existence of a $\sblurElem \in \blurrer$ as required
	follows from Conditions~\ref{itm:blurrer-twist-first} and~\ref{itm:blurrer-no-twist-other}.
\end{proof}
Note that while~$\blurElem$ only contains tuples
satisfying $\sum \blurElem = 0$,
we have that $\sum \blurElemFix = a$.

\begin{lemma}
	\label{lem:blurrer-size-odd}
	Let~$\blurrer$ be a $(k,\pot,a,d)$-blurrer.
	Then~$|\blurrer|$ is odd.
\end{lemma}
\begin{proof}
	Let $K \subseteq [d]$ with $|K| = k$.
	We partition $\blurrer = M \cup N$ into
	\begin{align*}
		M &:= \setcond[\big]{\blurElem \in \blurrer}{\restrictVect{\blurElem}{K} = \restrictVect{\blurElemFix}{K} },\\
		N &:= \setcond[\big]{\blurElem \in \blurrer}{\restrictVect{\blurElem}{K} \neq  \restrictVect{\blurElemFix}{K}},
	\end{align*}
	where $\blurElemFix:=(a, 0, \dots, 0)$ is the tuple from Lemma~\ref{lem:blurrer-sum}.
	The size of~$|M|$ is odd 
	by Condition~\ref{itm:blurrer-twist-first} if $1 \in K$
	and otherwise by Condition~\ref{itm:blurrer-no-twist-other}.
	By Condition~\ref{itm:blurrer-remaining}, the size~$|N|$ is even.
	If it was odd, then some~$\tup{b}$ would violate Condition~\ref{itm:blurrer-remaining}.
\end{proof}

We now construct blurrers.
\begin{lemma}\label{lem:blurrer-basics}
	If there is a $(k,\pot,a,d)$-blurrer $\blurrer$, then
	\begin{enumerate}
		\item there is a $(k,\pot,a,d')$-blurrer for every $d' \geq d$,
		\item $\blurrer$ is a $(k', \pot,a,d)$-blurrer for every $k' \leq k$, and
		\item there is a $(k,\pot, c\cdot a,d)$-blurrer for every $c \in \ZZ_{2^\pot}$.
	\end{enumerate}
\end{lemma}
\begin{proof}
	To prove the first statement,
	we just fill up the tuples of~$\blurrer$ with zeros to be of length~$d'$.
	To prove the second statement, let $K'\subseteq K \subseteq [d]$ such that  $|K|= k$
	and let $\tup{b}' \in \ZZ_{2^\pot}^{|K'|}$.
	Then 
	\[\countVects{K'}{\tup{b'}}{\blurrer} =\sum_{\substack{\tup{b} \in \ZZ_{2^\pot}^k,\\ \restrictVect{\tup{b}}{K'} = \tup{b}'}}
	\countVects{K}{\tup{b}}{\blurrer}.\]
	Assume $1 \in K'$ and $\tup{b}' = (a,0,\dots,0)$.
	Then for $\tup{b} = (a,0,\dots,0) \in \ZZ_{2^\pot}^k$
	we have $\restrictVect{\tup{b}}{K'} =\tup{b}'$ and
	$\countVects{K}{\tup{b}}{\blurrer} = 1$ by Condition~\ref{itm:blurrer-twist-first}.
	For all other~$\tup{b}$ we have $\countVects{K}{\tup{b}}{\blurrer} = 0$ by Condition~\ref{itm:blurrer-remaining}.
	So the sum is~$1$.
	The case that $1 \notin K'$ and $\tup{b} =\tup{0}$ is similar using Condition~\ref{itm:blurrer-no-twist-other}.
	In the remaining case all summands are $0$ by blurrer Condition~\ref{itm:blurrer-remaining}.
	
	To prove the last statement, let $c \in \ZZ_{2^\pot}$ and set $\blurrer' := \setcond{c \cdot \blurElem}{\blurElem \in \blurrer}$.
	If $\sum \blurElem = 0$, then clearly $\sum c \cdot \blurElem = 0$.
	We verify blurrer Condition~\ref{itm:blurrer-twist-first}, the others are similar.
	Let $K \subseteq [d]$ of size~$k$ and $\tup{b} = (a, 0, \dots,0) \in \ZZ_{2^\pot}^k$.
	From Conditions~\ref{itm:blurrer-twist-first} and~\ref{itm:blurrer-no-twist-other} it follows that 
	\[
	\countVects{K}{c \cdot \tup{b}}{\blurrer'} = \countVects{K}{\tup{b}}{\blurrer} + \sum_{
		\substack{\tup{b}' \in \ZZ_{2^\pot}^k,\\
			\tup{b'} \neq \tup{b}, \\ c\cdot \tup{b}' =  c\cdot \tup{b}}} \countVects{K}{\tup{b}'}{\blurrer} = 1 +\sum_{
		\substack{\tup{b}' \in \ZZ_{2^\pot}^k,\\
		\tup{b'} \neq \tup{b}, \\ c\cdot \tup{b}' =  c\cdot \tup{b}}} 0 = 1.\vspace{-21pt}
	\]
\end{proof}

\begin{lemma}
	\label{lem:binoim-parity}
	Let $m,n\in \nat$. If $0<m < 2^n$, then $\binom{2^n}{m}$ is even.
	If $m \leq 2^n -1$, then $\binom{2^n-1}{m}$ is odd.
\end{lemma}
\begin{proof}
	Let $k \in \nat$ and consider~$\binom{k}{m}$. We write~$k$ and~$m$
	in base~$2$ representation
	\begin{align*}
		m &= \sum_{i=0}^j m_i 2^i& k &= \sum_{i=0}^j k_i2^i
	\end{align*}
	for some suitable~$j$ and $m_i,k_i \in \set{0,1}$ for all $i \in \set{0,\dots,j}$.
	We apply Lucas's Theorem~\cite{Cameron1994}:
	\[\binom{k}{m} \bmod 2 = \prod_{i=0}^j \binom{k_i}{m_i} \bmod 2, \]
	where $\binom{0}{1} = 0$
	and $\binom{0}{0} = \binom{1}{0} = \binom{1}{1} = 1$.
	That is, $\binom{k}{m} \bmod 2 = 0$ if and only if
	there is an $i \in \set{0,\dots,j}$ such that $k_i = 0$ and $m_i = 1$.
	
	\begin{enumerate}[label=(\alph*)]
		\item Let $k = 2^n$ and $0<m < k$.
		Then there is an $i < n$ such that $m_i = 1$ .
		Because $k = 2^n$, $k_i = 0$ and so~$\binom{k}{m}$ is even.
	
		\item Let $k = 2^n -1$ and $m \leq k$.
		For every $i < n$ we have $k_i = 1$.
		For every $i$ such that $m_i = 1$
		it holds that $i < n$.
		That is,~$\binom{k}{m}$ is odd. \qedhere
	\end{enumerate}
\end{proof}

\begin{lemma}
	\label{lem:basic-blurrer}
	For every $i \in \nat$, there is a $(2^{i-1}-1,i,2^{i-1},2^i-1)$-blurrer.
\end{lemma}
\begin{proof}
	We set $k := 2^{i-1}-1$, $d:=2^i-1$,
	and define $\blurrer$ as follows:
	\begin{align*}
		\blurrer_1 &:= \setcond[\Bigg]{(2^{i-1}, a_2,\dots,a_{2^i-1})}{\sum_{j=2}^{2^i-1} a_j = 2^{i-1}, a_j \in \set{0,1} \text{ for every } j } ,  \\
		\blurrer_2 &:=  \setcond[\Bigg]{(2^{i-1}+1, a_2,\dots,a_{2^i-1})}{\sum_{j=2}^{2^i-1} a_j = 2^{i-1}-1, a_j \in \set{0,1} \text{ for every } j },\\
		\blurrer_{\phantom{1}} &:= \blurrer_1 \cup \blurrer_2.  
	\end{align*}
	To verify that~$\blurrer$ is indeed a $(k,i,2^{i-1}, d)$-blurrer,
	let $K \subseteq [d]$ be of size~$k$ and
	$\tup{b} \in \ZZ_{2^i}^k$.
	Set $\overline{K} := [d]\setminus K$.
	\begin{itemize}
		
		\item Let $1 \in K$ and $\tup{b} = (2^{i-1}, 0 , \dots, 0)$.
		Every $\blurElem \in\blurrer$ with $\restrictVect{\blurElem}{K} = \tup{b}$ is contained in~$\blurrer_1$.
		Because every $\blurElem \in\blurrer_1$ contains~$2^{i-1}$ many $1$-entries,~$\blurElem$ is of length $d = 2^i-1$,
		and~$\tup{b}$ contains $k-1 = 2^{i-1}-2$ many $0$-entries,
		every $\blurElem \in\blurrer_1$ such that $\restrictVect{\blurElem}{K} = \tup{b}$ satisfies $\restrictVect{\blurElem}{\overline{K}} = (1, \dots, 1)$.
		So there can be at most one such $\blurElem \in\blurrer_1$.
		It exists by construction of~$\blurrer_1$.
		Hence, $\countVects{K}{\tup{b}}{\blurrer} = 1$.
		
		\item Let $1 \in K$ and $\tup{b} = (2^{i-1}, b_2 , \dots, b_{2^i-1})$, $b_j \in \set{0,1}$ for all~$j$,
		and not all~$b_j$ equal zero.
		Again, every $\blurElem \in\blurrer$ such that $\restrictVect{\blurElem}{K} = \tup{b}$ is contained in~$\blurrer_1$.
		Because
		\[\sum \tup{b} \in \set*{2^{i-1}+1, \dots, 2^{i-1}+k-1}=\set*{2^{i-1}+1, \dots, 2^i - 2}\]
		it holds that
		\[m := - \sum \tup{b} = \sum \restrictVect{\blurElem}{\overline{K}} \in \set*{2,\dots, 2^{i-1} -1}\]
		for every $\blurElem \in \blurrer_1$.
		By construction,  $\restrictVect{\setcond{\blurElem\in \blurrer_1}{\restrictVect{\blurElem}{K} = \tup{b}}}{[2^i-1] \setminus K}$ is the set of all $0/1$\nobreakdash\nobreakdash-tuples of length~$2^{i-1}$
		that contain exactly~$m$ many ones.
		There are exactly~$\binom{2^{i-1}}{m}$ many $0/1$\nobreakdash-tuples of length $2^{i-1}$.
		For all possible values of~$m$, the number~$\binom{2^{i-1}}{m}$ is even by Lemma~\ref{lem:binoim-parity}.
		We conclude that $\countVects{K}{\tup{b}}{\blurrer} = 0$. 
		
		\item Let $1 \in K$ and $\tup{b} = (2^{i-1}+1, b_2 , \dots, b_{2^i-1})$.
		Now, every $\blurElem \in\blurrer$ with $\restrictVect{\blurElem}{K} = \tup{b}$ is contained in~$\blurrer_2$.
		Then 
		\[m := - \sum \tup{b} = \sum \restrictVect{\blurElem}{\overline{K}}\in \set*{1,\dots, 2^{i-1}-1}\]
		for every $\blurElem \in \blurrer_2$
		because
		\[\sum \tup{b} \in \set*{2^{i-1}+1, \dots, 2^{i-1}+1+(k-1)} = \set*{2^{i-1}+1, \dots, 2^i - 1}.\]
		Again, there are~$\binom{2^{i-1}}{m}$   many $0/1$\nobreakdash-tuples
		extending~$\tup{b}$ to a $\blurElem \in \blurrer$,
		which is an even number by Lemma~\ref{lem:binoim-parity} and thus 
		$\countVects{K}{\tup{b}}{\blurrer} = 0$.

		\item Let $1 \in K$ and~$\tup{b}$ not be covered by two cases before.
		Then there is no $\blurElem$ satisfying $\restrictVect{\blurElem}{K} = \tup{b}$ 
		and so $\countVects{K}{\tup{b}}{\blurrer} = 0$.
		\item Now the case that $1 \notin K$ remains.
		If there is no $\blurElem \in \blurrer$
		satisfying $\restrictVect{\blurElem}{K} = \tup{b}$,
		then clearly $\countVects{K}{\tup{b}}{\blurrer} = 0$.
		So assume that there is such a $\blurElem \in \blurrer$.
		\begin{itemize}
			\item
			Let us first consider~$\blurrer_1$.
			Let $\blurElem \in \blurrer_1$.
			Then
			\[0 = \sum{\blurElem} = \sum \tup{b} + \sum \restrictVect{\blurElem}{\overline{K}} = 
			\sum \tup{b} +
			2^{i-1} + \sum \restrictVect{\blurElem}{\overline{K}\setminus \set{1}}.\]
			The tuple $\restrictVect{\blurElem}{\overline{K}\setminus \set{1}}$
			is a 0/1-tuple of length $d-k-1 = 2^{i-1}-1$.
			So
			\[\sum \restrictVect{\blurElem}{\overline{K}\setminus \set{1}} \in \set*{0, \dots, 2^{i-1}-1}.\]
			That is, if $\sum \tup{b} = 0$,
			we obtain a contradiction because 
			there is no $\blurElem \in \blurrer_1$
			satisfying $\restrictVect{\blurElem}{K} = \tup{b}$
			and $\countVects{K}{\tup{b}}{\blurrer_1} = 0$.
			Otherwise,
			all $\blurElem \in \blurrer_1$ satisfying
			$\restrictVect{\blurElem}{K} = \tup{b}$
			extend~$\tup{b}$
			by a $0/1$-tuple of length~$2^{i-1}-1$
			containing
			\[m := - \sum \tup{b} - 2^{i-1} \in \set*{0, \dots, 2^{i-1}-1}\]
			many ones.
			There are~$\binom{2^{i-1}-1}{m}$  many $0/1$\nobreakdash-tuples
			of length $d-k-1 = 2^{i-1}-1$ and sum~$m$,
			which is an odd number by Lemma~\ref{lem:binoim-parity}.
			Hence, $\countVects{N}{\tup{b}}{\blurrer_1} = 1$.
			
			\item
			Now consider~$\blurrer_2$.
			For every $\blurElem \in \blurrer_2$ such that $\restrictVect{\blurElem}{K} = \tup{b}$,
			it similarly holds that
			\[0 = \sum{b} + 2^{i-1}+1 + \sum \restrictVect{\blurElem}{\overline{K}\setminus \set{1}}\]
			and thus $\sum \restrictVect{\blurElem}{\overline{K}\setminus \set{1}} \in \set{0, \dots, 2^{i-1}}$.
			Every $\blurElem \in \blurrer_1$ satisfying
			$\restrictVect{\blurElem}{K} = \tup{b}$
			extends~$\tup{b}$
			by a $0/1$-tuple of length $2^{i-1}-1$
			containing
			\[m := - \sum \tup{b} - 2^{i-1}-1 \in \set*{0, \dots, 2^{i-1}-1}\]
			many ones.
			The number $m$ is again odd by Lemma~\ref{lem:binoim-parity}.
			Hence, $\countVects{N}{\tup{b}}{\blurrer_2} = 1$.
		\end{itemize}
		Together, 
		if $\tup{b} = 0$,
		then $\countVects{K}{\tup{b}}{\blurrer} = \countVects{K}{\tup{b}}{\blurrer_1} + \countVects{K}{\tup{b}}{\blurrer_2}= 0 + 1 = 1$.
		Otherwise,
		$\countVects{K}{\tup{b}}{\blurrer_1} + \countVects{K}{\tup{b}}{\blurrer_2}= 1 + 1$ which is~$0$ modulo~$2$. \qedhere
	\end{itemize}
\end{proof}
\begin{remark}
Computer experiments suggest that for a given $2^{i-2} \leq k \leq 2^{i-1}-1$
our choice of $q=i$ is minimal to construct a $(k,q,2^{i-1},d)$-blurrer
and that $d=2^i-1$ could be improved in the case that $k \neq  2^{i-1}-1$,
but is minimal in the case that  $k=2^{i-1}-1$.
\end{remark}

We now lift a $(k,\pot,a,d)$-blurrer 
from the ring~$\ZZ_{2^{\pot}}$ to the ring~$\ZZ_{2^{\pot +\ell}}$.
Here we have two choices, both of which we need later:
the first is via the embedding of~$\ZZ_{2^\pot}$ in~$\ZZ_{2^{\pot +\ell}}$, 
the second will not change the value~$a$.

\begin{lemma} \label{lem:blurrer-embedding}
	Let $\pot,\ell \in \nat$
	and $\iota\colon\ZZ_{2^\pot} \to \ZZ_{2^{\pot+\ell}}$
	be the embedding of $\ZZ_{2^\pot}$ in $\ZZ_{2^{\pot+\ell}}$
	defined by $a \mapsto 2^\ell a$.
	If~$\blurrer$ is a $(k,\pot,a, d)$-blurrer,
	then $\iota(\blurrer)$ is a $(k, \pot + \ell, \iota(a),d)$-blurrer.
\end{lemma}
\begin{proof}
	This is straightforward from the definition.
\end{proof}

\begin{lemma}\label{lem:blurrer-larger-field}
	For every $i,\ell\in \nat$,
	there is a $(2^{i-1}-1,i+\ell,2^{i-1},2^i-1)$-blurrer.
\end{lemma}
\begin{proof}
	Let~$\blurrer$ be the $(2^{i-1}-1,i,2^{i-1}, 2^i-1)$-blurrer given by Lemma~\ref{lem:basic-blurrer} and suppose
	$c = 2^{\ell +1} - 1 \in \ZZ_{2^{i+\ell}}$.
	Let~$h$ be the following function that maps~$\blurrer$ to $\blurrer' := \setcond{h(\blurElem)}{\blurElem \in \blurrer}$:
	\[\blurElem \mapsto \left(- c \cdot \sum_{j=2}^d \blurElem_j, c \cdot \blurElem_2, \dots, c \cdot \blurElem _d\right). \]
 	The operations are all in~$\ZZ_{2^{i+\ell}}$.
	By definition $\sum \blurElem' = 0$ for every $\blurElem' \in \blurrer'$.
	Let $K\subseteq [d]$ be of size~$k$ and $\tup{b} \in \ZZ_{2^{i+\ell}}^k$.
	Note that~$c$ is a unit because~$c$ is odd.
	\begin{itemize}
		\item Let $1 \notin K$.
		Because~$c$ is a unit, multiplication with~$c$ is a bijection and thus we have
		$\countVects{K}{\tup{b}}{\blurrer'} = \countVects{K}{\inv{c} \cdot \tup{b}}{\blurrer}$,
		which is~$1$ if and only if $\inv{c} \cdot \tup{b} = \tup{b} = \tup{0}$.
		\item 
		Let $1\in K$. We argue that also the action of~$h$ on the first position is a bijection.
		Because $\sum \blurElem = 0$ for all $\blurElem \in \blurrer$,
		the map $\blurElem_1 \mapsto \sum_{j=2}^d \blurElem_j$ is a bijection
		and so is the action of~$h$ because~$c$ is a unit.
		So we have $\countVects{K}{\tup{b}}{\blurrer'} = \countVects{K}{\tup{a}}{\blurrer}$
		for some $\tup{a} \in \ZZ_{2^i}^k$.
		
		It holds that $- c a_1 =  -(2^{\ell  +1} - 1) a_1 = a_1 - 2^{\ell +1}a_1$ (over $\ZZ_{2^{i+\ell}}$).
		So $a_1 = 2^{i-1}$ if and only if  $-c a_1 = 2^{i-1} - 2^{\ell  +1}2^{i-1} = 2^{i-1} - 2^{\ell  +i} =  2^{i-1} -0 = 2^{i-1}$.
		Hence, $\tup{b} = (2^i, 0,\dots, 0)$ if and only if $\tup{a} = (2^i, 0,\dots, 0)$, which is the case if and only if $\countVects{K}{\tup{b}}{\blurrer'} = 1$.
		 \qedhere
	\end{itemize}
\end{proof}

\subsection{Similarity Matrix for One Round}

We now construct a similarity matrix $k$-blurring the twist.
To be able to define this matrix,
we need bounds on the degree, the girth, and the connectivity of the base graph
as well as certain guarantees for the placement of the pebbles. 
Therefore, we define the following functions $\nat^+ \to \nat^+$,
which will give us the needed bounds.
In the definitions let $i\in \nat$ 
be the unique number that for the given~$k$ satisfies
$2^{i-1}-1 < k \leq 2^i -1$.
\begin{align*}
	\minRad(k) &:= \begin{cases}
		\mathrlap{1}\phantom{\max\set{2^{i+1}+m-1, d(k-1,m+1)}} 			& \text{if } k = 1,\\
		\max\set{4 \cdot \minRad(k-1)+2, 2k+2}	&\text{otherwise},
	\end{cases}\\
	\twistVal(k) &:= \begin{cases}
		\mathrlap{1}\phantom{\max\set{2^{i+1}+m-1, d(k-1,m+1)}} 			& \text{if } k = 1,\\
		i + \twistVal(k-1)	&\text{otherwise},
	\end{cases}\\
	d(k,m) &:= \begin{cases}
		3+m			& \text{if } k = 1,\\
		\max\set{2^{i+1}+m-1, d(k-1,m+1)}&\text{otherwise},
	\end{cases}\\
	\pot(k) &:= 1 + \twistVal(k).
\end{align*}

\begin{lemma}
	\label{lem:blur-kary}
	For every $k,m \in \nat$,
	\begin{itemize}
		\item every regular and $(2k + m + 1)$-connected
	base graph $G=(V,E, \leq)$ of degree  ${d \geq d(k,m)}$ and
	girth at least $2\minRad(k+1)$,
	\item
	every edge $\set{t,t'} \in E$,
	\item every $\pot \geq \pot(k)$,
	\item every $\twistVal = a \cdot 2^{\twistVal(k)} \in \ZZ_{2^\pot}$ (for an arbitrary $a \in \ZZ_{2^\pot}$),
	\item every $f, g \colon E \to \ZZ_{2^\pot}$
	such that
	$f(e) = g(e)$ for all $e\in E\setminus{\set{\set{t,t'}}}$
	and $g(\set{t', t}) = f(\set{t', t}) + \twistVal$, and
	\item every $m$-tuple $\pTupA \in \StructVA^m$
	of $\StructA_f := \CFIgraph{2^\pot}{G}{f}$ and $\StructA_g := \CFIgraph{2^\pot}{G}{g}$, both with universe $\StructVA$,
	for which $\distance{G}{t}{\orig{\pTupA}} >  \minRad(k+1)$
\end{itemize}
	there is an odd-filled $\StructV^k \times \StructV^k$ matrix~$S$,
	both orbit-diagonal and orbit-invariant over $(\StructA_f, \pTupA)$ and $(\StructA_g, \pTupA)$,
	that $k$-blurs the twist between $(\StructA_f, \pTupA)$ and $(\StructA_g, \pTupA)$ and those active region satisfies $\activeRegionIdx{f,g}{\pTupA}{S} \subseteq \neighborsK{G}{\minRad(k+1)}{t}$.
\end{lemma}

The proof is by induction on~$k$ and spans the rest of this section.
We already proved the case $k=1$ in
Lemmas~\ref{lem:1ary-summary} and~\ref{lem:1ary-active-region} for $\pot \geq 2$ and $\twistVal = 2^{\pot-1}$
using a $(1,\pot,2^{\pot-1}, d)$ blurrer.
This can easily be adapted for the case that $\twistVal = a\cdot 2^{\twistVal(1)} = 2a$ for some $a \in \ZZ_{2^\pot}$.
We start with a $(1,\pot,2,3)$-blurrer given by Lemma~\ref{lem:blurrer-larger-field}
and turn it 
into a $(1,\pot,2a,d)$-blurrer using Lemma~\ref{lem:blurrer-basics}.
Then the proof proceeds exactly the same.

So assume $k>1$.
Let $m \in \nat$, $G=(V,E,\leq)$ be a regular and $(2k+m+1)$-connected base graph of degree $d \geq d(k,m)$ and girth at least $2\minRad(k+1)$, and
$\set{\stip,\stip'} \in E$.
The bound on the connectivity of~$G$ is needed in the following
to apply the results from Sections~\ref{sec:cfi-orbits},~\ref{sec:matrices-cfi}, and~\ref{sec:active-region}
and we will not further mention this when applying them.
Let $\pot \geq \pot(k)$.
As before, we denote for every $f \colon E \to \ZZ_{2^\pot}$
by~$\Struct_f$ the CFI structure $\CFIgraph{2^\pot}{G}{f}$
with universe~$\StructVA$ (which is equal for all such~$f$).
Let $f,g\colon E \to\ZZ_{2^\pot}$ such that $f(e) = g(e)$ for all $e \in E \setminus \set{\set{\stip, \stip'}}$ and $g(\set{\stip, \stip'}) = f(\set{\stip, \stip'}) + \twistVal$, where $\twistVal = a \cdot 2^{\twistVal(k)}$
for some $a \in \ZZ_{2^\pot}$.
Furthermore, let $\pTupA\in \StructV^m$ such that $\distance{G}{\stip}{\orig{\pTupA}}> \minRad(k+1)$ be arbitrary but fixed.
In particular, $f,g$ do not twist $\orig{\pTupA}$.

Let $\scenter$ be a vertex with $\distance{G}{\scenter}{\stip'} =  \minRad(k)+1$ and $\distance{G}{\scenter}{\stip} = \minRad(k)+2$.
Choose vertices $\stip=\stip_1, \dots, \stip_d$ and $\stip'=\primeSub{\stip}{1}, \dots, \primeSub{\stip}{d}$
such that there are simple paths $\spath_i= (\scenter, \dots, \primeSub{\stip}{i}, \stip_i)$ of length $\minRad(k)+2 > 2k+1$
for all $i \in [d]$
forming a star.
Such paths exist because the girth of~$G$ is at least $2\minRad(k+1) >  2 \minRad(k) + 4$ and the degree is~$d$ (cf.~Figure~\ref{fig:star}).

\begin{figure}
	\centering
	
	\begin{tikzpicture}
		
		\tikzstyle{vertex} = [circle, fill=black, inner sep=0.5mm]
		\tikzstyle{dot} = [circle, fill=black, inner sep=0.3mm]
		
		\draw[gray, thick]
		(0,0) -- (-1.7, 4.2)
		(0,0) -- (+1.7, 4.2);
		\fill[gray!10!white]
		(0,0) -- (-1.7, 4.2) -- (+1.7, 4.2) -- cycle;
		
		\draw[gray, thick]
		(0,0) -- (4.2,-1.7)
		(0,0) -- (4.2,+1.7);
		\fill[gray!10!white]
		(0,0) -- (4.2,-1.7) -- (4.2, +1.7) -- cycle;
		
		\draw[gray, thick]
		(0,0) -- (-4.2,-1.7)
		(0,0) -- (-4.2,+1.7);
		\fill[gray!10!white]
		(0,0) -- (-4.2,-1.7) -- (-4.2, +1.7) -- cycle;

		\node[vertex, label ={[below=1.7mm]$\scenter$} ] (z) at (0,0){};;
		\node[vertex, label ={[right]$\primeSub{\stip}{1}\hspace{-2pt}=\hspace{-1pt}\stip'$} ] (e1) at (0,3){};;
		\node[vertex, label ={[right]$\stip_1\hspace{-2pt}=\hspace{-1pt}\stip$} ] (e1p) at (0,3.7){};;
		\node[vertex, label ={[below=0.6mm]\strut$\primeSub{\stip}{2}$} ] (e2) at (3,0){};;
		\node[vertex, label ={[below=0.6mm]\strut$\stip_2$} ] (e2p) at (3.7,0){};;
		\node[vertex, label ={[below=0.6mm]\strut$\primeSub{\stip}{d}$} ] (ed) at (-3,0){};;
		\node[vertex, label ={[below=0.6mm]\strut$\stip_d$} ] (edp) at (-3.7,0){};;

	 	\draw[black]
	 		(e1) -- (e1p)
	 		(e2) -- (e2p)
	 		(ed) -- (edp);
		\draw[decorate,decoration={snake,pre length=3mm, post length=3mm} ]
			(z) -- (e1)
			node [black,midway,xshift=0.4cm,yshift=0.3cm,align=center] 
			{$\spath_1$};
		\draw[decorate,decoration={snake,pre length=3mm, post length=3mm} ]
			(z) -- (e2)
			node [black,midway,xshift=0.3cm,yshift=0.35cm,align=center] 
			{$\spath_2$};
		\draw[decorate,decoration={snake,pre length=3mm, post length=3mm} ]
			(z) -- (ed)
			node [black,midway,xshift=-0.3cm,yshift=0.35cm,align=center] 
			{$\spath_d$};
		
		\draw[decorate,decoration={brace,amplitude=15pt},xshift=-3mm, very thin]
		(0,0) -- (0,3.7) node [black,midway,xshift=-1.5cm,align=center] 
		{\footnotesize $\distance{G}{\scenter}{\stip_1}$\\\footnotesize$\geq \minRad(k)+2$};
		
		\node[dot] at (0,-1.4){};
		\node[dot] at (0.3,-1.35){};
		\node[dot] at (0.58,-1.25){};
		\node[dot] at (-0.3,-1.35){};
		\node[dot] at (-0.58,-1.25){};
		
	\end{tikzpicture}
	\caption{The star $\spath_1, \dots, \spath_d$.
		Each path~$\spath_i$ is contained in its own tree rooted at~$\scenter$
		(depicted in gray)
		because~$G$ has girth $\geq 4 \minRad(k)$.
	}
	\label{fig:star}
\end{figure}

\begin{claim}
	\label{clm:distance-param-e-i}
	We have $\distance{G}{\primeSub{\stip}{i}}{\orig{\pTupA}} \geq 2 \minRad(k)$ for every $i \in [d]$
	and $\distance{G}{\primeSub{\stip}{i}}{\primeSub{\stip}{j}} = 2 \minRad(k)+2$ for every $i \neq j$.
\end{claim}
\begin{claimproof}
	By choice of~$\scenter$, $\distance{G}{\scenter}{\primeSub{\stip}{i}}=\minRad(k) + 1$.
	Let $i,j \in [d]$ such that $i \neq j$.
	The $\primeSub{\stip}{i}$\nobreakdash-$\primeSub{\stip}{j}$\nobreakdash-path obtained by joining~$\spath_i$ and~$\spath_j$ (after removing~$\stip_i$ and~$\stip_j$ respectively)
	has length $2\cdot(\minRad(k) + 1)$ by construction.
	If this path was not a shortest path,
	then we would obtain a cycle of length less than $4\cdot(\minRad(k) + 1)$.
	This contradicts that~$G$ has girth at least
	$2 \minRad(k+1) \geq 4\cdot(\minRad(k) + 1)$.
	It follows $\distance{G}{\primeSub{\stip}{i}}{\primeSub{\stip}{j}} = \distance{G}{\primeSub{\stip}{i}}{\scenter}  + \distance{G}{\scenter}{\primeSub{\stip}{j}}  = 2\minRad(k) + 2$.
	
	Let $\bvertB \in \orig{\pTupA}$.
	Then $\distance{G}{\bvertB}{\primeSub{\stip}{1}} \leq \distance{G}{\bvertB}{\primeSub{\stip}{i}} + \distance{G}{\primeSub{\stip}{i}}{\primeSub{\stip}{1}}$.
	By assumption,
	$\distance{G}{\bvertB}{\primeSub{\stip}{1}} \geq \minRad(k+1) = 4 \minRad(k) + 2$
	and by the former argument
	$\distance{G}{\primeSub{\stip}{i}}{\primeSub{\stip}{1}} = 2 \minRad(k)+2$.
	It follows that 
	\[\distance{G}{\bvertB}{\primeSub{\stip}{i}} \geq \distance{G}{\bvertB}{\primeSub{\stip}{1}} - \distance{G}{\primeSub{\stip}{1}}{\primeSub{\stip}{i}}
	= 4\minRad(k)+2 - 2\minRad(k) - 2  \geq 2\minRad(k).\qedhere\]
\end{claimproof}

We call a set $\comp \subseteq V$ such that $|\comp| \leq 2k$
and $G[\comp]$ is connected a \defining{$2k$-component}.
Every $2k$-component is a component of some $2k$-orbit (cf.~Definition~\ref{def:component}).
Because~$k$ is fixed in this proof, we just call them components in the following.
\begin{definition}
	We call a component~$\comp$
	\begin{itemize}
		\item an \defining{$i$-tip component} if $\set{\stip_i,\primeSub{\stip}{i}} \subseteq \comp$ and $i \in [d]$,
		\item a \defining{star component} if~$\comp$ intersects non-trivially with the star $\spath_1, \dots, \spath_d$ and $\stip_i \notin \comp$ for all $i \in [d]$,
		\item an \defining{$i$-star component} if~$\comp$ is a star component,~$\comp$ intersects non-trivially with~$\spath_i$,
		and for every other $j\neq i$,~$\comp$ intersects trivially with~$\spath_j$,
		\item a \defining{star center component} if~$\comp$ is a star component but
		 not an $i$-star component for every $i \in [d]$,
		\item and otherwise a \defining{sky component}.
	\end{itemize}
\end{definition}

\begin{claim}
	\label{clm:star-tip-components-basics}
	Let~$\comp$ be a component and $\comp'\subseteq \comp$ be a component. 
	\begin{enumerate}
		\item If~$\comp$ is an $i$-star component,
		then~$\comp'$ is an $i$-star or a  sky component
		\item If~$\comp$ is a star component,
		then~$\comp'$ is a star or a  sky component.
		\item If~$\comp$ is an $i$-tip component,
		then~$\comp'$ is an $i$-tip, an $i$-star, or a sky component.
		\item If~$\comp$ is a star component,
		then~$\comp$ is an $i$-star component for some~$i$ if and only if
		$\scenter \notin \comp$.
	\end{enumerate}
\end{claim}
\begin{claimproof}
	To show Case~1, let~$\comp$ be an $i$-star component.
	By definition,~$\comp$ only has a non-trivial intersection with~$\spath_i$
	and $\stip_i \notin \comp$.
	So every $\comp' \subseteq \comp$ has either a trivial intersection with every~$\spath_j$,
	i.e.,~$\comp'$ is a sky component,
	or a non-trivial intersection only with $\spath_i$ and $\stip_i \notin \comp' \subseteq \comp$,
	i.e.,~$\comp'$ is an $i$-star component.
	The Case~2 where~$\comp$ is a star component in similar.
	
	For Case~3, let~$\comp$ be an $i$-tip component.
	Because $\distance{G}{\primeSub{\stip}{i}}{\primeSub{\stip}{j}} = 2r(k) + 2 > 4k+4$
	(Claim~\ref{clm:distance-param-e-i}),
	$\distance{G}{\primeSub{\stip}{i}}{\scenter} = r(k) + 1 > 2k$,
	$|\comp| \leq 2k$, and because~$G[\comp]$ is connected,
	$\comp$ has a trivial intersection with every~$\spath_j$ for all $j \neq i$.
	Let $\comp' \subseteq \comp$.
	Now~$\comp'$ is an $i$-tip component if $\set{\stip_i,\primeSub{\stip}{i}} \subseteq \comp'$,
	an $i$-star component if otherwise~$\comp'$ intersects non-trivially with~$\spath_i$,
	or otherwise a sky component.
	
	Finally, to prove Case~4, let~$\comp$ be a star component.
	If~$\comp$ is also an $i$-star component,
	then $\scenter \notin \comp$ because otherwise~$\comp$ would intersect non-trivially with all~$\spath_j$.
	For the reverse direction, let $\scenter \notin \comp$
	and~$\comp$ have a non-trivial intersection with~$\spath_i$.
	Because $\distance{G}{\primeSub{\stip}{i}}{\primeSub{\stip}{j}} = 2r(k) + 2 > 4k+4$,
	$|\comp| \leq 2k$, $G[\comp]$ is connected,
	and because $\scenter \notin \comp$,
	the component $\comp$ cannot have a non-trivial intersection with another~$\spath_j$.
	Hence,~$\comp$ is an $i$-star component.
\end{claimproof}

To blur the twist, we want to distribute it among the edges $\set{\stip_i,\primeSub{\stip}{i}}$ (for all $i \in[d]$)
similar to the $1$-ary case.
Here we blurred the twist between all edges adjacent to~$\scenter$.
%In contrast to the $1$-ary case,
%we have to consider the types of the tuples serving as row and column indices of the matrix, which we want to define. 
In the $1$-ary case, $1$\nobreakdash-tuples always had the same type in $(\StructA_f, \pTupA)$ and in $(\StructA_g, \pTupA)$.
However, for $k$\nobreakdash-tuples, this is no longer the case.
We now want to construct a function
that maps a tuple~$\tupA$ to a tuple~$\tupB$
such that~$\tupB$ has the same type in $(\StructA_g,\pTupA)$
as~$\tupA$ has in $(\StructA_f, \pTupA)$.
To do so, we use a path isomorphism on the path~$\spath_1$.
We generalize this and not only want to ``repair'' the types
for a twist at the edge $\set{\stip_1,\primeSub{\stip}{1}}$ but
possibly for multiple twists at all edges $\set{\stip_i,\primeSub{\stip}{i}}$.

Let $\tup{a} \in \ZZ_{2^\pot} ^d$.
We define a function $\fixTwistType_{\tup{a}} \colon \StructVA^{\leq2k} \to \StructVA^{\leq 2k}$
that preserves the size of tuples.
Set $\fixTwistTypeAut{\tup{a}}{i} := \pathiso{a_i}{\spath_i}$
(cf.~Definition~\ref{def:pathiso}).
The function~$\fixTwistType_{\tup{a}}$ applies $\fixTwistTypeAut{\tup{a}}{i}$ to tuples,
but only to those components containing some of the edges $\{\stip_i,\primeSub{\stip}{i}\}$.
These are precisely the $i$-tip components:
\begin{align*}
	\fixTwistType_{\tup{a}}(\tupA) &:= (\vertB_1, \dots, \vertB_{|\tupA|}), \text{where } \tupA=(\vertA_1, \dots, \vertA_{|\tupA|}) \text{ and}\\
	\vertB_j &:= \begin{cases}
		\fixTwistTypeAut{\tup{a}}{i}(\vertA_j) & \text{if } \orig{\vertA_j} \in \comp \text{ and } \comp \text{ is an $i$-tip component of } \tupA,\\
		\vertA_j & \text{otherwise.}
	\end{cases}
\end{align*}
Given a function $h \colon E \to \ZZ_{2^\pot}$,
we write $h + \tup{a}$ for the function $h' \colon E \to \ZZ_{2^\pot}$
such that
$h'(\set{\stip_i,\primeSub{\stip}{i}}) =h(\set{\stip_i,\primeSub{\stip}{i}})+ a_i$ for all $i \in [d]$, 
and $h'(e) = h(e)$ otherwise.

\begin{claim}
	\label{clm:fxTwistType-fixes-type-general}
	Suppose $\tup{a} \in \ZZ_{2^\pot} ^d$,
	$h \colon E \to \ZZ_{2^\pot}$, and
	$k' \leq 2k$.
	If $\BlockA \in \orbs{k'}{(\Struct_h,\pTupA)}$,
	then $\fixTwistType_{\tup{a}}(\BlockA) \in \orbs{k'}{(\StructA_{h + \tup{a}},\pTupA)}$
	and $\fixTwistType_{\tup{a}}(\BlockA)$ has the same type in $(\StructA_{h + \tup{a}},\pTupA)$ as~$\BlockA$ has in $(\StructA_h,\pTupA)$.
\end{claim}
\begin{claimproof}
	Let $\BlockA \in \orbs{k'}{(\StructA_h,\pTupA)}$.
	It suffices to consider the case that~$\BlockA$ has a single component~$\comp$
	because~$\fixTwistType_{\tup{a}}$ is defined component-wise
	and
	because by Lemma~\ref{lem:split-discon-orbits}
	the type of a disconnected orbit is given by the types of 
	the restrictions to the components.
	If~$\comp$ does not contain~$\stip_i$ and~$\primeSub{\stip}{i}$ for some $i \in [d]$,
	then $\fixTwistType_{\tup{a}}$ is the identity function.
	Because~$\BlockA$ does not cover the twisted edge,
	it has the same type in $(\Struct_h,\pTupA)$
	and in $(\StructA_{h + \tup{a}},\pTupA)$.
	
	So assume $\set{\stip_i,\primeSub{\stip}{i}} \subseteq \comp$.
	This is the case for exactly one $i \in [d]$ by Claim~\ref{clm:star-tip-components-basics}.
	Let~$h_i$ be the function equal to~$h$ for all edges apart from
	$h_i (\set{\stip_i,\primeSub{\stip}{i}}) := h(\set{\stip_i,\primeSub{\stip}{i}}) + a_i$ and $h_i (\set{\scenter, \scenter_i}) := h(\set{\scenter, \scenter_i}) - a_i$,
	where~$\scenter_i$ is the neighbor of~$\scenter$ used in the path~$\spath_i$.
	By Claim~\ref{clm:distance-param-e-i},
	the parameters~$\pTupA$ have distance $2 \minRad(k)$
	to $\primeSub{\stip}{i}$ and thus in particular are not contained in~$\spath_i$.
	So~$\fixTwistTypeAut{\tup{a}}{i}$
	is an isomorphism between $(\StructA_{h},\pTupA)$
	and $(\StructA_{h_i},\pTupA)$
	by Lemma~\ref{lem:pathiso}.
	Because~$\BlockA$ is a $k'$-orbit and $k' \leq 2k$,
	neither~$\scenter$ nor its neighbors (in particular not~$\scenter_i$) are contained in~$\comp$,
	because $2k < \distance{G}{\scenter}{\primeSub{\stip}{i}} = \minRad(k) +1$.
	So $\fixTwistType_{\tup{a}}(\BlockA) = \fixTwistTypeAut{\tup{a}}{i}(\BlockA)$
	has the same type in $(\StructA_{h_i},\pTupA)$
	as~$\BlockA$ has in $(\StructA_{h},\pTupA)$.
	Now, between $(\StructA_{h+\tup{a}},\pTupA)$
	and
	$(\StructA_{h_i},\pTupA)$
 	all edges $\set{\stip_\ell,\primeSub{\stip}{\ell}}$ for $\ell \neq i$
	and the edge $\set{\scenter,\scenter_i}$ are potentially twisted.
	Because $\scenter \notin \comp$ and $\set{\stip_\ell,\primeSub{\stip}{\ell}} \not\subseteq \comp$ for every $\ell \neq i$,
	$(\StructA_{h_i},\pTupA)[\comp] = (\StructA_{h+\tup{a}},\pTupA)[\comp]$.
	Hence, the type of $\fixTwistTypeAut{\tup{a}}{i}(\BlockA)$
	in $(\StructA_{h_i},\pTupA)$
	is equal to the type of $\fixTwistType_{\tup{a}}(\BlockA)$ in $(\StructA_{h+\tup{a}},\pTupA)$.
\end{claimproof}

We now construct a blurrer for our setting.
Let $i\in \nat$ such that
$2^{i-1}-1 < k \leq 2^i -1$.
\begin{enumerate}
	\item By Lemma~\ref{lem:blurrer-larger-field}, there is a
$(2^i -1, \pot-\twistVal(k-1), 2^i, 2^{i+1}-1)$-blurrer
(note that ${\pot - \twistVal(k-1) \geq \pot(k) - \twistVal(k-1) = i+1}$).
	\item We use Lemma~\ref{lem:blurrer-embedding} to turn it into a
$(2^i-1, \pot, 2^{i+\twistVal(k-1)}, 2^{i+1}-1)$-blurrer
by embedding $\ZZ_{2^\pot-\twistVal(k-1)}$ in $\ZZ_{2^\pot}$.
	\item We use Lemma~\ref{lem:blurrer-basics} to get a
$(k, \pot,  2^{i+\twistVal(k-1)}, 2^{i+1}-1)$-blurrer because ${k \leq 2^i -1}$,
	\item 
then a $(k, \pot,  2^{i+\twistVal(k-1)}, d)$-blurrer 
because $d \geq d(k,m) \geq 2^{i+1} - 1$,
and finally 
\item a $(k, \pot, a\cdot 2^{i+\twistVal(k-1)}, d)$-blurrer $\blurrer$.
\end{enumerate}
By this construction,
for every $\blurElem \in \blurrer$
and every $j \in [d]$
there is some $b \in \ZZ_{2^\pot}$
such that $\blurElem(j) = b \cdot 2^{\twistVal(k-1)}$
because we embedded $\ZZ_{2^\pot-\twistVal(k-1)}$ in $\ZZ_{2^\pot}$.
Let $\blurElemFix = (a\cdot 2^{i+\twistVal(k-1)}, 0, \dots, 0)$ be the tuple given by Lemma~\ref{lem:blurrer-sum} for $\blurrer$. Then $g = f + \blurElemFix$.
We set \defining{$\fixTwistType := \fixTwistType_{\blurElemFix}$}.
\begin{corollary}
	\label{cor:tau-fixes-orbit-type}
	If $k' \leq 2k$ and $\BlockA \in \orbs{k'}{(\StructA_f,\pTupA)}$,
	then $\fixTwistType(\BlockA) \in \orbs{k'}{(\StructA_g,\pTupA)}$
	and $\fixTwistType(\BlockA)$ has the same type in $(\StructA_g,\pTupA)$ as $\BlockA$ has in $(\StructA_f, \pTupA)$.
\end{corollary}
\begin{proof}
	Follows from Claim~\ref{clm:fxTwistType-fixes-type-general}.
\end{proof}
We have seen that with~$\fixTwistType$
we can ``repair'' the types of the orbits,
but~$\fixTwistType$ introduces inconsistencies between tuples
along the path~$\spath_1$.
This can easily be seen already for $k=2$:
Consider a $2$-tuple $(\vertA,\vertB) $ with origin $(\stip_1, \primeSub{\stip}{1})$
and the $2$-tuple $(\vertB, \vertC)$ with origin $(\primeSub{\stip}{1}, \bvertA)$, 
where~$\bvertA$ is the next vertex in the path~$\spath_1$.
Then $\fixTwistType((\vertA,\vertB)) = (\vertA,\vertB')$ for some $\vertB' \neq \vertB$ but $\fixTwistType((\vertB,\vertC)) = (\vertB,\vertC)$.
So clearly the composed $4$-tuple $(\vertA,\vertB,\vertB,\vertC)$
has a different type than $(\vertA, \vertB', \vertB, \vertC)$.
For these inconsistencies, we use the blurrer~$\blurrer$ as follows.

We now define another function which according to a $\blurElem \in \blurrer$
``distributes'' the twists among the edges $\set{\stip_i,\primeSub{\stip}{i}}$ using a star isomorphism.
We associate with each $\blurElem \in \blurrer$ a function $\autoB_\blurElem\colon\StructVA^{\leq k} \to \StructVA^{\leq k}$, again preserving tuple sizes:
Set $\autoA_\blurElem := \stariso{\blurElem}{\spath_1, \dots, \spath_d}$ (cf.~Definition~\ref{def:stariso}).
Then define $\autoB_\blurElem\colon \StructVA^{\leq k} \to \StructVA^{\leq k}$ as follows:
Let $\tupA \in \StructVA^{\leq k}$.
We set
\begin{align*}
	\autoB_\blurElem(\tupA) &:= (\vertB_1,\dots,\vertB_{|\tupA|}),
	\text{where } \tupA=(\vertA_1, \dots, \vertA_{|\tupA|}) \text{ and}\\
	\vertB_i &:= \begin{cases}
		\autoA_\blurElem(\vertA_i) & \text{if } \orig{\vertA_i} \in \comp \text{ and } \comp \text{ is a star component of } \tupA, \\
		\vertA_i & \text{otherwise.}
	\end{cases}
\end{align*} 
That is, we apply~$\autoA_\blurElem$ to all components of~$\tupA$
that do not contain the tips of the star $\spath_1,\dots,\spath_d$.
On the sky components~$\autoA_\blurElem$ is the identity function anyway.
From now on,
we will identify~$\blurElem$ with~$\autoB_\blurElem$ and write~$\blurElem(\tupA)$.

\begin{claim}
	\label{clm:component-wise-application}
	Let $\tupA \in \StructVA^\ell$
	and the components of~$\tupA$ be partitioned into~$D$ and~$D'$.
	Then $\fixTwistType_{\tup{a}}(\tupA) = \fixTwistType_{\tup{a}}(\tupA_D)\fixTwistType_{\tup{a}}(\tupA_{D'})$
	and $\blurElem(\tupA) = \blurElem(\tupA_D)\blurElem(\tupA_{D'})$
	for every $\tup{a} \in \ZZ_{2^\pot} ^ d$ and
	$\blurElem \in \blurrer$.
\end{claim}
\begin{claimproof}
	The claim is immediate because~$\fixTwistType_{\tup{a}}$ and~$\autoB_\blurElem$
	are defined component-wise.
\end{claimproof}

\begin{claim}
	\label{clm:commutativity}
	For every $\tup{a}\in \ZZ_{2^\pot}^d$ and $\blurElem \in \blurrer$,
	the functions $\fixTwistType_{\tup{a}}$, $\autoB_\blurElem$, and every automorphism $\autoA \in \autgrp{(\StructA_f, \pTupA)} = \autgrp{(\StructA_g, \pTupA)}$ commute.
\end{claim}
\begin{claimproof}
	The functions~$\fixTwistType_{\tup{a}}$ and~$\autoB_\blurElem$
	are defined component-wise by isomorphisms.
	We saw in Section~\ref{sec:cfi-isomorphisms} that isomorphisms between CFI structures
	are composed of automorphisms of each gadget.
	Because the automorphism group of a gadget is abelian (Lemma~\ref{lem:cfi-orbit-auto-group-regular}),
	the said functions commute.
\end{claimproof}

\begin{definition}
	An \defining{orbit-automorphism} is a function $\zeta \colon \StructVA^{\leq 2k} \to \StructVA^{\leq 2k}$
	that satisfies the following:
	For every $\BlockA\in \orbs{k'}{(\StructA_f, \pTupA)}=\orbs{k'}{(\StructA_g, \pTupA)}$ with $k'\leq 2k$
	there is an automorphism $\autoA_\BlockA \in \autgrp{(\StructA_f, \pTupA)}=\autgrp{(\StructA_g, \pTupA)}$
	such that  $\zeta(\tupA) = \autoA_\BlockA(\tupA)$ for all $\tupA \in \BlockA$.
\end{definition}
That is, an orbit-automorphism is a function,
whose action on a single orbit is the action of an automorphism.
For different orbits, the corresponding automorphisms may be different.
This matches the definition of an orbit-invariant matrix~(cf.~Definition~\ref{def:orbit-invariant-matrix}),
which is invariant under all orbit-automorphisms.
\begin{claim}
	\label{clm:blurelem-blockiso}
	Every $\blurElem \in \blurrer$
	is an orbit-automorphism.
\end{claim}
\begin{claimproof}
	By Lemma~\ref{lem:split-discon-orbits}, it suffices to show the
	claim for connected orbits $\BlockA \in \orbs{k'}{(\StructA_f, \pTupA)}$ with $k' \leq 2k$.
	If the origin of~$\BlockA$ is not a star component,
	then~$\blurElem$ is the identity function on~$\BlockA$
	and so clearly an orbit-automorphism.
	
	Otherwise, the origin $\comp := \orig{\BlockA}$ is a star component.
	Then in particular $\stip_i \notin \comp$ for all $i \in [d]$.
	We show that there are paths $\spath_i' = (\primeSub{\stip}{i}, \stip_i, \dots, \stip_1, \primeSub{\stip}{1})$
	that are completely disjoint from $\orig{\pTupA}$ and
	possibly apart from $\primeSub{\stip}{i}, \primeSub{\stip}{1}$ disjoint from~$\comp$.
	Set $\comp' := \comp \setminus \setcond{\primeSub{t}{i}}{i \in [d]}$.
	Consider the graph $G \setminus \comp' \setminus \orig{\pTupA}$.
	We removed at most $|\orig{\pTupA}| + k' \leq 2k + m$ many vertices.
	Because~$G$ is $(2k + m + 1)$-connected,
	the claimed paths exists in $G \setminus \comp' \setminus \orig{\pTupA}$.
	We then use path isomorphisms $\autoA_i :=\pathiso{\blurElem(i)}{\spath_i'}$
	for all $i \in [d] \setminus\set{1}$
	to move the twist introduced by~$\autoB_\blurElem$
	at the $\set{\stip_i, \primeSub{\stip}{i}}$ to $\{\stip_1,\primeSub{\stip}{1}\}$.
	We set $\autoB :=\autoB_\blurElem \circ \autoA_2 \circ \cdots \circ \autoA_d$.
	Now, we have that
	$(\StructA_f, \pTupA) \iso \autoB ((\StructA_f, \pTupA)) =  (\StructA_f,  \pTupA)$
	by Lemmas~\ref{lem:pathiso} and~\ref{lem:stariso}.
	Let $\tupA \in \BlockA$.
	The isomorphisms~$\autoA_i$ are the identity function on vertices in $\orig{\BlockA}$
	because~$\autoA_i$ is the identity on~$\primeSub{\stip}{i}$ and~$\primeSub{\stip}{1}$ (cf.~Definition~\ref{def:pathiso})
	and other vertices in $\orig{\BlockA}$ are not contained in~$\spath_i'$ for every $i \in [d]$.
	That is, $\BlockA = \autoB(\BlockA) = \autoB_\blurElem(\BlockA) = \blurElem(\BlockA)$ and~$\blurElem$ is an orbit-automorphism.
\end{claimproof}
We show how~$\blurElem$ and~$\fixTwistType_\blurElem$ work together.
That is, for $\tupA\tupB$ of length at most~$2k$,
the tuple $\fixTwistType_\blurElem(\blurElem(\tupA)) \fixTwistType_\blurElem(\blurElem(\tupB))$
has the same type as the tuple $\fixTwistType_\blurElem(\tupA\tupB)$.
This is not true for applying only~$\fixTwistType_\blurElem$
to~$\tupA$ and~$\tupB$ because their origins may overlap.
\begin{claim}
	\label{clm:blurrable-repair-orbits-general}
	Let $h\colon E \to \ZZ_{2^\pot}$, $k' \leq 2k$,
	$\BlockA \in \orbs{k'}{(\StructA_{h},\pTupA)}$,
	and $\blurElem \in \blurrer$.
	Every  $\tupA\tupB\in \StructVA^{k'}$
	satisfies  $\tupA\tupB \in \BlockA$ if and only if $\fixTwistType_\blurElem(\blurElem(\tupA)) \fixTwistType_\blurElem(\blurElem(\tupB)) \in \fixTwistType_\blurElem(\BlockA)$.
\end{claim}
\begin{claimproof}
	Set $\BlockA_1 := \restrictVect{\BlockA}{\set{1,\dots,|\tupA|}}$
	and $\BlockA_2 := \restrictVect{\BlockA}{\set{|\tupA|+1,\dots, k'}}$.
	Let $\comp^i_1, \dots, \comp^i_{\ell_i}$ be the components of~$\BlockA_i$ for every $i \in [2]$.
	Because~$\fixTwistType_\blurElem$ and~$\blurElem$ are defined component-wise,
	it suffices by Lemma~\ref{lem:split-discon-orbits} to 
	assume that~$\BlockA$ has a single component~$\comp$.
	We need  to verify
	that $\tupA\tupB \in \BlockA$ if and only if 
	$\fixTwistType_\blurElem(\blurElem(\tupA))\fixTwistType_\blurElem(\blurElem(\tupB)) \in \fixTwistType_\blurElem(\BlockA)$.
	The component~$\comp$ is the union of the~$\comp^i_j$.
	We perform the following case distinction:
	\begin{itemize}
		\item Assume~$\comp$ is an $n$-tip component.
		For every $i \in [2]$,
		let~$D_i^T$ be the set of the $n$-tip components~$\comp^i_j$,~%
		$D_i^S$ be the set of the $n$-star components~$\comp^i_j$,
		and~$D_i^R$ be the set of sky components~$\comp^i_j$.
		This yields a partition of all~$\comp^i_j$
		by Claim~\ref{clm:star-tip-components-basics}.
		
		Then we have by the definitions of~$\fixTwistType_\blurElem$
		and~$\blurElem$ that~$\blurElem$ is the identity function on~$D^T_i$,~%
		$\fixTwistType_\blurElem$ is the identity on~$D^S_i$,
		and both are the identity on~$D^R_i$.
		That is,
		\begin{align*}
			&\hiddenEq \fixTwistType_\blurElem(\blurElem(\tupA))\fixTwistType_\blurElem(\blurElem(\tupB))\\
			&= \fixTwistType_\blurElem(\blurElem(\tupA_{D^T_1}\tupA_{D^S_1}\tupA_{D^R_1}))
			\fixTwistType_\blurElem(\blurElem(\tupB_{D^T_2}\tupB_{D^S_2}\tupB_{D^R_2}))\\
			&= \fixTwistType_\blurElem(\tupA_{D^T_1})\blurElem(\tupA_{D^S_1})\fixTwistType_\blurElem(\tupA_{D^R_1})
			\fixTwistType_\blurElem(\tupB_{D^T_2})\blurElem(\tupB_{D^S_2})\fixTwistType_\blurElem(\tupB_{D^R_2}). \tag{$\star$}
		\end{align*}
		When working on the whole component~$\comp$,~%
		$\fixTwistType_\blurElem$ applies $\fixTwistTypeAut{\blurElem}{n}$ to vertices
		in components of~$D^S_1$ and~$D^S_2$
		because~$\comp$ is an $n$-tip component.
		We see that $\restrictVect{\autoA_\blurElem}{D^S_i} = \restrictVect{\fixTwistTypeAut{\blurElem}{n}}{D^S_i}$
		because~$D^S_i$ is an $n$-star-component
		and so does not contain~$\scenter$
		(cf. the Definitions~\ref{def:pathiso} and~\ref{def:stariso}
		and the definitions of~$\autoA_\blurElem$ and $\fixTwistTypeAut{\blurElem}{n}$). It follows that
		\begin{align*}
			 (\star) = \fixTwistType_\blurElem(\tupA_{D^T_1}\tupA_{D^S_1}\tupA_{D^R_1}
			\tupB_{D^T_2}\tupB_{D^S_2}\tupB_{D^R_2}) = \fixTwistType_\blurElem(\tupA\tupB).
		\end{align*}
		So $\tupA\tupB \in \BlockA$ if and only if $\fixTwistType_\blurElem(\blurElem(\tupA))
		\fixTwistType_\blurElem(\blurElem(\tupB)) = \fixTwistType_\blurElem(\tupA\tupB) \in \fixTwistType_\blurElem(\BlockA)$.
		\item 
		Otherwise,~$\comp$ is not a tip component.
		We distinguish two more cases:
		\begin{itemize}
			\item
			If~$\comp$ is a star component,
			let~$D^S_i$ be the set of star components~$\comp^i_j$
			and~$D^R_i$ be the set of sky components~$\comp^i_j$
			for all $i \in [2]$.
			There are no tip components of the~$\comp^i_j$ by Claim~\ref{clm:star-tip-components-basics}.
			So we again partitioned all components~$\comp^i_j$.
			Now~$\fixTwistType_\blurElem$
			is the identity function on all~$D^S_i$ and~$D^R_i$
			and~$\blurElem$ is the identity on all~$D^R_i$.
			So we have
			$\tupA\tupB \in \BlockA$
			if and only if
			$\fixTwistType_\blurElem(\blurElem(\tupA))\fixTwistType_\blurElem(\blurElem(\tupB))=\blurElem(\tupA)\blurElem(\tupB) =\blurElem(\tupA\tupB)
			\in \fixTwistType_\blurElem(\BlockA) = \BlockA$
			by Claim~\ref{clm:blurelem-blockiso}.
			\item
			Otherwise,~$\comp$ is a sky component
			and both~$\fixTwistType_\blurElem$ and~$\blurElem$
			are the identity function on~$\comp$ and all~$\comp^i_j$
			and the claim follows immediately.
			\qedhere
		\end{itemize}
	\end{itemize}
\end{claimproof}

Using the blurrer~$\blurrer$, we will be able to blur the twist
in many cases, but not in all.
The problem is the following:
If we only look at~$k$ many of the $\set{\stip_i,\primeSub{\stip}{i}}$ edges,
then the blurrer properties will ensure that
we cannot see the twist, 
i.e., ``summing'' over all elements in the blurrer 
maps a $2k$-orbit of $(\StructA_f, \pTupA)$
to a $2k$-orbit of the same type in $(\StructA_g,\pTupA)$
similar to the $1$-ary case.
Let us briefly recall the arguments
to prove Lemma~\ref{lem:1-ary-matrices-similar},
which shows that summing over blurrer elements indeed
yields a matrix bluring the twist for arity~$1$.
There are two cases:
First, for a tuple $(\vertA,\vertB)$ with origin $(\scenter,\scenter)$
the action of every $\blurElem \in \blurrer$ is the action of an automorphism,
so $(\blurElem(\vertA),\vertB)$ and $(\vertA, \inv{\blurElem}(\vertB))$
are in the same orbit.
Second, for a tuple $(\vertA,\vertB)$ with origin $(\scenter,\stip_1)$ 
only one index (namely the first for~$\stip_1$)
was relevant: $\blurElem(\vertA)\vertB$ and $\blurElem'(\vertA)\vertB$
are in the same orbit if $\blurElem(1) = \blurElem'(1)$.
So whenever $\blurElem(1) = \blurElem'(1)$,
the terms for  $\blurElem(\vertA)\vertB$ and $\blurElem'(\vertA)\vertB$
canceled out in the summation.
The blurrer properties ensured that 
only one term for $\blurElem(\vertA)\vertB$
of the same type in $(\StructA_g, \pTupA)$
as $\vertA\vertB$ in $(\StructA_f, \pTupA)$ remained.

For arity~$k$ the two cases (automorphism or blurrer properties) can be mixed.
Consider $k=2$ and a $4$-tuple $\tup{\vertA}$ with origin $(\scenter,\primeSub{ \stip}{1}, \scenter, \stip_1)$.
Then for every $\blurElem,\blurElem' \in \blurrer$, $\blurElem(\vertA_1\vertA_2)\vertA_3\vertA_4$
and $\blurElem'(\vertA_1\vertA_2)\vertA_3\vertA_4$
are in the same orbit if and only if $\blurElem = \blurElem'$
(because fixing one vertex with origin~$\scenter$ splits the gadget of~$\scenter$ into singleton orbits).
That is, we cannot argue solely with blurrer properties.
However, the two tuples
$\blurElem(\vertA_1\vertA_2)\vertA_3\vertA_4$
and $\vertA_1\vertA_2\inv{\blurElem}(\vertA_3\vertA_4)$
are not in the same orbit in general
because every automorphism mapping~$\vertA_1\vertA_2$ to $\blurElem(\vertA_1\vertA_2)$
cannot be the identity for~$\vertA_4$ but which $\inv{\blurElem}(\vertA_3\vertA_4)$ is.
So we also cannot argue solely with automorphisms.
The techniques of the $1$-ary case can only be applied
if~$\scenter$ is not in the origin of at least one of~$\vertA_1\vertA_2$ and~$\vertA_3\vertA_4$.

In general, let $\BlockA\in \orbs{2k}{(\StructA_f, \pTupA)}$
and $\tupA \in \BlockA$.
Then in~$\charMat{\BlockA}$
the first~$k$ positions of~$\tupA$ will serve as row index
and the remaining~$k$ positions as column index.
The problem with the blurrer only occurs
if both the first and second half of~$\tupA$ contain~$\scenter$ in its origin.
So we make a case distinction on whether a $k$-orbit contains~$\scenter$ in its origin.

\begin{definition}
	We call a $\BlockA \in \orbs{k}{(\StructA_f, \pTupA)}$
	\defining{blurrable} if $\scenter \notin \orig{\BlockA}$.
\end{definition}

To also be able to blur the twist for non-blurrable orbits,
we use a recursive approach.
Because $\sum \blurElem = 0$, $\StructA_{g - \blurElem}$ is isomorphic to $\StructA_g$ for every $\blurElem \in \blurrer$.
Let~$\pVertA_\scenter$ be an arbitrary vertex with origin~$\scenter$.
Our goal now is to blur the twist between $(\StructA_f,\pTupA\pVertA_\scenter)$
and $(\StructA_{g-\blurElem},\pTupA\pVertA_\scenter)$.
This will exactly undo the action of a $\blurElem \in \blurrer$,
when we consider an orbit that fixes a vertex with origin~$\scenter$.
We exploit the high girth of~$G$ to blur the twists at each $\set{\stip_i,\primeSub{\stip}{i}}$
independently.
Before we start to blur twists between $(\StructA_f,\pTupA\pVertA_\scenter)$
and $(\StructA_{g-\blurElem},\pTupA\pVertA_\scenter)$,
we first have to show that~$\fixTwistType_\blurElem$ and~$\blurElem$
are compatible with orbits when fixing the additional vertex~$\pVertA_\scenter$.
The following two claims are in some sense refinements of Claims~\ref{clm:fxTwistType-fixes-type-general}
and~\ref{clm:blurrable-repair-orbits-general}.

\begin{claim}
	\label{clm:rec-active-regions-orbits}
	For every $\tup{a}\in \ZZ_{2^\pot}$ and $k' \leq 2k$,
	every orbit $\BlockA \in\orbs{k'}{(\StructA_{f+\tup{a}},\pTupA\pVertA_\scenter)}$
	satisfying $\orig{\BlockA} \cap \neighborsK{G}{1}{\scenter} =\emptyset$
	is contained in $\orbs{k'}{(\StructA_{f+\tup{a}},\pTupA)}$.
\end{claim}
\begin{claimproof}
	Let $\BlockA \in\orbs{k'}{(\StructA_{f+\tup{a}},\pTupA\pVertA_\scenter)}$
	such that $\orig{\BlockA} \cap \neighborsK{G}{1}{\scenter} =\emptyset$.
	By Lemma~\ref{lem:cfi-same-type-automorphism},
	it suffices to show that two tuples $\tupA,\tupB \in \StructV^{k'}$,
	such that $\orig{\tupA}=\orig{\tupB}$ is disjoint from $\neighborsK{G}{1}{\scenter}$,
	have the same type in $(\StructA_{f+\tup{a}},\pTupA)$
	if and only if they have the same type in $(\StructA_{f+\tup{a}},\pTupA\pVertA_\scenter)$.
	Because $\orig{\BlockA} \cap \neighborsK{G}{1}{\scenter} = \emptyset$,
	the components of $\tupA\pTupA\pVertA_\scenter$
	are the components of~$\tupA\pTupA$ and the one of~$\pVertA_\scenter$.
	Hence,
	if~$\tupA$ and~$\tupB$ have the same type in $(\StructA_{f+\tup{a}},\pTupA)$,
	then they also have the same type in $(\StructA_{f+\tup{a}},\pTupA\pVertA_\scenter)$.
	The other direction is trivial.
\end{claimproof}

\begin{claim}
	\label{clm:fxTwistType-fixes-type-general-rec}
	Let $\tup{a} \in \ZZ_{2^\pot} ^d$,
	$h \colon E \to \ZZ_{2^\pot}$, and
	$k' \leq 2k$.
	If $\BlockA \in \orbs{k'}{(\Struct_h,\pTupA\pVertA_\scenter)}$,
	then $\fixTwistType_{\tup{a}}(\BlockA) \in \orbs{k'}{(\StructA_{h + \tup{a}},\pTupA\pVertA_\scenter)}$
	and $\fixTwistType_{\tup{a}}(\BlockA)$ has the same type in $(\StructA_{h + \tup{a}},\pTupA\pVertA_\scenter)$ as~$\BlockA$ has in $(\StructA_h,\pTupA\pVertA_\scenter)$.
\end{claim}
\begin{claimproof}
	Let $\BlockA \in \orbs{k'}{(\Struct_h,\pTupA\pVertA_\scenter)}$,~%
	$R$ be the set of components~$\comp$ of~$\BlockA$
	with $\comp \cap \neighborsK{G}{1}{\scenter} \neq \emptyset$,
	and~$D$ be the set of all remaining components of~$\BlockA$.
	Then $\BlockA = \restrictVect{\BlockA}{R} \times \restrictVect{\BlockA}{D}$
	by Claim~\ref{lem:split-discon-orbits}.
	Similarly, $\fixTwistType_{\tup{a}}(\BlockA) = \restrictVect{\fixTwistType_{\tup{a}}(\BlockA)}{R} \times \restrictVect{\fixTwistType_{\tup{a}}(\BlockA)}{D}$.
	Every component~$\comp$ in~$R$ does not contain the edges $\set{\stip_i, \primeSub{\stip}{i}}$
	for all $i \in [d]$
	because $|\comp| \leq k'$ but every path~$\spath_{i}$
	has length $\minRad(k) + 2 > 2k + 1 \geq k' + 1$.
	That is,~$\fixTwistType_{\tup{a}}$ is the identity
	on~$\restrictVect{\BlockA}{R}$ and so
	$\restrictVect{\BlockA}{R} = \restrictVect{\fixTwistType_{\tup{a}}(\BlockA)}{R}$
	and~$\restrictVect{\BlockA}{R}$ has the same type in 
	$(\StructA_{h + \tup{a}},\pTupA\pVertA_\scenter)$
	as it has in $(\StructA_h,\pTupA\pVertA_\scenter)$.
	By Claim~\ref{clm:rec-active-regions-orbits},
	the orbit $\restrictVect{\fixTwistType_{\tup{a}}(\BlockA)}{D}$
	is an orbit of $(\StructA_h,\pTupA)$
	and has by Claim~\ref{clm:fxTwistType-fixes-type-general}
	the same type in $(\StructA_h,\pTupA)$
	as $\restrictVect{\fixTwistType_{\tup{a}}(\BlockA)}{D}$
	has in $(\StructA_{h + \tup{a}},\pTupA)$.
	It follows that~$\BlockA$
	has the same type in $(\StructA_h,\pTupA)$
	as $\fixTwistType_{\tup{a}}(\BlockA)$ has in
	$(\StructA_{h + \tup{a}},\pTupA)$.
\end{claimproof}

\begin{claim}
	\label{clm:orbits-rec-blurelem}
	Suppose $\BlockB \in \orbs{k'}{(\StructA_{g-\blurElem}, \pTupA\pVertA_\scenter)}$ for some $\blurElem \in \blurrer$ and $k' \leq 2k-2$.
	Then $\fixTwistType_\blurElem(\BlockB) \in \orbs{k'}{(\StructA_g, \pTupA\pVertA_\scenter)}$ and has the same type as~$\BlockB$.
	Let~$\comp$ be the connected component of $G[\orig{\BlockB} \cup \set{\scenter}]$
	containing~$\scenter$,~%
	$R$ be the set of components of~$\BlockB$ contained in~$\comp$,~%
	$D$ be the set of all other components of~$\BlockB$, and
	let $\tupC\tupB \in \StructVA^{k'}$.
	Then $\tupC\tupB \in \BlockB$
	if and only if
	$\tupC_R\blurElem(\fixTwistType_\blurElem(\tupC_D))
	\tupB_R\blurElem(\fixTwistType_\blurElem(\tupB_D)) \in \fixTwistType_\blurElem(\BlockB)$.
\end{claim}
\begin{claimproof}
	We split~$\BlockB$ using Lemma~\ref{lem:split-discon-orbits} in 
	$\BlockB = \restrictVect{\BlockB}{R} \times \restrictVect{\BlockB}{D}$.
	Because $k' \leq 2k-2$, components in~$R$ cannot be tip components.
	Hence, $\fixTwistType_\blurElem(\BlockB) = \restrictVect{\BlockB}{R} \times \fixTwistType_\blurElem(\restrictVect{\BlockB}{D})$.
	By Claim~\ref{clm:rec-active-regions-orbits}, $\restrictVect{\BlockB}{D}$ is also an orbit of $(\StructA_{g-\blurElem}, \pTupA)$ because its origin has distance greater than $1$ to $\orig{\pVertA_\scenter}$.
	Then
	$\fixTwistType_\blurElem(\restrictVect{\BlockB}{D})$
	has the same type in $(\StructA_g, \pTupA)$
	and is also an orbit of $(\StructA_g, \pTupA\pVertA_\scenter)$
	of the same type, too, by Claim~\ref{clm:fxTwistType-fixes-type-general}.
	It follows that $\fixTwistType_\blurElem(\BlockB)$
	has the same type as~$\BlockB$.
	Let $\tupC\tupB \in \StructVA^{k'}$.
	Using the splitting above,
	from Claim~\ref{clm:blurrable-repair-orbits-general} it follows
	that $\tupC_D\tupB_D \in \restrictVect{\BlockB}{D}$
	if and only if
	$\blurElem(\fixTwistType_\blurElem(\tupC_D))
	\blurElem(\fixTwistType_\blurElem(\tupB_D)) \in \fixTwistType_\blurElem(\restrictVect{\BlockB}{D})$. 
	The claim follows because $\BlockB = \restrictVect{\BlockB}{R} \times \restrictVect{\BlockB}{D}$.
\end{claimproof}

Now we construct matrices $(k-1)$-bluring the twist between $(\StructA_f, \pTupA\pVertA)$
and $(\StructA_{g-\blurElem}, \pTupA\pVertA_\scenter)$.
For every $\blurElem \in \blurrer$ and $j \in [d+1]$ we define
$\gblurelem{\blurElem}{j} \colon E \to \ZZ_{2^\pot}$ to be the following function:
\begin{align*}
	\gblurelem{\blurElem}{j} (e) &:= \begin{cases}
		f(\set{\stip_i, \primeSub{\stip}{i}}) & \text{if } e = \set{\stip_i, \primeSub{\stip}{i}} \text { for some } i \geq j,\\
		(g-\blurElem)(e) &\text{otherwise}.
	\end{cases}
\end{align*}
Note that $f(e)=g(e) = \gblurelem{\blurElem}{j}(e)$
for all~$e$ different from the edges $\set{\stip_i, \primeSub{\stip}{i}}$,
$\gblurelem{\blurElem}{1} = f$, $\gblurelem{\blurElem}{d+1} =g -\blurElem$,
and the only possibly twisted edge by $\gblurelem{\blurElem}{j}$ and $\gblurelem{\blurElem}{j+1}$ 
is $\set{\stip_{j}, \primeSub{\stip}{j}}$ for every $j \in [d]$.
Define $N_j := \neighborsK{G}{\minRad(k)}{\stip_j}$ for every $j \in [d]$.

\begin{claim}
	\label{clm:blur-twist-blur-elem-single}
	For every $\blurElem \in \blurrer$ and $j \in [d]$,
	there is an $\StructVA^{k-1}\times \StructVA^{k-1}$
	matrix $\Sblurelem{\blurElem}{j}$,
	which is odd-filled and both orbit-diagonal and orbit-invariant over 
	$(\StructA_{\gblurelem{\blurElem}{j}},\pTupA \pVertA_\scenter)$ and $(\StructA_{\gblurelem{\blurElem}{j+1}},\pTupA \pVertA_\scenter)$,
	which ${(k-1)}$\nobreakdash-blurs the twist between $(\StructA_{\gblurelem{\blurElem}{j}},\pTupA \pVertA_\scenter)$ and $(\StructA_{\gblurelem{\blurElem}{j+1}},\pTupA \pVertA_\scenter)$,
	and those active region satisfies ${\activeRegionIdx{\gblurelem{\blurElem}{j},\gblurelem{\blurElem}{j+1}}{\pTupA \pVertA_\scenter}{\Sblurelem{\blurElem}{j}} \subseteq N_j}$.
	In particular, $\Sblurelem{\blurElem}{j} = \idmat$ if $\blurElem(j) = 0$.
\end{claim}
\begin{claimproof}
	Let $\blurElem \in \blurrer$ and $j \in [d]$.
	If $\blurElem(j) = 0$, then $\gblurelem{\blurElem}{j} = \gblurelem{\blurElem}{j+1}$
	and $\Sblurelem{\blurElem}{j} := \idmat$ trivially satisfies the claim.
	Otherwise, the matrix~$\Sblurelem{\blurElem}{j}$ is obtained from the induction hypothesis:
	The number of parameters increased by one,
	but we consider tuples of length $k-1$.
	We continue to consider~$\ZZ_{2^\pot}$.
	\begin{itemize}
		\item Clearly, $\pot \geq \pot(k) \geq \pot(k-1)$, 
		the degree of $G$ is $d\geq d(k,m) \geq d(k-1, m+1)$, and the girth of~$G$ is at least $4r(k) +2 > 2r(k)$.
		\item 
		We have $2k+m+1 \geq 2 (k-1) +  (m+1) +1$
		and so~$G$ satisfies the connectivity condition.	
		\item By construction, we have that $\gblurelem{\blurElem}{j}(e) = \gblurelem{\blurElem}{j+1}(e)$ for all $e \in E \setminus\set{\set{\stip_j, \primeSub{\stip}{j}}}$
		for every $j \in[d]$.
		\item We consider the value of the twist:
		Let $j \in [d]$. Then $\blurElem(j) = b\cdot 2^{\twistVal(k-1)}$ for some $b \in \ZZ_{2^\pot}$ (as shown before when constructing the blurrer~$\blurrer$). 
		If $j \neq 1$,
		then it holds that $\gblurelem{\blurElem}{j+1}(\set{\stip_j,\primeSub{\stip}{j}}) = \gblurelem{\blurElem}{j}(\set{\stip_j,\primeSub{\stip}{j}}) - b \cdot 2^{\twistVal(k-1)}$.
		If otherwise $j=1$, then we have
		$\gblurelem{\blurElem}{2}(\set{\stip_1,\primeSub{\stip}{1}}) =
		g(\set{\stip_1,\primeSub{\stip}{1}}) - \blurElem(1) =
		\gblurelem{\blurElem}{1}(\set{\stip_1,\primeSub{\stip}{1}}) - \blurElem(1) + \twistVal$
		because $\gblurelem{\blurElem}{1} = f$
		and $g(\set{\stip_1,\primeSub{\stip}{1}}) = f(\set{\stip_1,\primeSub{\stip}{1}}) + \twistVal$.
		By assumption, we have that $\twistVal = a \cdot 2^{\twistVal(k)} 
		= a \cdot 2^{i + \twistVal(k-1)}$ for some $a \in \ZZ_{2^q}$.
		Clearly $- \blurElem(1) + \twistVal = 
		-b \cdot  2^{\twistVal(k-1)} + a \cdot 2^{i + \twistVal(k-1)} = 
		(a \cdot2^{i} -b)\cdot2 ^{\twistVal(k-1)}$.
		So the value of the twist at the $\set{\stip_j,\primeSub{\stip}{j}}$
		is in all cases $c \cdot 2^{\twistVal(k-1)}$ for some $c \in \ZZ_{2^q}$.
		\item 
		By Claim~\ref{clm:distance-param-e-i},
		it holds that $\distance{G}{\primeSub{\stip}{i}}{\orig{\pTupA}} \geq 2 \minRad(k)$
		for every $i \in [d]$,
		in particular $\distance{G}{\stip_i}{\orig{\pTupA}} > \minRad(k)$.
		By construction, it holds that $\distance{G}{\primeSub{\stip}{i}}{\orig{\pVertA_\scenter}} = \distance{G}{\primeSub{\stip}{i}}{\scenter} = r(k)+1$.
		So $\distance{G}{\stip_i}{\orig{\pTupA\pVertA_\scenter}} > \minRad(k)$
		for every $i \in [d]$.\qedhere
	\end{itemize}
\end{claimproof}

We now define for every $\blurElem \in \blurrer$ the $\StructVA^{k-1}\times \StructVA^{k-1}$ matrix~$S^\blurElem$ as follows:
\[S^\blurElem := \Sblurelem{\blurElem}{1} \cdot \ldots \cdot \Sblurelem{\blurElem}{d},\]
where~$\Sblurelem{\blurElem}{j}$ is the matrix given by Claim~\ref{clm:blur-twist-blur-elem-single} for~$\blurElem$ and~$j$.

\begin{claim}
	\label{clm:blur-twist-blur-elem}
	For every $\blurElem \in \blurrer$,
	the matrix~$S^\blurElem$
	is odd-filled and both orbit-diagonal and orbit-invariant over $(\StructA_f, \pTupA \pVertA_\scenter)$ and $(\StructA_{g-\blurElem}, \pTupA \pVertA_\scenter)$,
	$(k-1)$-blurs the twist between $(\StructA_f, \pTupA \pVertA_\scenter)$ and $(\StructA_{g-\blurElem}, \pTupA \pVertA_\scenter)$,
	and satisfies $\activeRegionIdx{f,g - \blurElem}{\pTupA \pVertA_\scenter}{S^\blurElem} \subseteq \bigcup_{i=1}^d N_i$.
\end{claim}
\begin{claimproof}
	For every $j \in [d]$, the matrix~$\Sblurelem{\blurElem}{j}$
	is odd-filled and both orbit-diagonal and orbit-invariant over 
	$(\StructA_{\gblurelem{\blurElem}{j}},\pTupA \pVertA_\scenter)$ and $(\StructA_{\gblurelem{\blurElem}{j+1}},\pTupA \pVertA_\scenter)$,
	$(k-1)$-blurs the twist between $(\StructA_{\gblurelem{\blurElem}{j}},\pTupA \pVertA_\scenter)$ and $(\StructA_{\gblurelem{\blurElem}{j+1}},\pTupA \pVertA_\scenter)$,
	and satisfies $\activeRegionIdx{\gblurelem{\blurElem}{j},\gblurelem{\blurElem}{j+1}}{\pTupA \pVertA_\scenter}{\Sblurelem{\blurElem}{j}} \subseteq N_j$
	by Claim~\ref{clm:blur-twist-blur-elem-single}.

	Because $f = \gblurelem{\blurElem}{1}$ and $g-\blurElem = \gblurelem{\blurElem}{d+1}$,
	$S^\blurElem = \Sblurelem{\blurElem}{1} \cdot \ldots \cdot \Sblurelem{\blurElem}{d}$
	is orbit-diagonal and orbit-invariant over $(\StructA_f, \pTupA \pVertA_\scenter)$ and $(\StructA_{g-\blurElem}, \pTupA \pVertA_\scenter)$
	by Lemmas~\ref{lem:orbit-diagonal-multiply} and~\ref{lem:orbit-invariant-mult}.
	By Lemma~\ref{lem:odd-filled-product}, the matrix~$S^\blurElem$ is odd-filled.
	It $(k-1)$-blurs the twist between $(\StructA_f, \pTupA \pVertA_\scenter)$ and $(\StructA_{g-\blurElem}, \pTupA \pVertA_\scenter)$
	by Lemma~\ref{lem:blur-twist-mult}.
	Finally, $\activeRegionIdx{f,g - \blurElem}{\pTupA \pVertA_\scenter}{S^\blurElem} \subseteq \bigcup_{i=1}^d
	\activeRegionIdx{\gblurelem{\blurElem}{j},\gblurelem{\blurElem}{j+1}}{\pTupA \pVertA_\scenter}{\Sblurelem{\blurElem}{j}}
	 \subseteq \bigcup_{i=1}^d N_i$
	by Lemma~\ref{lem:active-region-mult}.
\end{claimproof}

\begin{claim}
	\label{clm:neighborhoods_different}
	For every pair of distinct $i,j \in [d]$ it holds that $N_i \cap N_j = \emptyset$.
\end{claim}
\begin{claimproof}
	Let $i \neq j$.
	Assume that there is an $\bvertA \in N_i \cap N_j$.
	Then there is a path from~$\stip_i$ to~$\stip_j$ of length at most $2 \minRad(k)$.
	By construction $\distance{G}{\stip_i}{\scenter} = \distance{G}{\stip_j}{\scenter} = \minRad(k) +2$ and so $\scenter \notin N_i$ and $\scenter \notin N_j$.
	But that means that there is a cycle of length at most $
	\distance{G}{\stip_i}{\scenter} + \distance{G}{\stip_j}{\scenter} + \distance{G}{\stip_i}{x} + \distance{G}{\stip_j}{x} \leq 4\minRad(k) + 4$ contradicting that~$G$ has girth at least $ 2 \minRad(k+1) \geq  8\minRad(k) + 4$.
\end{claimproof}

\begin{claim}
	\label{clm:blur-twist-blur-elem-orbit-auto}
	For every $\blurElem \in \blurrer$, 
 	every $\BlockA \in \orbs{k-1}{(\StructA_f, \pTupA \pVertA_\scenter)}$,
	$\BlockB \in \orbs{k-1}{(\StructA_{g-\blurElem}, \pTupA \pVertA_\scenter)}$
	of the same type in $(\StructA_{g-\blurElem}, \pTupA \pVertA_\scenter)$ as $\BlockA$ in $(\StructA_f, \pTupA \pVertA_\scenter)$,
	$\tupA \in \BlockA$,
	$\tupB \in \BlockB$, and
	$\autoA \in \autgrp{(\Struct_f,\pTupA)}$
	the matrix~$S^\blurElem$
	satisfies $S^\blurElem(\tupA,\tupB) = S^\blurElem(\autoA(\tupA),\autoA(\tupB))$.
\end{claim}
\begin{claimproof}
	Let $\blurElem \in \blurrer$.
	By Claim~\ref{clm:blur-twist-blur-elem}, 
	the matrix~$S^\blurElem$ is orbit-invariant over $(\StructA_f, \pTupA \pVertA_\scenter)$ and $(\StructA_{g-\blurElem}, \pTupA \pVertA_\scenter)$
	and thus satisfies the claim for all $\autoA \in \autgrp{(\Struct_f,\pTupA\pVertA_\scenter)}$.
	But now we also want to consider automorphisms not fixing $\pVertA_\scenter$.
	
	So let $\BlockA \in \orbs{k-1}{(\StructA_f, \pTupA \pVertA_\scenter)}$,
	$\BlockB \in \orbs{k-1}{(\StructA_{g-\blurElem}, \pTupA \pVertA_\scenter)}$
	be of the same type in $(\StructA_{g-\blurElem}, \pTupA \pVertA_\scenter)$ as $\BlockA$ in $(\StructA_f, \pTupA \pVertA_\scenter)$,
	$\tupA \in \BlockA$,
	$\tupB \in \BlockB$, and
	$\autoA \in \autgrp{(\Struct_f,\pTupA)}$.
	Let~$R$ be the set of components of~$\tupA$ (and thus of~$\tupB$) containing a vertex of $\neighborsK{G}{1}{\scenter}$
	(so in particular~$\scenter$ itself).
	Let~$D$ be the set of the remaining components.
	
	Because $\tupA_{D}\tupB_{D}$ and $\autoA(\tupA_{D}\tupB_{D})$
	are in the same orbit in $(\StructA_{f+\tup{a}},\pTupA)$,
	they are also in the same orbit in $(\StructA_{f+\tup{a}},\pTupA\pVertA_\scenter)$
	by Claim~\ref{clm:rec-active-regions-orbits}.
	Hence, there is a $\autoB \in \autgrp{(\StructA_{f+\tup{a}},\pTupA\pVertA_\scenter)}$
	satisfying $\autoA(\tupA_{D}\tupB_{D})= \autoB(\tupA_{D}\tupB_{D})$.
	We now apply that $S^\blurElem$ is orbit-invariant over $(\StructA_f, \pTupA \pVertA_\scenter)$ and $(\StructA_{g-\blurElem}, \pTupA \pVertA_\scenter)$:
	\begin{align*}
		S^\blurElem(\tupA,\tupB) &= S^\blurElem(\tupA_{R}\tupA_{D},\tupB_{R}\tupB_{D})\\
		&=S^\blurElem(\tupA_{R}\autoB(\tupA_{D}),\tupB_{R}\autoB(\tupB_{D})) \\
		&= S^\blurElem(\tupA_{R}\autoA(\tupA_{D}),\tupB_{R}\autoA(\tupB_{D})) \tag{$\star$}.
	\end{align*}
	For every $\comp \in R$ it holds that $\comp \not\subseteq \activeRegionIdx{f,g - \blurElem}{\pTupA \pVertA_\scenter}{S^\blurElem} \subseteq \bigcup_{i=1}^d N_i$
	because $\neighborsK{G}{1}{\scenter} \cap \bigcup_{i=1}^d N_i = \emptyset$,
	which follows from $N_i = \neighborsK{G}{\minRad(k)}{\stip_i}$
	and $\distance{G}{\scenter}{\stip_i} =  \minRad(k)+2$.
	So we can apply Condition~\ref{itm:active-region-replace} of the active region because $\tupA_{R} = \tupB_{R}$
	if and only if $\autoA(\tupA_{R}) = \autoA(\tupB_{R})$:
	\begin{align*}
		(\star) &= S^\blurElem(\autoA(\tupA_{D_\scenter})\tupA_{D_R},\autoA(\tupB_{D_\scenter})\tupB_{D_R})\\
		&= S^\blurElem(\autoA(\tupB),\autoA(\tupB)).\qedhere
	\end{align*}
\end{claimproof}

\begin{claim}
	\label{clm:sum-s-blur-subtree}
	Let
	$k' \leq k-1$,
	$\blurElem \in \blurrer$,
	$\BlockA' \in \orbs{k'}{(\StructA_f, \pTupA\pVertA_\scenter)}$,
	$K \subseteq [d]$,~%
	$D$ be the set of all components~$C$ of~$\BlockA'$
	satisfying $C \subseteq N_i$ for some $i \in K$,
	and~$R$ be the set of remaining components.
	Let $\BlockB' = \fixTwistType(\BlockA') \in \orbs{k'}{(\StructA_g, \pTupA\pVertA_\scenter)}$,
	$\tupA \in \BlockA'$, and $\tupB \in \BlockB'$.
	Then
	\[\sum_{\tupC_D \in \restrictVect{\BlockA'}{D}} 
	S^\blurElem(\tupC_D\tupA_R,\tupB_D\tupB_R) = 
	\bigg(\prod_{i \in [d] \setminus K}\Sblurelem{\blurElem}{i}\bigg)(\tupB_D\tupA_R,\tupB_D\tupB_R)\]
	and $\activeRegionIdx{f, f+\tup{a}}{\pTupA \pVertA_\scenter}{\prod_{i \in [d] \setminus K}\Sblurelem{\blurElem}{i}} \subseteq \bigcup_{i \in [d]\setminus K} N_i$,
	where $\tup{a} \in \ZZ_{2^\pot}^d$ satisfies
	$a_i =  (g - \blurElem)(\set{\stip_i,\primeSub{\stip}{i}}) - f(\set{\stip_i,\primeSub{\stip}{i}})$ if $i\notin K$ and $a_i=0$ otherwise for every $i \in [d]$.
\end{claim}
\begin{claimproof}
	Recall that $S^\blurElem = \Sblurelem{\blurElem}{1} \cdot \ldots \cdot \Sblurelem{\blurElem}{d}$
	and that
	$\activeRegionIdx{\gblurelem{\blurElem}{j},\gblurelem{\blurElem}{j+1}}{\pTupA \pVertA_\scenter}{\Sblurelem{\blurElem}{j}} \subseteq N_j$ for every $j \in [d]$
	by Claim~\ref{clm:blur-twist-blur-elem-single}.
	The first part of the claim follows from repeated application of Lemma~\ref{lem:orbit-diagonal-multiply-sum-middle-disjoint-active-region}
	using that the sets~$N_j$ are disjoint (Claim~\ref{clm:neighborhoods_different}).
	The second part follows from repeated application of 
	Lemmas~\ref{lem:active-region-change-equally} and~\ref{lem:active-region-mult}.
\end{claimproof}

We introduce more notation.
Let $\tupA \in \StructVA^{k'}$ such that $\scenter \in \orig{\tupA}$.
Then $\remZ{\tupA} \in \StructVA^{k'-1}$ is the tuple obtained from~$\tupA$
by deleting the first entry with origin~$\scenter$.
This first entry is denoted by~$\tupA_\scenter$.
Similarly to our convention for~$\tupA_\comp$ for a component $\comp$ in Section~\ref{sec:cfi-orbits},
we write $\tupA_\scenter\remZ{\tupA}$
not for concatenation but for inserting~$\tupA_\scenter$ at the correct position
such that $\tupA_\scenter\remZ{\tupA} = \tupA$. 
Now we are ready to define the $\StructVA^k \times \StructVA^k$ matrix~$S$.
For $j \in \set{k,2k}$ we set 
$\PartA_j := \orbs{j}{(\StructA_f, \pTupA)}$.
We define the $\BlockA \times \fixTwistType(\BlockA)$ block of~$S$
for every $\BlockA \in \PartA_k$. All other blocks are zero.
\[
	S_{\BlockA \times \fixTwistType(\BlockA)}(\tupA,\tupB) := \begin{cases}
		\sum
		\limits_{\substack{\blurElem \in \blurrer, \\ \fixTwistType(\blurElem(\tupA)) = \tupB}} 1 & \text{if } \BlockA \text{ is blurrable}, \\
		\sum\limits_{\substack{\blurElem \in \blurrer, \\ \fixTwistType_\blurElem(\blurElem(\tupA_\scenter)) = \tupB_\scenter}}
		S^{\blurElem} (\blurElem(\remZ{\tupA}), \inv{\fixTwistType_\blurElem}(\remZ{\tupB})) & \text{if } \BlockA \text{ is not blurrable}.
	\end{cases}
\]

We first check that we make reasonable use of the matrices~$S^\blurElem$:
\begin{claim}
	\label{clm:matrix-blurrer-tuple-same-type}
	Let $\BlockA \in \PartA_k$ be non-blurrable,
	$\tupA \in \BlockA$, and
	$\tupB \in \fixTwistType(\BlockA)$.
	If~$\tupA$ has the same type in $(\StructA_f,\pTupA)$
	as~$\tupB$ has in $(\StructA_g,\pTupA)$,
	then 
	for every $\blurElem \in \blurrer$ such that $\fixTwistType_\blurElem(\blurElem(\tupA_\scenter)) = \tupB_\scenter$,
	the tuple $\blurElem(\remZ{\tupA})$ has the same type in $(\Struct_f,\pTupA\pVertA_\scenter)$
	as $\inv{\fixTwistType_\blurElem}(\remZ{\tupB})$ has in $(\StructA_{g - \blurElem}, \pTupA\pVertA_\scenter)$.
\end{claim}
\begin{claimproof}
	Let  $\blurElem \in \blurrer$ such that $\fixTwistType_\blurElem(\blurElem(\tupA_\scenter)) = \tupB_\scenter$.
	Let~$R$ be the star component of $\orig{\BlockA}$ containing~$\scenter$
	and let~$D$ be the set of remaining components of $\orig{\BlockA}$.
	Then ${\BlockA= \restrictVect{\BlockA}{R}\times \restrictVect{\BlockA}{D}}$
	by Lemma~\ref{lem:split-discon-orbits}.
	In particular $\fixTwistType(\BlockA) = \restrictVect{\BlockA}{R}\times\fixTwistType(\restrictVect{\BlockA}{D})$
	because~$\fixTwistType$ is the identity on star components.
	The orbit $\fixTwistType(\restrictVect{\BlockA}{D})$ has the same type in $(\StructA_g, \pTupA)$
	as~$\restrictVect{\BlockA}{D}$ has in $(\StructA_f,\pTupA)$ (Corollary~\ref{cor:tau-fixes-orbit-type})
	and $\fixTwistType_\blurElem(\fixTwistType(\restrictVect{\BlockA}{D}))$
	has thus the same type in $(\StructA_{g-\blurElem}, \pTupA)$ as 
	$\restrictVect{\BlockA}{D}$ has in $(\StructA_f,\pTupA)$ (Claim~\ref{clm:fxTwistType-fixes-type-general}).
	Because $\scenter \in R$ ($\BlockA$ is not blurrable),
	all components in~$D$ have distance at least~$2$ to~$\scenter$.
	Thus, $\restrictVect{\BlockA}{D}$ is also an orbit of $(\StructA_f,\pTupA\pVertA_\scenter)$ (Claim~\ref{clm:rec-active-regions-orbits})
	and has the same type in $(\StructA_f,\pTupA\pVertA_\scenter)$
	as $\fixTwistType_\blurElem(\fixTwistType(\restrictVect{\BlockA}{D}))$
	has in $(\StructA_{g-\blurElem}, \pTupA\pVertA_\scenter)$.
	Because~$\blurElem$ is an orbit-automorphism (Claim~\ref{clm:blurelem-blockiso}),
	$\blurElem(\restrictVect{\BlockA}{D}) = \restrictVect{\BlockA}{D}$.
	It follows that $\blurElem(\tupA_{D}) \in \restrictVect{\BlockA}{D}$ and $\fixTwistType_\blurElem(\tupB_{D}) \in\fixTwistType_\blurElem(\fixTwistType(\restrictVect{\BlockA}{D})) $ have the same type. 
	
	\hspace{-0.4pt}It suffices to show that $\blurElem(\tupA_{R})$ has the same type in $(\StructA_f, \pTupA\pVertA_\scenter)$ as $\fixTwistType_\blurElem(\tupB_{R})$
	has in $(\StructA_{g-\blurElem}, \pTupA\pVertA_\scenter)$.
	Because~$R$ contains star components,~$\fixTwistType_\blurElem$ is the identity on~$R$ and so $\tupB_{R}\in \restrictVect{\BlockA}{R}$.
	Further, $\blurElem(\tupA_R) \in \restrictVect{\BlockA}{R}$ because $\blurElem$ is an orbit-automorphism.
	That is, there is an automorphism $\autoB \in \autgrp{(\StructA_f,\pTupA)}$
	such that $\autoB(\blurElem(\tupA_{R})) = \tupB_{R}$.
	Because by assumption $\fixTwistType_\blurElem(\blurElem(\tupA_\scenter)) = \blurElem(\tupA_\scenter) = \tupB_\scenter$ ($\fixTwistType_\blurElem$ is the identity on vertices with origin~$\scenter$),~%
	$\autoB$ is the identity on vertices with origin~$\scenter$
	and thus $\autoB \in \autgrp{(\StructA_f,\pTupA\pVertA_\scenter)}$.
	That is, $\blurElem(\tupA_{R})$ and~$\tupB_{R}$
	are in the same orbit of $(\StructA_f,\pTupA\pVertA_\scenter)$.
	Because~$R$ contains star components,
	it does not contain any edge twisted by~$f$ and $g-\blurElem$
	and thus~$\blurElem(\tupA_{R})$ and~$\tupB_{R}$ have the same type
	in $(\StructA_f,\pTupA\pVertA_\scenter)$ and in $(\StructA_{g-\blurElem},\pTupA\pVertA_\scenter)$.
\end{claimproof}

\begin{claim}
	The matrix~$S$ is orbit-diagonal over $(\StructA_f, \pTupA)$ and $(\StructA_g, \pTupA)$.
\end{claim}
\begin{claimproof}
	By definition, the only nonzero blocks of~$S$ are the
	$\BlockA \times \fixTwistType(\BlockA)$ blocks.
	By Corollary~\ref{cor:tau-fixes-orbit-type},~$\BlockA$ and~$\fixTwistType(\BlockA)$ have the same type
	for every $\BlockA \in \PartA_k$.
\end{claimproof}

\begin{claim}
	The matrix~$S$ is orbit-invariant over $(\StructA_f, \pTupA)$ and $(\StructA_g, \pTupA)$.
\end{claim}
\begin{claimproof}
	Let $\autoA \in \autgrp{(\StructA_f, \pTupA)} = \autgrp{(\StructA_g, \pTupA)}$,
	$\BlockA \in \PartA_k$, $\tupA \in \BlockA$, and $\tupB \in \BlockB := \fixTwistType(\BlockA)$.
	We perform a case distinction:
	Assume that~$\BlockA$ is blurrable.
	The functions~$\fixTwistType$,~$\blurElem$, and~$\autoA$ commute (Claim~\ref{clm:commutativity})
	and thus
	$\fixTwistType(\blurElem(\autoA(\tupA))) =  \autoA(\fixTwistType(\blurElem(\tupA)))$
	for every $\blurElem \in \blurrer$.
	Because~$\autoA$ is a bijection,
	$\autoA(\fixTwistType(\blurElem(\tupA))) = \autoA(\tupB)$
	if and only if $\fixTwistType(\blurElem(\tupA)) = \tupB$.
	So
	\[S(\autoA(\tupA),\autoA(\tupB))
	=\sum_{\substack{\blurElem \in \blurrer, \\ \fixTwistType(\blurElem(\autoA(\tupA))) = \autoA(\tupB)}} 1
	=\sum_{\substack{\blurElem \in \blurrer, \\ \autoA(\fixTwistType(\blurElem(\tupA))) = \autoA(\tupB)}} 1
	=\sum_{\substack{\blurElem \in \blurrer, \\ \fixTwistType(\blurElem(\tupA)) = \tupB}} 1
	= S(\tupA,\tupB)
	 .\]
	Otherwise, assume that~$\BlockA$ is non-blurrable.
	Then for every $\blurElem \in \blurrer$ the following holds
	because~$\fixTwistType_\blurElem$,~$\blurElem$, and~$\autoA$ commute 
	by Claim~\ref{clm:commutativity}
	and because~$S^\blurElem$ is invariant under~$\autoA$
	by Claim~\ref{clm:blur-twist-blur-elem-orbit-auto}:
	\begin{align*}
		&\hiddenEq S^\blurElem(\blurElem(\autoA(\remZ{\tupA})), \inv{\fixTwistType_\blurElem}(\autoA(\remZ{\tupB})))\\
		&= S^\blurElem(\autoA(\blurElem(\remZ{\tupA})), \autoA(\inv{\fixTwistType_\blurElem}(\remZ{\tupB})))\\
		&= S^\blurElem(\blurElem(\remZ{\tupA}), \inv{\fixTwistType_\blurElem}(\remZ{\tupB})).
	\end{align*}
	Because~$\fixTwistType_\blurElem$,~$\blurElem$, and~$\autoA$ commute
	and because~$\autoA$ is a bijection,
	it holds that $\fixTwistType_\blurElem(\blurElem(\autoA(\tupA_\scenter))) = \autoA(\tupB_\scenter)$
	if and only if $\fixTwistType_\blurElem(\blurElem(\tupA_\scenter)) = \tupB_\scenter$.
	Hence,
	\begin{align*}
	S(\autoA(\tupA),\autoA(\tupB)) &=
	\sum\limits_{\substack{\blurElem \in \blurrer, \\ \fixTwistType_\blurElem(\blurElem(\autoA(\tupA_\scenter))) = \autoA(\tupB_\scenter)}}
	S^{\blurElem} (\blurElem(\autoA(\remZ{\tupA})), \fixTwistType_\inv{\blurElem}(\autoA(\remZ{\tupB})))\\
	&=\sum\limits_{\substack{\blurElem \in \blurrer, \\ \fixTwistType_\blurElem(\blurElem(\tupA_\scenter)) = \tupB_\scenter}}
	S^{\blurElem} (\blurElem(\remZ{\tupA}), \fixTwistType_\inv{\blurElem}(\remZ{\tupB}))
	= S(\tupA, \tupB).\qedhere
	\end{align*}
\end{claimproof}

\begin{claim}
	The matrix~$S$ is odd-filled.
\end{claim}
\begin{claimproof}
	Let $\BlockA \in \PartA_k$ and $\tupA \in \BlockA$.
	Then $\fixTwistType(\tupA) \in \fixTwistType(\BlockA)$
	and because every $\blurElem \in \blurrer$ is an orbit-automorphism
	(Claim~\ref{clm:blurelem-blockiso}),
	it holds that
	$\blurElem(\tupA) \in \BlockA$ and
	$\fixTwistType(\blurElem(\tupA)) \in \fixTwistType(\BlockA)$
	for every $\blurElem \in \blurrer$, too.

	Assume first that~$\BlockA$ is blurrable.
	Now, we sum up (over~$\FF_2$) the entries in the row indexed by~$\tupA$:
	\[
	\sum_{\tupB \in \fixTwistType(\BlockA)} S(\tupA,\tupB) =
	\sum_{\tupB \in \fixTwistType(\BlockA)} 
	\sum
	\limits_{\substack{\blurElem \in \blurrer, \\ \fixTwistType(\blurElem(\tupA)) = \tupB}} 1 = \sum
	\limits_{\substack{\blurElem \in \blurrer, \\ \fixTwistType(\blurElem(\tupA)) \in \fixTwistType(\BlockA)}} 1  
	=|\blurrer| \bmod 2.\]
	The last step holds, because -- as seen -- $\fixTwistType(\blurElem(\tupA)) \in \fixTwistType(\BlockA)$ for every $\blurElem \in \blurrer$.
	Finally,~$|\blurrer|$ is odd by Lemma~\ref{lem:blurrer-size-odd}
	and so the number of ones in the row indexed by~$\tupA$ is odd, too.
	
	Assume otherwise that~$\BlockA$  is not blurrable
	and set $\BlockB := \fixTwistType(\BlockA)$,
	which is of the same type as~$\BlockA$ (Corollary~\ref{cor:tau-fixes-orbit-type}).
	For every $\blurElem \in \blurrer$
	we set $\BlockB_\blurElem := \setcond{\remZ{\tupB}}{\tupB \in \BlockB, 
	\fixTwistType_\blurElem(\blurElem(\tupA_\scenter)) = \tupB_\scenter }$.
	Then $\BlockB_\blurElem 
	\in \orbs{k-1}{(\StructA_g,\pTupA\pVertA_\scenter)} = \orbs{k-1}{(\StructA_{g-\blurElem},\pTupA\pVertA_\scenter)}$ 
	by Corollaries~\ref{cor:k-orbit-fix-vertex} and~\ref{cor:cfi-same-autgroup-and-orbits} 
	(the center~$\scenter$ has distance greater than~$1$ to $\orig{\pTupA}$).
	\begin{align*}
	\sum_{\tupB \in \BlockB} S(\tupA,\tupB)
	&=
	\sum_{\tupB \in \BlockB} 
	\sum\limits_{\substack{\blurElem \in \blurrer, \\ \fixTwistType_\blurElem(\blurElem(\tupA_\scenter)) = \tupB_\scenter}}
S^{\blurElem} (\blurElem(\remZ{\tupA}), \inv{\fixTwistType_\blurElem}(\remZ{\tupB})) \\
	&= \sum_{\blurElem \in \blurrer}
\sum\limits_{\substack{\tupB \in \BlockB, \\ \fixTwistType_\blurElem(\blurElem(\tupA_\scenter)) = \tupB_\scenter}}
S^{\blurElem} (\blurElem(\remZ{\tupA}), \inv{\fixTwistType_\blurElem}(\remZ{\tupB})) \\
	&= \sum_{\blurElem \in \blurrer}
	\sum_{\tupB \in \BlockB_\blurElem}
	S^{\blurElem} (\blurElem(\remZ{\tupA}), \inv{\fixTwistType_\blurElem}(\tupB))\\
	&= \sum_{\blurElem \in \blurrer}
	\sum_{\tupB \in \inv{\fixTwistType_\blurElem}(\BlockB_\blurElem)}
	S^{\blurElem} (\blurElem(\remZ{\tupA}), \tupB). \tag{$\star$}
	\end{align*}
	In the first line of the equation,
	$\blurElem(\remZ{\tupA})$
	has the same type in $(\Struct_f, \pTupA\pVertA_\scenter)$
	as $\inv{\fixTwistType_\blurElem}(\remZ{\tupB})$
	has in $(\Struct_{g-\blurElem}, \pTupA\pVertA_\scenter)$
	for every $\blurElem \in \blurrer$ such that $\fixTwistType_\blurElem(\blurElem(\tupA_\scenter)) = \tupB_\scenter$
	(Claim~\ref{clm:matrix-blurrer-tuple-same-type}).
	One sees that we always sum over the same column indices of~$S^\blurElem$ (for a fixed~$\blurElem$)
	and we only manipulate the way in which we express the sum.
	Hence, in the last line,
	$\blurElem(\remZ{\tupA})$ has the same type in $(\Struct_f, \pTupA\pVertA_\scenter)$ as $\tupB$ has in $(\Struct_{g-\blurElem}, \pTupA\pVertA_\scenter)$.
	By Claim~\ref{clm:fxTwistType-fixes-type-general-rec},
	$\inv{\fixTwistType_\blurElem}(\BlockB_\blurElem) \in \orbs{k-1}{(\StructA_{g-\blurElem},\pTupA\pVertA_\scenter)}$,
	so $\sum_{\tupB \in \inv{\fixTwistType_\blurElem}(\BlockB_\blurElem)}
	S^{\blurElem} (\blurElem(\remZ{\tupA}), \tupB) = 1$ because we sum over all entries in one row of a block on the diagonal of~$S^\blurElem$
	and~$S^\blurElem$ is odd-filled (Claim~\ref{clm:blur-twist-blur-elem}).
	It follows that
	\begin{align*}
	(\star) &= 
	\sum\limits_{\blurElem \in \blurrer} 1 = |\blurrer| \bmod 2 = 1.
	\end{align*}
	For the last step we used again that~$|\blurrer|$ is odd by Lemma~\ref{lem:blurrer-size-odd}.
\end{claimproof}
Now it follows from Lemma~\ref{lem:orbit-diagonal+orbit-invariant+odd-filled-implies-invertible}
that~$S$ is invertible.

\begin{claim}
	$\activeRegionIdx{f,g}{\pTupA}{S} \subseteq \neighborsK{G}{\minRad(k+1)}{\stip}$.
\end{claim}
\begin{claimproof}
	We show that $\neighborsK{G}{\minRad(k+1)}{\stip}$ satisfies the conditions of the active region.
	This implies that $\activeRegionIdx{f,g}{\pTupA}{S} \subseteq \neighborsK{G}{\minRad(k+1)}{\stip}$.
	To show Condition~\ref{itm:active-region-contaied} of the active region,
	we have to show that $\comp \subseteq \neighborsK{G}{\minRad(k+1)}{t}$
	for every component $\comp \in \activeRegionBIdx{f,g}{\pTupA}{S}{\BlockA}$
	for every $\BlockA \in \PartA_k$.
	
	Let $\BlockA \in \PartA_k$ be blurrable
	and~$\comp$ be a component of~$\BlockA$.
	By definition of~$\blurElem$ and~$\fixTwistType$,
	the matrix~$S$ is only active on~$\comp$
	if~$\comp$ is a star or a tip component.
	But this means that $\comp \subseteq \neighborsK{G}{\minRad(k)+2 + k}{\scenter}$
	because~$\comp$ is connected, of size at most~$k$,
	and contains a vertex~$\bvertA$ of some~$\spath_i$
	(which have length $\minRad(k)+2$).
	Because $\distance{G}{\scenter}{\stip} = \minRad(k) +2$,
	every vertex with distance at most $\minRad(k)+2 + k$ to $\scenter$
	has distance at most $2 \minRad(k) +4 + k \leq 4\minRad(k) + 2= \minRad(k+1)$
	to~$\stip$
	(one immediately sees that $\minRad(k) \geq  k \geq 2$).
	Thus,
	\[\comp \subseteq \neighborsK{G}{\minRad(k)+2 + k}{\scenter} \subseteq \neighborsK{G}{\minRad(k+1)}{\stip}.\] 
	
	Let otherwise $\BlockA \in \PartA_k$ be non-blurrable and~$\comp$ be a component of~$\BlockA$.
	Then~$S$ is possibly active on~$\comp$ if $\comp \subseteq \neighborsK{G}{\minRad(k)+2 + k}{\stip}$ (as seen in the blurrable case)
	or
	\[\comp \subseteq \activeRegionIdx{f, g-\blurElem}{\pTupA\pVertA_\scenter}{S^\blurElem} \subseteq \bigcup_{i \in d} N_i\]
	for some $\blurElem \in \blurrer$
	(Claim~\ref{clm:blur-twist-blur-elem}).
	Recall that $N_i = \neighborsK{G}{\minRad(k)}{\stip_i}$
	and that $\distance{G}{\stip_i}{\scenter} = \minRad(k) + 2$ by construction.
	If follows that
	\[\bigcup_{i \in d} N_i \subseteq \neighborsK{G}{4\minRad(k) + 2}{\scenter} = \neighborsK{G}{\minRad(k+1)}{\stip}\]
	because every vertex with distance $\minRad(k)$ to some~$\stip_i$
	has distance at most $3\minRad(k)+4 \leq 4\minRad(k)+2$ to~$\stip$
	(we again use that $\minRad(k) \geq 2$ for $k\geq 2$).
	
	To prove Condition~\ref{itm:active-region-replace} of the active region,
	we see that~$\blurElem$,~$\fixTwistType$, and~$\fixTwistType_\blurElem$
	are defined component-wise (for every $\blurElem \in \blurrer$)
	and that  $\activeRegionIdx{f, g-\blurElem}{\pTupA\pVertA_\scenter}{S^\blurElem} \subseteq \bigcup_{i \in d} N_i$.
\end{claimproof}

Now, we want to show that~$S$ actually $k$-blurs the twist.
For $\BlockA \in \PartA_{2k}$ we set ${\BlockA_1 := \restrictVect{\BlockA}{\set{1,\dots, k}}}\in \PartA_k$ and
${\BlockA_2 := \restrictVect{\BlockA}{\set{k+1,\dots, 2k}}} \in \PartA_k$
to be the unique $k$-orbits
such that $\BlockA \subseteq \BlockA_1 \times \BlockA_2$
(and similar for a $\BlockB \in \PartB_{2k}$).
Our aim is to prove that $\charMat{\BlockA} \cdot S = S \cdot \charMat{\BlockB}$ for every $\BlockA \in \PartA_{2k}$ and $\BlockB := \fixTwistType(\BlockA)$.
Because~$S$ is orbit-diagonal and 
$\charMat{\BlockA}_{\BlockA_1 \times \BlockA_2}$ and
$\charMat{\BlockB}_{\BlockB_1 \times \BlockB_2}$
are the only nonzero blocks of~$\charMat{\BlockA}$ and~$\charMat{\BlockB}$,
it suffices to show that
$\charMat{\BlockA}_{\BlockA_1 \times \BlockA_2} S_{\BlockA_2 \times \BlockB_2} = S_{\BlockA_1 \times \BlockB_1} \charMat{\BlockB}_{\BlockB_1 \times \BlockB_2}$.
We begin with blurrable orbits and define the set of indices $i \in [d]$ of the blurrer~$\blurrer$ which are relevant for a blurrable orbit.
\begin{definition}
	The \defining{occupied indices $\occIndices{\BlockA}$}
	of a blurrable orbit $\BlockA \in \PartA_k$
	is the set of indices $i \in [d]$,
	such that there is a component~$\comp$ of~$\BlockA$
	that satisfies $\comp \subseteq N_i$  or $\comp$ is an $i$-star component.
\end{definition}
Note that the definition also covers $i$-tip components
because every $i$-tip component is contained in~$N_i$.
Also note that $|\occIndices{\BlockA}| \leq k$
because~$\BlockA$ is a $k$-orbit
and thus we can use the blurrer properties later.
The following lemma states that if
one of~$\BlockA_1$ and~$\BlockA_2$ is blurrable, say~$\BlockA_1$,
then it does not matter whether we apply~$\blurElem$ or~$\blurElem'$
to~$\BlockA_2$ as long as~$\blurElem$ and~$\blurElem'$ agree on $\occIndices{\BlockA_1}$.
\begin{claim}
	\label{clm:blurrable-depends-on-occupied}
	Suppose $\BlockA \in \PartA_{2k}$,
	$\blurElem, \blurElem' \in \blurrer$, and
	$\tupA,\tupB \in \StructVA^k$.
	If $\BlockA_1 \in \PartA_k$ is blurrable and $\restrictVect{\blurElem}{\occIndices{\BlockA_1}} = \restrictVect{\blurElem'}{\occIndices{\BlockA_1}}$,
	then 
	$\tupA\blurElem(\tupB) \in \BlockA$
	if and only if $\tupA\blurElem'(\tupB) \in \BlockA$.
	Likewise,
	if~$\BlockA_2$ is blurrable
	and $\restrictVect{\blurElem}{\occIndices{\BlockA_2}} = \restrictVect{\blurElem'}{\occIndices{\BlockA_2}}$,
	then 
	$\blurElem(\tupA)\tupB \in \BlockA$
	if and only if $\blurElem'(\tupA)\tupB \in \BlockA$.
	The same holds for ${\BlockB := \fixTwistType(\BlockA)}$.
\end{claim}
\begin{claimproof}
	Consider the case that~$\BlockA_1$ is blurrable.
	Assume that $\tupA\blurElem(\tupB) \in \BlockA$.
	Let $\comp^i_1, \dots, \comp^i_{\ell_i}$ be the components of~$\BlockA_i$ for every $i \in [2]$.
	Because~$\fixTwistType$ and~$\blurElem$ are defined component-wise,
	it suffices by Lemma~\ref{lem:split-discon-orbits} to
	assume that~$\BlockA$ is a $k'$-orbit of $(\StructA_f, \pTupA)$ for some $k'\leq k$
	and has a single component~$\comp$.
	The component~$\comp$ is the union of all~$\comp^i_j$.
	
	First consider the case
	that~$\comp$ is a star component.
	Because~$\BlockA_1$ is blurrable,
	we can partition all components~$\comp^1_j$ by Claim~\ref{clm:star-tip-components-basics}
	into $D^1_1, \dots, D^1_d$ and~$D^1_R$
	such that~$D^1_i$ contains all $i$-star components for all $i \in[d]$
	and~$D^1_R$ the remaining sky components.
	Set $\blurElem'' := \blurElem' - \blurElem$.
	Then $\restrictVect{\blurElem''}{\occIndices{\BlockA_1}} = 0$,
	$\blurElem''(\blurElem(\tupB)) = \blurElem'(\tupB)$, and
	$\blurElem''$ is the identity function on the components in all~$D^1_i$
	because every~$D^1_i$ only contains $i$-star components
	and is nonempty only if $i \in \occIndices{\BlockA_1}$.
	Additionally,~$\blurElem''$ is the identity on the sky components~$D^1_R$.
	Hence $\autoA_{\blurElem''}(\tupA) = \tupA$
	by Definition~\ref{def:stariso}.
	Because~$\comp$ is a star component, we have that
	\begin{align*}
		\blurElem''(\tupA \blurElem(\tupB))
		= \autoA_{\blurElem''}(\tupA \blurElem(\tupB))
		= \autoA_{\blurElem''}(\tupA)\autoA_{\blurElem''}(\blurElem(\tupB))
		= \tupA\blurElem''(\blurElem(\tupB))
		= \tupA\blurElem'(\tupB).
	\end{align*}
	Because~$\blurElem''$ is an orbit-automorphism by Claim~\ref{clm:blurelem-blockiso},
	it holds that $\tupA \blurElem(\tupB) \in \BlockA$
	if and only if $\blurElem''(\tupA \blurElem(\tupB)) = \tupA\blurElem'(\tupB) \in \BlockA$.
	
	If otherwise~$\comp$ is not a star component,
	then none of the~$\comp^i_j$ is a star component
	and~$\blurElem$ and~$\blurElem'$ are the identity function on~$\tupB$
	and the claim follows immediately.
	The cases for~$\BlockA_2$ and~$\BlockB$ are analogous.
\end{claimproof}

\begin{claim}
	\label{clm:blurrable-repair-orbits}
	For every $\BlockA \in \PartA_{2k}$,
	$\blurElem \in \blurrer$ such that
	$\restrictVect{\blurElem}{\occIndices{\BlockA}} = \restrictVect{\blurElemFix}{\occIndices{\BlockA}}$, and
	$\tupA, \tupB\in \StructVA^k$ it holds that
	$\tupA\tupB \in \BlockA$ if and only if $\fixTwistType(\blurElem(\tupA)) \fixTwistType(\blurElem(\tupB)) \in \fixTwistType(\BlockA)$.
\end{claim}
\begin{claimproof}
	Let $\BlockA \in \PartA_{2k}$,
	$\blurElem \in \blurrer$ such that
	$\restrictVect{\blurElem}{\occIndices{\BlockA}} = \restrictVect{\blurElemFix}{\occIndices{\BlockA}}$, and 
	$\tupA, \tupB\in \StructVA^k$.
	Using Claim~\ref{clm:blurrable-repair-orbits-general}
	we obtain $\tupA\tupB \in \BlockA$ if and only if $\fixTwistType_\blurElem(\blurElem(\tupA)) \fixTwistType_\blurElem(\blurElem(\tupB)) \in \fixTwistType_\blurElem(\BlockA)$.
	Because $\restrictVect{\blurElem}{\occIndices{\BlockA}} = \restrictVect{\blurElemFix}{\occIndices{\BlockA}}$,
	the action of~$\fixTwistType_\blurElem$ and~$\fixTwistType$ is equal on the $i$-tip components of $\BlockA$ for $i \in \occIndices{\BlockA}$.
	Thus, $\fixTwistType_\blurElem(\BlockA) = \fixTwistType(\BlockA)$.
	Similarly, $\fixTwistType_\blurElem(\blurElem(\tupA)) = \fixTwistType(\blurElem(\tupA))$ and
	$\fixTwistType_\blurElem(\blurElem(\tupB)) = \fixTwistType(\blurElem(\tupB))$.
\end{claimproof}

For every blurrable orbit $\BlockA \in \PartA_k$
it holds that $|\occIndices{\BlockA}| \leq k$
and so we can apply the blurrer properties as follows:

\begin{claim}
	\label{clm:blurrable-case}
	Let $\BlockA \in \PartA_{2k}$ and $\BlockB = \fixTwistType(\BlockA)$.
	If~$\BlockA_1$  (and so~$\BlockB_1$) and~$\BlockA_2$ (and so~$\BlockB_2$) are blurrable,
	then $\charMat{\BlockA} \cdot S = S \cdot \charMat{\BlockB}$.
\end{claim}
\begin{claimproof}
	With the definition of~$S$ on blurrable orbits we obtain for $\tupA \in \BlockA_1$ and $\tupB \in \BlockB_2$ that
	\begin{align*}
		 (\charMat{\BlockA} \cdot S) (\tupA, \tupB)
		&= \sum_{\tupC \in \BlockA_2} \charMat{\BlockA}(\tupA,\tupC) \cdot S_{\BlockA_2 \times \BlockB_2} (\tupC, \tupB)\\
		&= \sum_{\tupC \in \BlockA_2} \charMat{\BlockA}(\tupA,\tupC) \cdot 
		\sum
		\limits_{\substack{\blurElem \in \blurrer, \\ \fixTwistType(\blurElem(\tupC)) = \tupB}} 1\\
		&= \sum_{\blurElem \in \blurrer} \charMat{\BlockA}(\tupA,\inv{\blurElem}(\inv{\fixTwistType}(\tupB))). \tag{$\star$}
	\end{align*}
	Now for every $\blurElem \in \blurrer$, $\charMat{\BlockA}(\tupA,\inv{\blurElem}(\inv{\fixTwistType}(\tupB)))$
	depends only on $\restrictVect{\blurElem}{\occIndices{\BlockA_1}}$
	by Claim~\ref{clm:blurrable-depends-on-occupied},
	i.e., $\blurElem \mapsto \charMat{\BlockA}(\tupA,\inv{\blurElem}(\inv{\fixTwistType}(\tupB)))$
	is actually a function $\restrictVect{\blurrer}{\occIndices{\BlockA_1}} \to \FF_2$.
	Because~$\BlockA_1$ is a blurrable $k$-orbit,
	it holds that $|\occIndices{\BlockA_1}|\leq k$.
	Then, by Lemma~\ref{lem:blurrer-sum},
	it follows that for some $\sblurElem \in \blurrer$ with $\restrictVect{\sblurElem}{\occIndices{\BlockA_1}} = \restrictVect{\blurElemFix}{\occIndices{\BlockA_1}}$ we have that
	\begin{align*}
		(\star) &=
		\charMat{\BlockA}(\tupA,\invsblurElem(\inv{\fixTwistType}(\tupB)))\\
		&= \charMat{\BlockB}(\fixTwistType(\sblurElem(\tupA)),\tupB)\\
		&= \sum_{\blurElem \in \blurrer} \charMat{\BlockB}(\fixTwistType(\blurElem(\tupA)),\tupB)\\
		&= \sum_{\tupC \in \BlockB_1} S_{\BlockA_1 \times \BlockB_1} (\tupA, \tupC) \cdot \charMat{\BlockB}(\tupC,\tupB)\\
		&= (S\cdot \charMat{\BlockB}) (\tupA, \tupB),
	\end{align*}
	where the transition from~$\BlockA$ to~$\BlockB$ is by Claim~\ref{clm:blurrable-repair-orbits}.
	The last step is the inverse reasoning as for~$\BlockA$
	using Claim~\ref{clm:blurrable-depends-on-occupied}.
\end{claimproof}

Next, we want to consider the case that~$\BlockA_1$ is blurrable
and~$\BlockA_2$ is not (or vice versa, which is symmetric).
We would like to use Claim~\ref{clm:blurrable-repair-orbits}
as in the case that both~$\BlockA_1$ and~$\BlockA_2$
are blurrable.
But this is not sufficient because
$S_{\BlockA_2 \times \fixTwistType(\BlockA_2)}$
is defined using the matrices~$S^\blurElem$.
We show that the matrices~$S^\blurElem$
cancel out in this case.
Intuitively, the idea is to use that
the active regions of the $\Sblurelem{\blurElem}{i}$ are disjoint
(Claim~\ref{clm:neighborhoods_different}).
As a consequence the matrices $\Sblurelem{\blurElem}{i}$
for all $i \notin \occIndices{\BlockA_1}$
cancel out.
For the remaining $\Sblurelem{\blurElem}{i}$
we use the blurrer properties to show that they vanish.

\begin{claim}
	\label{clm:blurrable-non-blurrable-case}
	Suppose $\BlockA \in \PartA_{2k}$ and $\BlockB = \fixTwistType(\BlockA)$.
	If $\BlockA_1$ is blurrable and $\BlockA_2$ is not,
	then $\charMat{\BlockA} \cdot S = S \cdot \charMat{\BlockB}$.
\end{claim}
\begin{claimproof}
	We first unfold the definition of~$S$.
	Let $\tupA \in \BlockA_1$ and $\tupB\in \BlockB_2$.
	\begin{align*}
		(\charMat{\BlockA} \cdot S) (\tupA, \tupB)
		&= \sum_{\tupC \in \BlockA_2} \charMat{\BlockA}(\tupA,\tupC) \cdot S_{\BlockA_2 \times \BlockB_2} (\tupC, \tupB)\\
		&= \sum_{\tupC \in \BlockA_2}
		\charMat{\BlockA}(\tupA,\tupC) \cdot
		\sum_{\substack{\blurElem \in \blurrer,\\ \blurElem(\tupC_\scenter) = \inv{\fixTwistType_\blurElem}(\tupB_\scenter)}}
		S^\blurElem(\blurElem(\remZ{\tupC}), \inv{\fixTwistType_\blurElem}(\remZ{\tupB}))\\
		&= \sum_{\blurElem \in \blurrer}
		\sum_{\substack{\tupC \in \BlockA_2,\\ \blurElem(\tupC_\scenter) = \inv{\fixTwistType_\blurElem}(\tupB_\scenter)}}
		\charMat{\BlockA}(\tupA,\tupC) \cdot
		S^\blurElem(\blurElem(\remZ{\tupC}), \inv{\fixTwistType_\blurElem}(\remZ{\tupB})).\tag{$\star$}
	\end{align*}
	For every $\blurElem\in\blurrer$, we set $\BlockA'_{\blurElem,2} := \setcond{\remZ{\tupC}}{\tupC \in \BlockA_2, \tupC_\scenter = \blurElem(\tupB_\scenter)}$
	(note that $\inv{\fixTwistType_\blurElem}(\tupB_\scenter) = \tupB_\scenter$
	because $\orig{\tupB_\scenter} = \scenter$).
	Here $\tupC\inv{\blurElem}(\inv{\fixTwistType_\blurElem}(\tupB_\scenter))$
	denotes the tuple~$\tupC'$ such that 
	$\remZ{\tupC'} = \tupC$ and $\tupC'_\scenter = \inv{\blurElem}(\inv{\fixTwistType_\blurElem}(\tupB_\scenter))$.
	It holds that $\BlockA'_{\blurElem,2} \in \orbs{k-1}{(\StructA_f, \pTupA\pVertA_\scenter)}$ by Corollary~\ref{cor:k-orbit-fix-vertex}.
	Then
	\begin{align*}
		(\star)&= \sum_{\blurElem \in \blurrer} \sum_{\tupC \in \BlockA'_{\blurElem,2}} 
		\charMat{\BlockA}(\tupA,\tupC\inv{\blurElem}(\inv{\fixTwistType_\blurElem}(\tupB_\scenter))) \cdot
		S^\blurElem(\blurElem(\tupC), \inv{\fixTwistType_\blurElem}(\remZ{\tupB})).
	\end{align*}
	Let $K\subseteq[d]$ be the maximal set of indices,
	such that there is no component~$\comp$ of~$\BlockA_1$ satisfying $\comp\subseteq N_i$.
	So in particular $\occIndices{\BlockA_1} \subseteq [d] \setminus K$.
	We partition the components of~$\BlockA_2$ as follows:
	\begin{itemize}
		\item Let~$D$ be the set of components~$\comp$ of~$\BlockA_2$ that are also components of~$\BlockA$
		and satisfy $\comp \subseteq N_i$ for some $i\in K$.
		\item Let~$E$ be the set of~$\comp$ of~$\BlockA_2$
		that are not contained in~$D$ and
		satisfy $\comp \subseteq N_i$ for some $i\in[d]$.
		\item Let~$R$ be the set of all remaining components of $\BlockA_2$.
	\end{itemize}
	We split $\tupC = \tupC_D\tupC_E\tupC_R$ into the components belonging to~$D$,~$E$, and~$R$.
	We split ${\tupB = \tupB_D\tupB_E\tupB_R\tupB_\scenter}$ likewise,
	where, for simplicity, we set $\tupB_D := \remZ{\tupB}_D$
	and similar for~$\tupB_E$ and~$\tupB_R$. 
	From Lemma~\ref{lem:split-discon-orbits} it follows that
	$\BlockA'_{\blurElem,2} = \restrictVect{\BlockA'_{\blurElem,2}}{D} \times \restrictVect{\BlockA'_{\blurElem,2}}{E \cup R}$
	and
	$\BlockA = \restrictVect{\BlockA}{D} \times \restrictVect{\BlockA}{R'}$,
	where~$R'$ are the components of~$\BlockA$ not contained in~$D$.
	By the definition of~$K$ and~$D$, we have that 
	$\restrictVect{\BlockA}{D}=\restrictVect{\BlockA'_{\blurElem,2}}{D}$
	because the components in~$D$ are components of $\BlockA$,
	are disjoint with $\orig{\BlockA_1}$,
	and do not contain $\scenter$ ($\scenter \notin N_i$ for all $i \in [d]$).
	We obtain that
	\begin{align*}
		&\hiddenEq(\star)\\
		&=\sum_{\blurElem \in \blurrer}
		\sum_{\tupC_D \in \restrictVect{\BlockA'_{\blurElem,2}}{D}}
		\sum_{\tupC_E\tupC_R \in \restrictVect{\BlockA'_{\blurElem,2}}{E\cup R}}
		\charMat{\BlockA}(\tupA,\tupC_D\tupC_E\tupC_R\inv{\blurElem}(\inv{\fixTwistType_\blurElem}(\tupB_\scenter))) \cdot
		S^\blurElem(\blurElem(\tupC_D\tupC_E\tupC_R), \inv{\fixTwistType_\blurElem}(\tupB_D\tupB_E\tupB_R))\\
		&=\sum_{\blurElem \in \blurrer}
		\sum_{\tupC_E\tupC_R \in \restrictVect{\BlockA'_{\blurElem,2}}{E\cup R}}
		\charMat{\restrictVect{\BlockA}{R'}}(\tupA,\tupC_E\tupC_R\inv{\blurElem}(\inv{\fixTwistType_\blurElem}(\tupB_\scenter))) \cdot
		\sum_{\tupC_D \in \restrictVect{\BlockA'_{\blurElem,2}}{D}}
		S^\blurElem(\blurElem(\tupC_D\tupC_E\tupC_R), \inv{\fixTwistType_\blurElem}(\tupB_D\tupB_E\tupB_R)).
	\end{align*}
	By Claim~\ref{clm:component-wise-application}, 
	the functions~$\blurElem$ and~$\fixTwistType_\blurElem$ can be applied component-wise.
	For every~$\blurElem$, the matrix~$S^\blurElem$ is not active on
	the components in~$R$ (by definition of~$R$ and Claim~\ref{clm:blur-twist-blur-elem}).
	Thus, $\blurElem(\tupC_R) = \inv{\fixTwistType_\blurElem}(\tupB_R)$
	unless $S^\blurElem(\blurElem(\tupC_D\tupC_E\tupC_R),  \inv{\fixTwistType_\blurElem}(\tupB_D\tupB_E\tupB_R)) = 0$.
	We again use Lemma~\ref{lem:split-discon-orbits}
	to split $\restrictVect{\BlockA'_{\blurElem,2}}{E\cup R} = \restrictVect{\BlockA'_{\blurElem,2}}{E}\times \restrictVect{\BlockA'_{\blurElem,2}}{R}$.
	Such a split of~$\BlockA$ is not possible
	because the components of~$\BlockA'_{\blurElem,2}$
	in~$E$ and~$R$ might not be components of~$\BlockA$.
	Here we have $\inv{\blurElem}(\inv{\fixTwistType_\blurElem}(\tupB_R))\inv{\blurElem}(\inv{\fixTwistType_\blurElem}(\tupB_\scenter)) = \inv{\blurElem}(\inv{\fixTwistType_\blurElem}(\tupB_R\tupB_\scenter))$
	because if $\comp\cup \set{\scenter}$ is a component,
	then $\comp \cup \set{\scenter}$ is a star component
	and~$\comp$ is a union of star and sky components. Thus,
	\begin{align*}
		&\hiddenEq(\star)\\
		&=\sum_{\blurElem \in \blurrer}
		\sum_{\tupC_E \in \restrictVect{\BlockA'_{\blurElem,2}}{E}}
		\charMat{\restrictVect{\BlockA}{R'}}(\tupA,\tupC_E\inv{\blurElem}(\inv{\fixTwistType_\blurElem}(\tupB_R\tupB_\scenter))) \cdot
		\sum_{\tupC_D \in \restrictVect{\BlockA'_{\blurElem,2}}{D}}
		S^\blurElem(\blurElem(\tupC_D\tupC_E)\inv{\fixTwistType_\blurElem}(\tupB_R),  \inv{\fixTwistType_\blurElem}(\tupB_D\tupB_E\tupB_R)).
	\end{align*}
	We know that $\restrictVect{\BlockA'_{\blurElem,2}}{D} \in \orbs{k'}{(\StructA_f,\pTupA\pVertA_\scenter)}$ for some $k' \leq k-1$
	by Lemma~\ref{lem:split-discon-orbits}.
	Because~$\blurElem$ is an orbit-automorphism for every $\blurElem \in \blurrer$ (Claim~\ref{clm:blurelem-blockiso}) and~$D$ contains no star center components,
	every $\blurElem \in \blurrer$ permutes $\restrictVect{\BlockA'_{\blurElem,2}}{D}$.
	Hence, we can sum over $\blurElem(\tupC_D) \in \restrictVect{\BlockA'_{\blurElem,2}}{D}$
	instead over $\tupC_D \in \restrictVect{\BlockA'_{\blurElem,2}}{D}$.
	Then we can apply Claim~\ref{clm:sum-s-blur-subtree}:
	\begin{align*}
		(\star) &=\sum_{\blurElem \in \blurrer}
		\sum_{\tupC_E \in \restrictVect{\BlockA'_{\blurElem,2}}{E}}
		\charMat{\restrictVect{\BlockA}{R'}}(\tupA,\tupC_E\inv{\blurElem}(\inv{\fixTwistType_\blurElem}(\tupB_R\tupB_\scenter))) \cdot {}
		\\
		&\hspace{4cm}
		\left(\prod_{i \in [d] \setminus K}\Sblurelem{\blurElem}{i}\right)(\inv{\fixTwistType_\blurElem}(\tupB_D)\blurElem(\tupC_E)\inv{\fixTwistType_\blurElem}(\tupB_R),  \inv{\fixTwistType_\blurElem}(\tupB_D\tupB_E\tupB_R)).
	\end{align*}
	We show that this term only depends on $\restrictVect{\blurElem}{[d]\setminus K}$ as follows: Let $\blurElem, \blurElem' \in \blurrer$ such that 
	$\restrictVect{\blurElem}{[d]\setminus K} = \restrictVect{\blurElem'}{[d]\setminus K}$.
	\begin{enumerate}[label=(\alph*)]
		\item Consider the right term  $(\prod_{i \in [d] \setminus K}\Sblurelem{\blurElem}{i})(\cdot,\cdot)$ in the equation.
		First, $\Sblurelem{\blurElem}{i} = \Sblurelem{\blurElem'}{i}$
		for every $i \in [d] \setminus K$
		by construction of the matrices~$\Sblurelem{\blurElem}{i}$
		because $\blurElem(i) = \blurElem'(i)$.
		Because~$R$ does not contain tip components,
		it holds that $\fixTwistType_\blurElem(\tupB_R) = \tupB_R$.
		Second, all components in~$E$ are contained in some~$N_i$ for $i \in [d] \setminus K$
		by definition of~$E$.
		That is, $\blurElem(\tupC_E) = \blurElem'(\tupC_E)$
		and $\fixTwistType_\blurElem(\tupB_E) = \fixTwistType_{\blurElem'}(\tupB_E)$.
		Third, the active region of $\prod_{i \in [d] \setminus K}\Sblurelem{\blurElem}{i}$ is bounded by $\bigcup_{i \in [d]\setminus K} N_i$ by Claim~\ref{clm:sum-s-blur-subtree}.
		Because components in $D \cup R$ are not contained in $\bigcup_{i \in [d]\setminus K} N_i$ by definition of~$D$ and~$R$,
		we can exploit Condition~\ref{itm:active-region-replace} of the active region and apply~$\fixTwistType_\blurElem$ 
		to~$\tupB_D$ and~$\tupB_R$ on both sides
		(because~$\fixTwistType_\blurElem$ is a bijection),
		that is, 
		\begin{align*}
			&\hiddenEq\left(\prod_{i \in [d] \setminus K}\Sblurelem{\blurElem}{i}\right)(\inv{\fixTwistType_\blurElem}(\tupB_D)\blurElem(\tupC_E)\inv{\fixTwistType_\blurElem}(\tupB_R),  \inv{\fixTwistType_\blurElem}(\tupB_D\tupB_E\tupB_R))\\
			&=
		\left(\prod_{i \in [d] \setminus K}\Sblurelem{\blurElem}{i}\right)(\tupB_D\blurElem(\tupC_E)\tupB_R,  \tupB_D\inv{\fixTwistType_\blurElem}(\tupB_E)\tupB_R)\\
		&= \left(\prod_{i \in [d] \setminus K}\Sblurelem{\blurElem}{i}\right)(\tupB_D\blurElem'(\tupC_E)\tupB_R,  \tupB_D\inv{\fixTwistType_{\blurElem'}}(\tupB_E)\tupB_R)\\
		& = \left(\prod_{i \in [d] \setminus K}\Sblurelem{\blurElem'}{i}\right)(\inv{\fixTwistType_{\blurElem'}}(\tupB_D)\blurElem'(\tupC_E)\inv{\fixTwistType_{\blurElem'}}(\tupB_R),  \inv{\fixTwistType_{\blurElem'}}(\tupB_D\tupB_E\tupB_R)).
		\end{align*}
		\item Now consider the left term $\charMat{\restrictVect{\BlockA}{R'}}(\cdot,\cdot)$.
		First, note that~$R$ does not contain tip components and thus
		$\inv{\fixTwistType_\blurElem}(\tupB_R\tupB_\scenter) = \tupB_R\tupB_\scenter$.
		Second,
		$\tupA\tupC_E\inv{\blurElem}(\tupB_R\tupB_\scenter) \in \restrictVect{\BlockA}{R'}$
		if and only if 
		$\tupA\tupC_E\inv{\blurElem'}(\tupB_R\tupB_\scenter) \in \restrictVect{\BlockA}{R'}$
		by Claim~\ref{clm:blurrable-depends-on-occupied}:
		By repeating entries, one sees that 
		Claim~\ref{clm:blurrable-depends-on-occupied} also holds
		for $k'$-orbits with $k'\leq 2k$
		and a partition of $[k']$ into two parts each of size at most~$k$.
		Hence,
		\[\charMat{\restrictVect{\BlockA}{R'}}(\tupA,\tupC_E\inv{\blurElem}(\inv{\fixTwistType_\blurElem}(\tupB_R\tupB_\scenter))) =
		\charMat{\restrictVect{\BlockA}{R'}}(\tupA,\tupC_E\inv{\blurElem'}(\inv{\fixTwistType_{\blurElem'}}(\tupB_R\tupB_\scenter))).\]
	\end{enumerate}
	Hence, the Equation~$(\star)$ is of the form $\sum_{\blurElem \in \blurrer} h(\restrictVect{\blurElem}{[d]\setminus K})$ for some function $h \colon \restrictVect{\blurrer}{[d]\setminus K} \to \FF_2$.
	Because~$\BlockA_1$ contains $k$\nobreakdash-tuples,
	it holds that $|K|\geq d-k$ and hence that $|[d]\setminus K| \leq k$.
	That is, we can apply Lemma~\ref{lem:blurrer-sum}
	and obtain for some $\sblurElem \in \blurrer$,
	which satisfies that $\restrictVect{\sblurElem}{[d]\setminus K} = \restrictVect{\blurElemFix}{[d]\setminus K}$,
	the following:
	\begin{align*}
		(\star) &=
		\sum_{\tupC_E \in \restrictVect{\BlockA'_{\sblurElem,2}}{E}}
		\charMat{\restrictVect{\BlockA}{R'}}(\tupA,\tupC_E\invsblurElem(\inv{\fixTwistType_{\sblurElem}}(\tupB_R\tupB_\scenter))) \cdot{}
		\\
		&\hspace{3cm}
		\left(\prod_{i \in [d] \setminus K}\Sblurelem{\sblurElem}{i}\right)(\inv{\fixTwistType_{\sblurElem}}(\tupB_D)\sblurElem(\tupC_E)\inv{\fixTwistType_{\sblurElem}}(\tupB_R),  \inv{\fixTwistType_{\sblurElem}}(\tupB_D\tupB_E\tupB_R)).
	\end{align*}
	Now, the matrices $\Sblurelem{\sblurElem}{i}$ in the equation
	blur a twist of value~$0$ for every $i \in [d]\setminus K$:
	$\Sblurelem{\sblurElem}{i}$ blurs a twist of value $(g - \sblurElem - f)(\set{\stip_i, \primeSub{\stip}{i}}) =
	(\blurElemFix - \sblurElem)(i)$,
	which is~$0$ for every $i \in [d]\setminus K$
	because $\restrictVect{\sblurElem}{[d]\setminus K} = \restrictVect{\blurElemFix}{[d]\setminus K}$.
	That is,  $\Sblurelem{\sblurElem}{i} = \idmat$ by definition (cf.~Claim~\ref{clm:blur-twist-blur-elem-single})
	and $\sblurElem(\tupC_E) = \inv{\fixTwistType_{\sblurElem}}(\tupB_E)$ unless the right factor is zero.
	Hence,
	\begin{align*}
		(\star) &=
		\charMat{\restrictVect{\BlockA}{R'}}(\tupA,\invsblurElem(\inv{\fixTwistType_{\sblurElem}}(\tupB_E\tupB_R\tupB_\scenter))).
	\end{align*}
	On the components of~$E$
	the functions~$\fixTwistType$ and~$\fixTwistType_{\sblurElem}$
	act equally because $\restrictVect{\sblurElem}{[d]\setminus K} = \restrictVect{\blurElemFix}{[d]\setminus K}$.
	On~$R$, both are the identity.
	Hence, $\inv{\fixTwistType_{\sblurElem}}(\tupB_E\tupB_R\tupB_\scenter) = \inv{\fixTwistType}(\tupB_E\tupB_R\tupB_\scenter)$ and
	\begin{align*}
		(\star) &=
		\charMat{\restrictVect{\BlockA}{R'}}(\tupA,\invsblurElem(\inv{\fixTwistType}(\tupB_E\tupB_R\tupB_\scenter))).
	\end{align*}	
	We show that
	$\tupA\invsblurElem(\inv{\fixTwistType}(\tupB_E\tupB_R\tupB_\scenter)) \in \restrictVect{P}{R'}$
	if and only if
	$\tupA\invsblurElem(\inv{\fixTwistType}(\tupB_D\tupB_E\tupB_R\tupB_\scenter)) \in \BlockA$.
	This holds because~$D$ is a set of component of~$\BlockA$,
	$\tupB_D \in \restrictVect{\BlockB}{D}$, and 
	$\inv{\fixTwistType}(\tupB_D) \in \restrictVect{\BlockA}{D}$ if and only if $\tupB_D \in \restrictVect{\BlockB}{D}$ by Claim~\ref{clm:blurrable-repair-orbits}
	and $\inv{\fixTwistType}(\tupB_D) \in \restrictVect{\BlockA}{D}$
	if and only if $\invsblurElem(\inv{\fixTwistType}(\tupB_D)) \in \restrictVect{\BlockA}{D}$ by Claim~\ref{clm:blurelem-blockiso}.
	It follows that
	\begin{align*}	
		(\star) &= 
		\charMat{\BlockA}(\tupA,\invsblurElem(\inv{\fixTwistType}(\tupB_D\tupB_E\tupB_R\tupB_\scenter)))=
		\charMat{\BlockA}(\tupA,\invsblurElem(\inv{\fixTwistType}(\tupB))).
	\end{align*}
	We finish the proof similar to Claim~\ref{clm:blurrable-case} because~$\BlockA_1$ (and thus~$\BlockB_1$) is blurrable:
	\begin{align*}
		(\star) &=\charMat{\BlockB}(\fixTwistType(\sblurElem(\tupA)),\tupB)\\
		&= \sum_{\blurElem \in \blurrer} \charMat{\BlockB}(\fixTwistType(\blurElem(\tupA)),\tupB)\\
		&= \sum_{\tupC \in \BlockB_1} S_{\BlockA_1 \times \BlockB_1} (\tupA, \tupC) \cdot \charMat{\BlockB}(\tupC,\tupB)\\
		&= (S\cdot \charMat{\BlockB}) (\tupA, \tupB).\qedhere
	\end{align*}
\end{claimproof}
The case when~$\BlockA_2$ is blurrable and~$\BlockA_1$ is not
is analogous.
Finally, to solve the case where both~$\BlockA_1$ and~$\BlockA_2$ are non-blurrable,
we argue with the induction hypothesis for the matrices~$S^\blurElem$ (Claim~\ref{clm:blur-twist-blur-elem}).
It formally becomes elaborate for two reasons.
First, we need to argue precisely that the types of occurring orbits are the same.
Second, we need to treat the components containing~$\scenter$ differently
(Claim~\ref{clm:orbits-rec-blurelem})
because in the recursive step we have $\pVertA_\scenter$ as an additional
parameter and thus cannot apply the orbit-automorphisms $\blurElem \in \blurrer$ freely
(because they are not the identity function on~$\pVertA_\scenter$).
These components can be treated specially because they are not contained in the active region of the matrices~$S^\blurElem$.
\begin{claim}
	\label{clm:non-blurrable-case}
	Suppose $\BlockA \in \PartA_{2k}$ and $\BlockB = \fixTwistType(\BlockA)$.
	If~$\BlockA_1$ and~$\BlockA_2$ are non-blurrable,
	then $\charMat{\BlockA} \cdot S = S \cdot \charMat{\BlockB}$.
\end{claim}
\begin{claimproof}
	Let $\tupA \in \BlockA_1$ and $\tupB \in \BlockB_2$. We expand the definition of~$S$:
	\begin{align*}
	 	(\charMat{\BlockA} \cdot S) (\tupA, \tupB)
		&= \sum_{\tupC \in \BlockA_2} \charMat{\BlockA}(\tupA,\tupC) \cdot S_{\BlockA_2 \times \BlockB_2} (\tupC, \tupB)\\
		&= \sum_{\tupC \in \BlockA_2}
		\charMat{\BlockA}(\tupA,\tupC) \cdot
		\sum_{\substack{
				\blurElem \in \blurrer,\\ \blurElem(\tupC_\scenter) = \inv{\fixTwistType_\blurElem}(\tupB_\scenter)}}
		S^\blurElem(\blurElem(\remZ{\tupC}), \inv{\fixTwistType_\blurElem}(\remZ{\tupB}))\\
		&= 
		\sum_{\blurElem \in \blurrer}
		\sum_{\substack{\tupC \in \BlockA_2,\\
		\tupC_\scenter = \inv{\blurElem}(\inv{\fixTwistType_\blurElem}(\tupB_\scenter))}}
		\charMat{\BlockA}(\tupA,\tupC) \cdot
		S^\blurElem(\blurElem(\remZ{\tupC}), \inv{\fixTwistType_\blurElem}(\remZ{\tupB})) \tag{$\star$}.
	\end{align*}
	We define for every $\blurElem \in \blurrer$
	\begin{align*}
		\BlockA_\blurElem &:= \setcond[\big]{\remZ{\tupA'}\remZ{\tupC'}}{\tupA' \in \BlockA_1, \tupC' \in \BlockA_2, \tupA'\tupC' \in \BlockA, \primeSub{\tupA}{\scenter} = \tupA_\scenter, \primeSub{\tupC}{\scenter} = \inv{\blurElem}(\inv{\fixTwistType_\blurElem}(\tupB_\scenter))}, \\
		\BlockA_{\blurElem,2} &:= \setcond[\big]{\remZ{\tupC'}}{\tupC' \in \BlockA_2, \primeSub{\tupC}{\scenter} = \inv{\blurElem}(\inv{\fixTwistType_\blurElem}(\tupB_\scenter))}.
	\end{align*}
	By Lemma~\ref{lem:k-orbit-fix-vertex},
	we have that $\BlockA_\blurElem \in \orbs{2k-2}{(\StructA_f, \pTupA\pVertA_\scenter)} \cup \set{\emptyset}$
	and, by Corollary~\ref{cor:k-orbit-fix-vertex}, that $\BlockA_{\blurElem,2}  \in \orbs{k-1}{(\StructA_f, \pTupA\pVertA_\scenter)}$.
	It depends on $\blurElem\in\blurrer$
	whether the set~$\BlockA_\blurElem$ is empty.
	We continue the equation:
	\begin{align*}
		(\star) &=
		\sum_{\blurElem \in \blurrer}
		\sum_{\tupC \in \BlockA_{\blurElem,2}}
		\charMat{\BlockA_\blurElem}(\remZ{\tupA},\tupC) \cdot
		S^\blurElem(\blurElem(\tupC), \inv{\fixTwistType_\blurElem}(\remZ{\tupB})).
	\end{align*}
	We use that~$S^\blurElem$ is invariant under automorphism of $(\StructA_f, \pTupA)$ (Claim~\ref{clm:blur-twist-blur-elem-orbit-auto})
	and that~$\blurElem$ is an orbit-automorphism (Claim~\ref{clm:blurelem-blockiso}) for very $\blurElem \in \blurrer$.
	\begin{align*}
		(\star) &=
		\sum_{\blurElem \in \blurrer}
		\sum_{\tupC \in \BlockA_{\blurElem,2}}
		\charMat{\BlockA_\blurElem}(\remZ{\tupA},\tupC) \cdot
		S^\blurElem(\tupC, \inv{\blurElem}(\inv{\fixTwistType_\blurElem}(\remZ{\tupB})))\\
		&= \sum_{\blurElem \in \blurrer}
		(\charMat{\BlockA_\blurElem}\cdot S^\blurElem)(\remZ{\tupA},\inv{\blurElem}(\inv{\fixTwistType_\blurElem}(\remZ{\tupB})))\\
		&= \sum_{\blurElem \in \blurrer}
		(S^\blurElem \cdot \charMat{\BlockB_\blurElem})(\remZ{\tupA},\inv{\blurElem}(\inv{\fixTwistType_\blurElem}(\remZ{\tupB})))\\
		&= \sum_{\blurElem \in \blurrer}
		\sum_{\tupC \in \BlockB_{\blurElem,1}}
		S^\blurElem (\remZ{\tupA}, \tupC) \cdot  \charMat{\BlockB_\blurElem}(\tupC,\inv{\blurElem}(\inv{\fixTwistType_\blurElem}(\remZ{\tupB}))),
	\end{align*}
	where $\BlockB_\blurElem \in \orbs{2k-2}{(\StructA_{g-\blurElem},\pTupA\pVertA_\scenter)}$
	has the same type in $(\StructA_{g-\blurElem},\pTupA\pVertA_\scenter)$ as~$\BlockA_\blurElem$ has in $(\StructA_f, \pTupA\pVertA_\scenter)$
	and $\BlockB_{\blurElem,1} := \restrictVect{\BlockB_\blurElem}{[k-1]}$
	for every $\blurElem \in \blurrer$
	(or $\BlockB_\blurElem = \BlockB_{\blurElem,1} = \emptyset$ if $\BlockA_\blurElem=\emptyset$).
	The step from~$\BlockA_\blurElem$ to~$\BlockB_\blurElem$ is 
	possible because~$S^\blurElem$ blurs the twist between $(\StructA_f, \pTupA\pVertA_\scenter)$
	and $(\StructA_{g-\blurElem},\pTupA\pVertA_\scenter)$
	for every $\blurElem \in \blurrer$ (Claim~\ref{clm:blur-twist-blur-elem}).
	
	We analyze the structures of these orbits.
	Let~$R$ be the star center component of~$\BlockA$ (and so of~$\BlockB$).
	(Note that~$R$ may be split into multiple components for~$\BlockA_1$,~$\BlockA_2$,~$\BlockA_\blurElem$, etc.).
	We apply Lemma~\ref{lem:split-discon-orbits} to split
	$\BlockA = \restrictVect{\BlockA}{R} \times \restrictVect{\BlockA}{D}$,
	$\BlockA_\blurElem = \restrictVect{\BlockA_\blurElem}{R} \times \restrictVect{\BlockA_\blurElem}{D}$, and
	$\BlockB_\blurElem = \restrictVect{\BlockB_\blurElem}{R} \times \restrictVect{\BlockB_\blurElem}{D}$ for every $\blurElem \in \blurrer$,
	where~$D$ is the set of components of~$\BlockA$ apart from~$R$.
	The components in~$D$ have distance greater than~$1$ to~$\scenter$
	because $\scenter \in R$.
	Hence,
	$\restrictVect{\BlockA}{D} = \restrictVect{\BlockA_\blurElem}{D}$ and
	\begin{align*}
		\BlockA = \restrictVect{\BlockA}{R} \times \restrictVect{\BlockA_\blurElem}{D}
	\end{align*}
	for every $\blurElem \in \blurrer$.
	Because $\restrictVect{\BlockA_\blurElem}{R}$
	has the same type in $(\StructA_f, \pTupA\pVertA_\scenter)$
	as $\restrictVect{\BlockB_\blurElem}{R}$ has in 
	$(\StructA_{g-\blurElem}, \pTupA\pVertA_\scenter)$
	and their origins do not contain any of the vertices $\primeSub{\stip}{i}$,
	it even follows that $\restrictVect{\BlockA_\blurElem}{R} = \restrictVect{\BlockB_\blurElem}{R}$
	for every $\blurElem \in \blurrer$
	because $(\StructA_f, \pTupA\pVertA_\scenter)[R]=(\StructA_{g-\blurElem}, \pTupA\pVertA_\scenter)[R]$.
	So
	\begin{align*}
		\BlockB_\blurElem = \restrictVect{\BlockA_\blurElem}{R} \times \restrictVect{\BlockB_\blurElem}{D}
	\end{align*}
	for every $\blurElem \in\blurrer$.
	For readability, we set $\tupA_D := \remZ{\tupA}_{D}$ and $\tupA_R := \remZ{\tupA}_{R}$.
	Then $\remZ{\tupA}= \tupA_R\tupA_D$.
	We perform the same for $\remZ{\tupB} = \tupB_R\tupB_D$
	and $\tupC = \tupC_R\tupC_D$. So we obtain in the next step that
	\begin{align*}
		(\star) &= \sum_{\blurElem \in \blurrer}
		\sum_{\tupC_R\tupC_D \in \BlockB_{\blurElem,1}}
		S^\blurElem (\tupA_R\tupA_D, \tupC_R\tupC_D) \cdot  \charMat{\BlockB_\blurElem}
		(\tupC_R\tupC_D,\inv{\blurElem}(\inv{\fixTwistType_\blurElem}(\tupB_R\tupB_D))).
	\end{align*}
	We set \[\BlockB_D := \fixTwistType_\blurElem(\restrictVect{\BlockB_\blurElem}{D})\] for some $\blurElem \in \blurrer$.
	Claim~\ref{clm:orbits-rec-blurelem} states
	that~$\BlockB_D$ is an orbit of $(\StructA_g, \pTupA\pVertA_\scenter)$
	and has the same type in $(\StructA_g, \pTupA\pVertA_\scenter)$ as 
	$\restrictVect{\BlockB_\blurElem}{D}$ has in $(\StructA_{g-\blurElem}, \pTupA\pVertA_\scenter)$.
	As seen before, this is the same type as
	$\restrictVect{\BlockA_\blurElem}{D}$ has in $(\StructA_f, \pTupA\pVertA_\scenter)$.
	Because, as already seen, $\restrictVect{\BlockA_\blurElem}{D} = \restrictVect{\BlockA}{D}$ for every $\blurElem \in \blurrer$,
	the type of~$\BlockB_D$ is independent of~$\blurElem$
	and~$\BlockB_D$ is well-defined.
	So~$\BlockB_D$ is also an orbit of $(\StructA_g, \pTupA)$
	and $\blurElem(\BlockB_D) = \BlockB_D$ for every $\blurElem \in \blurrer$ because~$\blurElem$ is an orbit-automorphism (Claim~\ref{clm:blurelem-blockiso}).
	Thus, $\BlockB_D = \blurElem(\fixTwistType_\blurElem(\restrictVect{\BlockB_\blurElem}{D}))$ for every $\blurElem \in \blurrer$.
	We set for every $\blurElem \in \blurrer$ 
	\[\BlockB_\blurElem' := \restrictVect{\BlockB_\blurElem}{R} \times \BlockB_D = \restrictVect{\BlockA_\blurElem}{R} \times \BlockB_D.\]
	By Claim~\ref{clm:orbits-rec-blurelem}, for every $\blurElem \in \blurrer$ it holds that
	\begin{align*}
		&\hiddenEq \charMat{\BlockB_\blurElem}
	(\tupC_R\tupC_D,\inv{\blurElem}(\inv{\fixTwistType_\blurElem}(\tupB_R\tupB_D))) \\
		&= \charMat{\BlockB'_\blurElem}(\tupC_R\blurElem(\fixTwistType_\blurElem(\tupC_D)),
	\inv{\blurElem}(\inv{\fixTwistType_\blurElem}(\tupB_R))
	\blurElem(\fixTwistType_\blurElem(\inv{\blurElem}(\inv{\fixTwistType_\blurElem}(\tupB_D)))))\\
		&=\charMat{\BlockB'_\blurElem}(\tupC_R\blurElem(\fixTwistType_\blurElem(\tupC_D)),
		\inv{\blurElem}(\inv{\fixTwistType_\blurElem}(\tupB_R))
		\tupB_D).
	\end{align*}
	We used that~$\blurElem$ and~$\fixTwistType_\blurElem$ commute
	(Claim~\ref{clm:commutativity})
	and can be applied component-wise (Claim~\ref{clm:component-wise-application}).
	With $\BlockB_{\blurElem,1}' := \restrictVect{\BlockB'_\blurElem}{[k-1]}$ for every $\blurElem\in\blurrer$, we  obtain that
	\begin{align*}
	 (\star) &= \sum_{\blurElem \in \blurrer}
		\sum_{\tupC_R\tupC_D \in \BlockB_{\blurElem,1}}
		S^\blurElem (\tupA_R\tupA_D, \tupC_R\tupC_D) \cdot  \charMat{\BlockB'_\blurElem}(\tupC_R\blurElem(\fixTwistType_\blurElem(\tupC_D)),
		\inv{\blurElem}(\inv{\fixTwistType_\blurElem}(\tupB_R))
		\tupB_D)\\
		&= \sum_{\blurElem \in \blurrer}
		\sum_{\tupC_R\tupC_D \in \BlockB_{\blurElem,1}'}
		S^\blurElem (\tupA_R\tupA_D, \tupC_R\inv{\blurElem}(\inv{\fixTwistType_\blurElem}(\tupC_D))) \cdot  \charMat{\BlockB'_\blurElem}(\tupC_R\tupC_D,
		\inv{\blurElem}(\inv{\fixTwistType_\blurElem}(\tupB_R))
		\tupB_D).
	\end{align*}
	We claim that
	\[\BlockB = \restrictVect{\BlockA}{R} \times \BlockB_D.\]
	First, $\restrictVect{\BlockA}{R} \times \BlockB_D \in \orbs{2k}{(\StructA_g, \pTupA)}$
	by Lemma~\ref{lem:split-discon-orbits}.
	Both parts are orbits and their origins are not connected.
	Second, because~$\BlockB_D$ has the same type in $(\StructA_g, \pTupA\pVertA_\scenter)$ as $\restrictVect{\BlockA_\blurElem}{D}=\restrictVect{\BlockA}{D}$
	in $(\StructA_f, \pTupA\pVertA_\scenter)$,~%
	$\BlockB_D$ has also the same type in $(\StructA_g,\pTupA)$
	as $\restrictVect{\BlockA_\blurElem}{D}=\restrictVect{\BlockA}{D}$
	in $(\StructA_f, \pTupA)$.
	Because the components in~$R$ do not contain a twisted edge,
	$\restrictVect{\BlockA}{R} = \restrictVect{\BlockB}{R}$
	and so indeed $\BlockB = \restrictVect{\BlockA}{R} \times \BlockB_D$.
	Because $\restrictVect{\BlockB'_\blurElem}{R} = \restrictVect{\BlockA_\blurElem}{R}$
	and by the definition of $\restrictVect{\BlockA_\blurElem}{R}$,
	it holds that $\tupC_R\tupC_D\inv{\blurElem}(\inv{\fixTwistType_\blurElem}(\tupB_R)) \in \BlockB_\blurElem'$ if and only if
	$\tupC_\scenter\tupC_R\tupC_D,
	\inv{\blurElem}(\inv{\fixTwistType_\blurElem}(\tupB_\scenter\tupB_R))
	\in \BlockB$. We continue as follows:
	\begin{align*}
		(\star) &= \sum_{\blurElem \in \blurrer}
		\sum_{\tupC_R\tupC_D \in \BlockB_{\blurElem,1}'}
		S^\blurElem (\tupA_R\tupA_D, \tupC_R\inv{\blurElem}(\inv{\fixTwistType_\blurElem}(\tupC_D))) \cdot  \charMat{\BlockB}(\tupA_\scenter\tupC_R\tupC_D,
		\inv{\blurElem}(\inv{\fixTwistType_\blurElem}(\tupB_\scenter\tupB_R))
		\tupB_D)\\
		&= \sum_{\blurElem \in \blurrer}
		\sum_{\substack{\tupC_\scenter\tupC_R\tupC_D \in \BlockB_1,\\
			\tupC_\scenter = \tupA_\scenter}}
		S^\blurElem (\tupA_R\tupA_D, \tupC_R\inv{\blurElem}(\inv{\fixTwistType_\blurElem}(\tupC_D))) \cdot  \charMat{\BlockB}(\tupC_\scenter\tupC_R\tupC_D,
		\inv{\blurElem}(\inv{\fixTwistType_\blurElem}(\tupB_\scenter\tupB_R))
		\tupB_D).
	\end{align*}
	For every $\blurElem \in \blurrer$, it holds that
	$\inv{\blurElem}(\inv{\fixTwistType_\blurElem}(\restrictVect{\BlockB}{R})) = \restrictVect{\BlockB}{R}$
	because~$\fixTwistType_\blurElem$ is the identity function on $R$-vertices
	($R$ is a star center component)
	and because~$\blurElem$ is an orbit-automorphism (Claim~\ref{clm:blurelem-blockiso}).
	So by Lemma~\ref{lem:split-discon-orbits},
	$\inv{\blurElem}(\inv{\fixTwistType_\blurElem}(\tupC_R))\tupC_D \in \BlockB$
	if and only if $\tupC_R\tupC_D \in \BlockB$ for every $\blurElem \in \blurrer$.
	We substitute~$\tupC_R\tupC_D \mapsto\inv{\blurElem}(\inv{\fixTwistType_\blurElem}(\tupC_R))\tupC_D$
	(this is a bijection).
	\begin{align*}
		&\hiddenEq(\star) \\
		&= \smash{%
			\sum_{\blurElem \in \blurrer}
		\sum_{\substack{\tupC_\scenter\tupC_R\tupC_D \in \BlockB_1,\\
				\inv{\blurElem}(\inv{\fixTwistType_\blurElem}(\tupC_\scenter)) = \tupA_\scenter}}}
		S^\blurElem (\tupA_R\tupA_D, \inv{\blurElem}(\inv{\fixTwistType_\blurElem}(\tupC_R))\inv{\blurElem}(\inv{\fixTwistType_\blurElem}(\tupC_D))) \cdot{} \phantom{\sum}\\ &\hspace{5cm}\charMat{\BlockB}(\inv{\blurElem}(\inv{\fixTwistType_\blurElem}(\tupC_\scenter\tupC_R))\tupC_D,
		\inv{\blurElem}(\inv{\fixTwistType_\blurElem}(\tupB_\scenter\tupB_R))
		\tupB_D)\\[5pt]
		&= \sum_{\blurElem \in \blurrer}
		\sum_{\substack{\tupC_\scenter\tupC_R\tupC_D \in \BlockB_1,\\
				\inv{\blurElem}(\inv{\fixTwistType_\blurElem}(\tupC_\scenter)) = \tupA_\scenter}}
		S^\blurElem (\tupA_R\tupA_D, \inv{\blurElem}(\inv{\fixTwistType_\blurElem}(\tupC_R\tupC_D))) \cdot  \charMat{\BlockB}(\inv{\blurElem}(\inv{\fixTwistType_\blurElem}(\tupC_\scenter\tupC_R))\tupC_D,
		\inv{\blurElem}(\inv{\fixTwistType_\blurElem}(\tupB_\scenter\tupB_R))
		\tupB_D)\\
		&= \sum_{\blurElem \in \blurrer}
		\sum_{\substack{\tupC_\scenter\tupC_R\tupC_D \in \BlockB_1,\\
				\inv{\blurElem}(\inv{\fixTwistType_\blurElem}(\tupC_\scenter)) = \tupA_\scenter}}
		S^\blurElem (\tupA_R\tupA_D, \inv{\blurElem}(\inv{\fixTwistType_\blurElem}(\tupC_R\tupC_D))) \cdot  \charMat{\BlockB}(\tupC_\scenter\tupC_R\tupC_D,
		\tupB_\scenter\tupB_R\tupB_D).
	\end{align*}
	In the last step
	we again used that~$\fixTwistType_\blurElem$ is the identity on the $R$-component
	and that~$\blurElem$ is an orbit-automorphism.
	We finish the proof with the reverse steps as to how we started it.
	\begin{align*}
		(\star) &= \sum_{\blurElem \in \blurrer}
		\sum_{\substack{\tupC \in \BlockB_1,\\
				\inv{\blurElem}(\inv{\fixTwistType_\blurElem}(\tupC_\scenter)) = \tupA_\scenter}}
		S^\blurElem (\remZ{\tupA}, \inv{\blurElem}(\inv{\fixTwistType_\blurElem}(\remZ{\tupC}))) \cdot  \charMat{\BlockB}(\tupC,
		\tupB)\\
		&= 	\sum_{\tupC \in \BlockB_1}
		\sum_{\substack{\blurElem \in \blurrer,\\
				\blurElem(\tupA_\scenter) = \inv{\fixTwistType_\blurElem}(\tupC_\scenter)}}
		S^\blurElem (\remZ{\tupA}, \inv{\blurElem}(\inv{\fixTwistType_\blurElem}(\remZ{\tupC}))) \cdot  \charMat{\BlockB}(\tupC,
		\tupB)\\
		&= 	\sum_{\tupC \in \BlockB_1}
		S(\tupA,\tupC) \cdot \charMat{\BlockB}(\tupC,\tupB)\\
		&=(S \cdot \charMat{\BlockB}) (\tupA, \tupB).\qedhere
	\end{align*}
\end{claimproof}
Claims~\ref{clm:blurrable-case},~\ref{clm:blurrable-non-blurrable-case}, and~\ref{clm:non-blurrable-case} show that~$S$ indeed $k$-blurs the twist
between $(\StructA_f, \pTupA)$ and $(\StructA_g, \pTupA)$.
This finishes the proof of Lemma~\ref{lem:blur-kary}.
\begin{remark}
	We checked our construction in the proof
	for $k \leq 2$ on the computer.
	For larger~$k$ it was computationally not tractable.
\end{remark}

\section{Separating Rank Logic from Choiceless Polynomial Time}
\label{sec:separating}
In this section we finally separate rank logic from \CPT{}.
To apply Lemma~\ref{lem:blur-kary},
we need to construct a suitable class of base graphs.
\begin{lemma}
	\label{lem:regular-girth-connected-unbounded}
	For every $n \in \nat$, there is a regular graph
	that has degree at least~$n$, girth at least~$n$,
	and is $n$-connected.
\end{lemma}
\begin{proof}
	A $(d,g)$-cage is a $d$-regular graph with girth $g$ of minimum order.
	For every odd  $g \geq 7$,
	every $(d,g)$-cage is
	$\lceil \frac{d}{2} \rceil$-connected~\cite{BalbuenaS12}.
	So it suffices to show that for every $n$
	there is a $d \geq 2n$ and an odd $g \geq n$
	such that there is a $(d,g)$-cage.
	For every $d\geq 2$
	and $g\geq 3$ there is a $d$-regular graph of girth $g$~\cite{Sachs63}
	and so in particular a $(d,g)$-cage.
\end{proof}

\begin{lemma}
	\label{lem:girth-distant-vertex}
	Let $G=(V,E)$ be a $d$-regular graph of girth at least $2(\ell +2)+1$ for some $\ell \in \nat$.
	Then for every set $V' \subseteq V$ of size $|V'| < d$,
	there is a vertex $\bvertA \in V$ such that $\distance{G}{V'}{\bvertA} > \ell$.
\end{lemma}
\begin{proof}
	For every $\bvertB \in V$ and $i \leq \ell +2$
	it holds that $|\neighborsK{G}{i}{\bvertB}| = 1 + d\sum_{j = 0}^{i-1} (d-1)^j$
	because~$G$ is $d$-regular and has girth at least $2(\ell +2) +1$
	and thus $G[\neighborsK{G}{i}{\bvertB}]$ is a tree,
	in which the root~$y$ has~$d$ many subtrees, which are all complete $(d-1)$-ary trees of height $i-1$.
	Let $\bvertB \in V'$. Then
	\begin{align*}
		\big|\neighborsK{G}{\ell +2}{\bvertB} \setminus \bigcup_{\bvertC \in V' } \neighborsK{G}{\ell}{\bvertC}\big| &\geq
		    1+d\sum_{j = 0}^{\ell +1} (d-1)^j - |V'|\cdot \left(1+ d \sum_{j = 0}^{\ell-1} (d-1) ^j\right)\\
		&\geq d\sum_{j = 0}^{\ell +1} (d-1)^j - (d-1) - d\sum_{j = 0}^{\ell} (d-1)^j \\
		&\geq d(d-1)^{\ell+1}-(d-1) > 0.
	\end{align*}
	Hence, there is at least one $\bvertA \in V$ satisfying the claim.
\end{proof}

\begin{theorem}
	\label{thm:graph-class-inv-matrix-duplicator-wins}
	There is a class of base graphs $\GraphClass$,
	such that for every $k,m \in \nat $,
	there is a graph $G =(V,E,\leq) \in \GraphClass$ and a $\pot \in \nat$
	such that $\CFIgraph{2^\pot}{G}{f} \invertibleMapEquiv{2k+m}{k}{\set{2}} \CFIgraph{2^\pot}{G}{g}$ for every $f,g \colon E \to \ZZ_{2^\pot}$ satisfying $\sum g = \sum f + 2^{\pot-1}$.
\end{theorem}
\begin{proof}
	Let~$\GraphClass$ be the class of graphs
	given by Lemma~\ref{lem:regular-girth-connected-unbounded} for every $n\in \nat$ equipped with some total order.
	Let $k,m\in \nat$ and $G \in \GraphClass$ such that
	it has degree $d \geq d(k,m) > m$,
	girth at least $2(2\minRad(k+1)+2)+1$,
	and~$G$ is at least $(2k+m+1)$-connected.
	Let $\pot = \pot(k)$,
	$e = \set{\bvertA,\bvertB} \in E$, and
	$g,f \colon E \to \ZZ_{2^\pot}$ such that
	$\sum g = \sum f + 2^{\pot-1}$.
	Up to isomorphism of the CFI structures,
	we can assume that~$e$ is the only edge twisted by~$f$ and~$g$
	and that $g(e) = f(e) + 2^{\pot-1}$.
	Let $\StructA_f = \CFIgraph{2^\pot}{G}{f}$
	and $\StructA_g = \CFIgraph{2^\pot}{G}{g}$
	both with universe~$\StructVA$.
	We show that Duplicator has a winning strategy in the invertible-map game $\invertMapGame{2k+m}{k}{\set{2}}$ played on $\StructA_f$ and $\StructA_g$.
	
	We consider the case where~$m$ many pebbles are placed on the structures.
	Starting with fewer pebbles makes the proof more elaborate, without providing any additional insights.
	Let $\tupA \in \StructVA^{2k+m}$ and $\tupB \in \StructVA^{2k+m}$
	such that the type of~$\tupA$  in~$\StructA_f$ is the same as of~$\tupB$ in~$\StructA_g$
	and $e \not\subseteq \orig{\tupA}$,
	i.e., there exists a local isomorphism mapping~$\tupA$ to~$\tupB$.
	Assume we play on $(\StructA_f, \tupA)$ and $(\StructA_g, \tupB)$.
	Let ${\BlockA \in \orbs{2k+m}{(\StructA_f, \tupA)}}$ contain~$\tupA$
	and ${\BlockB \in \orbs{2k+m}{(\StructA_g, \tupB)}}$ contain~$\tupB$. 	
	Because~$\tupA$ and~$\tupB$ have the same type,~%
	$\BlockA$ and~$\BlockB$ have the same type.
	Because $e \not\subseteq \orig{\tupA}$, 
	we have that $\StructA_f[\orig{\BlockA}] = \StructA_g[\orig{\BlockB}]$
	and thus $\BlockA = \BlockB$.
	That is, there is an automorphism $\autoA \in \autgrp{\StructA_g}$
	such that $\autoA(\tupB) = \tupA$ (Corollary~\ref{cor:orbit-types-unique}).
	Up to isomorphism, we can continue the game on
	$(\StructA_f, \tupA)$ and $\autoA((\StructA_g, \tupB)) = (\StructA_g, \tupA)$.
	Spoiler picks up~$2k$ pebbles on each graph
	leaving us with the structures 
	$(\StructA_f, \tupC)$ and $(\StructA_g, \tupC)$
	for some $\tupC \in \StructVA^{m}$,
	where~$\tupC$ has the same type in~$\StructA_f$ and in~$\StructA_g$.
	
	There is a vertex $\bvertC \in V$
	such that $\distance{G}{\orig{\tupC}}{\bvertC} > 2\minRad(k+1)$
	by Lemma~\ref{lem:girth-distant-vertex}.
	We can assume up to exchanging~$\bvertA$ and~$\bvertB$ that $e \cap \orig{\tupC} \subseteq \set{\bvertA}$
	because $e \not\subseteq \orig{\tupC}$.
	Because~$G$ is $(2k+m+1)$-connected,
	there is a path $\spath = (\bvertA,\bvertB, \dots, \bvertC', \bvertC)$
	such that~$\spath$ and $\orig{\tupC}$ are disjoint apart from~$\bvertA$.
	Then we apply the path isomorphism $\autoB := \pathiso{2^{\pot-1}}{\spath}$
	and w.l.o.g.~can continue the game on
	$(\StructA_f, \tupC)$ and $\autoB((\StructA_g, \tupC)) = (\StructA_{h}, \tupC)$,
	where~$f$ and~$h$ only twist the edge $\set{\bvertC, \bvertC'}$.
	Now between~$\StructA_f$ and~$\StructA_{h}$ the twist is moved sufficiently
	far away from $\orig{\tupC}$. 
	
	Duplicator chooses the partitions
	$\PartA := \orbs{2k}{(\StructA_f, \tupC)}$ and
	$\PartB := \orbs{2k}{(\StructA_h, \tupC)}$
	of $\StructVA^k \times \StructVA^k$
	and the invertible matrix~$S$ that $k$-blurs the twist between
	$(\StructA_f, \tupC)$ and $(\StructA_h, \tupC)$
	given by Lemma~\ref{lem:blur-kary}.
	By construction, the conditions of the lemma are satisfied.
	The matrix~$S$ induces a map $\PartA \to \PartB$
	mapping $\BlockA \mapsto \BlockB$ if and only if $\BlockA$ and $\BlockB$ have the same type in $(\StructA_f, \tupC)$ respectively $(\StructA_h, \tupC)$
	(Definition~\ref{def:blur-twist}).
	Spoiler chooses orbits $\BlockA\in \PartA$
	and $\BlockB\in \PartB$ of the same type
	and $\tupA' \in \BlockA$ and $\tupB' \in \BlockB$.
	Then~$\tupA'$ has the same type in $(\StructA,\tupC)$
	as~$\tupB'$ in $(\StructA_h,\tupC)$.
	So~$\tupC\tupA'$ and~$\tupC\tupB'$ induce a local isomorphism
	and the next round starts.

	As we can see, we are in the same situation as before,
	that is~$\tupC\tupA'$ and~$\tupC\tupB'$ have the same type
	in~$\StructA_f$ and~$\StructA_h$ respectively
	and we can apply Duplicator's strategy again. 
	So Duplicator has winning strategy in the
	$\invertMapGame{2k+m}{k}{\set{2}}$-game.
\end{proof}

\begin{theorem}
	There is a class of $\sig$-structures~$\GraphClass$,
	such that $\IFPR < \PTime$ on~$\GraphClass$
	and $\CPT = \PTime$ on $\GraphClass$.
\end{theorem}
\begin{proof}
	Let~$\GraphClass'$ be the graph class from Theorem~\ref{thm:graph-class-inv-matrix-duplicator-wins},
	and set $\GraphClass := \CFITwoClass{\GraphClass'}$ (recall Definition~\ref{def:cfi-graph-class}).
	We want to show that
	the CFI query for~$\GraphClass$
	is not \IFPR{}-definable.
	This is the task of deciding whether for a given
	$\CFIgraph{2^\pot}{G}{f} \in \GraphClass$
	it holds that $\sum f = 0$.
	By Lemma~\ref{lem:IFPR-to-IFPR2}, it suffices to show
	that the CFI query is not \IFPRP{\set{2}}-definable.
	The claim follows from Lemma~\ref{lem:invertible-map-refines-rank-logic}
	and Theorem~\ref{thm:graph-class-inv-matrix-duplicator-wins}.
	
	We now argue that \CPT{} captures \PTime{} on $\GraphClass$.
	By Theorem~\ref{thm:CPT-ordered-abelian-colors},
	it suffices to show that $\GraphClass$ is a class of
	structures with abelian and ordered colors (cf.~Section~\ref{sec:rank-logc}).
	By Lemma~\ref{lem:cfi-autgroup-abelian-two-group},
	the colors are abelian
	and it remains to define for every color class 
	an ordered and  transitive permutation group in \CPT.
	Every gadget forms a color class.
	As seen in Section~\ref{sec:cfi-structures},
	every gadget forms a $1$-orbit and has a regular, so in particular transitive,
	automorphism group (Lemmas~\ref{lem:cfi-same-type-automorphism} and~\ref{lem:cfi-orbit-auto-group-regular}).
	Consider a degree~$d$ gadget.
	Then every automorphism~$\autoA$ of the gadget can be identified with a tuple $\tup{a} \in \ZZ_{2^\pot}^{d}$ satisfying $\sum \tup{a} = 0$
	such that $\autoA(\vertA) = \vertB$ if and only if $\tup{a}(\vertA) = \vertB$.
	Indeed, every such tuple represents an automorphism.
	So the automorphism group of the gadget can be represented
	by $\setcond{\tup{a} \in \ZZ_{2^\pot}^{d}}{\sum \tup{a} = 0}$,
	which is a \CPT{}-definable set (the tuples, although of unbounded length,
	can be represented by standard set encoding of tuples).
	This set can be ordered using the lexicographical order on tuples.
	Because the base graph is ordered,
	the automorphism corresponding to the tuple~$\tup{a}$
	is \CPT{}-definable: The $i$-th entry defines the action on the $i$-th outgoing edge, which is definable using the relation~$\rel_I$.
\end{proof}

\section{Discussion}
\label{sec:discussion}

We showed that rank logic does not capture \CPT{} and in particular not \PTime{}
on the class of CFI structures over rings~$\ZZ_{2^i}$,
even if the base graph is totally ordered.
To do so, we used combinatorial objects called blurrers
and a recursive approach over the arity.
The non-locality of $k$\nobreakdash-tuples for $k>1$
increased the difficulty of $k$-ary rank operators dramatically
compared to the arity~$1$ case.
It was suggested in~\cite{GradelPakusa19}
that CFI graphs over~$\ZZ_4$ could be a separating example
for rank logic and \PTime{}.
Our concepts for blurrers required rings~$\ZZ_{2^i}$ with $i>2$
for higher arities.
Actually, our computer experiments to check Lemma~\ref{lem:blur-kary}
indicate that CFI graphs over~$\ZZ_4$ could possibly
be distinguished using the $k$-ary invertible map game for $k>1$.
It might also be possible that the CFI query over~$\ZZ_4$ is
definable in rank logic using rank operators of higher arities,
but this remains an open question.

There are various definitions of rank logic,
which slightly differ in the way the matrices in the rank operator are defined.
In particular, there is an extension,
in which rank operators not only bind universe variables,
but also numeric variables~\cite{GradelPakusa19,Laubner2011,pakusa2015}.
It is not clear whether this extension
is more expressive or not.
However, for a suitable adaptation of the invertible-map game,
which also supports numeric variables to construct matrices,
we strongly believe that our arguments work exactly the same.
In fact, we think that at least in the invertible-map game
numeric variables do not increase the expressiveness
and thus our arguments directly apply.

A natural question is how rank logic can be extended
such that it can define the CFI query.
We have shown that it is not sufficient to compute ranks over finite fields only.
Even more, our construction applies to arbitrary linear-algebraic operators over finite fields~\cite{DawarGraedelLichter22}. 
However, it is not clear how rank logic can be extended to rings~$\ZZ_i$.
Over rings, there are several non-equivalent notions of the rank of a matrix.
For a discussion see~\cite{DawarGHKP2013, pakusa2015}.
As opposed to rank logic,
solvability logic can easily be extended to rings
and thus should be able to define the CFI query over all~$\ZZ_i$.
Notably, such an extension would also capture \PTime{} on structures with bounded and abelian colors~\cite{pakusa2015}.

\subsection*{Acknowledgments}
I thank Jendrik Brachter for his help with Lemma~\ref{lem:orbit-diagonal+orbit-invariant+odd-filled-implies-invertible}.
I also thank the anonymous reviewers for their detailed and helpful comments.

\bibliographystyle{plain}
\bibliography{rank_logic_ptime_full}

\begin{thebibliography}{10}

\bibitem{AtseriasBD09}
Albert Atserias, Andrei~A. Bulatov, and Anuj Dawar.
\newblock Affine systems of equations and counting infinitary logic.
\newblock {\em Theor. Comput. Sci.}, 410(18):1666--1683, 2009.

\bibitem{BalbuenaS12}
Camino Balbuena and Juli{\'{a}}n Salas.
\newblock A new bound for the connectivity of cages.
\newblock {\em Appl. Math. Lett.}, 25(11):1676--1680, 2012.

\bibitem{BerkholzG17}
Christoph Berkholz and Martin Grohe.
\newblock Linear diophantine equations, group csps, and graph isomorphism.
\newblock In Philip~N. Klein, editor, {\em Proceedings of the Twenty-Eighth
  Annual {ACM-SIAM} Symposium on Discrete Algorithms, {SODA} 2017, Barcelona,
  Spain, Hotel Porta Fira, January 16-19}, pages 327--339. {SIAM}, 2017.

\bibitem{Berman67}
Samuil~D. Berman.
\newblock On the theory of group codes.
\newblock {\em Kibernetika}, 3(1):31--39, 1967.

\bibitem{bgs1999}
Andreas Blass, Yuri Gurevich, and Saharon Shelah.
\newblock Choiceless polynomial time.
\newblock {\em Ann. Pure Appl. Logic}, 100(1-3):141--187, 1999.

\bibitem{CaiFI1992}
Jin{-}yi Cai, Martin F{\"{u}}rer, and Neil Immerman.
\newblock An optimal lower bound on the number of variables for graph
  identification.
\newblock {\em Combinatorica}, 12(4):389--410, 1992.

\bibitem{Cameron1994}
Peter~J. Cameron.
\newblock {\em Combinatorics: Topics, Techniques, Algorithms}.
\newblock Cambridge University Press, 1994.

\bibitem{ChandraHarel82}
Ashok~K. Chandra and David Harel.
\newblock Structure and complexity of relational queries.
\newblock {\em J. Comput. Syst. Sci.}, 25(1):99--128, 1982.

\bibitem{DawarGHKP2013}
Anuj Dawar, Erich Gr{\"{a}}del, Bjarki Holm, Eryk Kopczynski, and Wied Pakusa.
\newblock Definability of linear equation systems over groups and rings.
\newblock {\em Log. Methods Comput. Sci.}, 9(4), 2013.

\bibitem{DawarGraPak19}
Anuj Dawar, Erich Gr{\"{a}}del, and Wied Pakusa.
\newblock Approximations of isomorphism and logics with linear-algebraic
  operators.
\newblock In {\em 46th International Colloquium on Automata, Languages, and
  Programming, {ICALP} 2019.}, pages 112:1--112:14, 2019.

\bibitem{DawarGHL09}
Anuj Dawar, Martin Grohe, Bjarki Holm, and Bastian Laubner.
\newblock Logics with rank operators.
\newblock In {\em Proceedings of the 24th Annual {IEEE} Symposium on Logic in
  Computer Science, {LICS} 2009, 11-14 August 2009, Los Angeles, CA, {USA}},
  pages 113--122. {IEEE} Computer Society, 2009.

\bibitem{DawarGraedelLichter22}
Anuj Dawar, Erich Grädel, and Moritz Lichter.
\newblock Limitations of the invertible-map equivalences.
\newblock {\em Journal of Logic and Computation}, 09 2022.

\bibitem{DawarHolm12}
Anuj Dawar and Bjarki Holm.
\newblock Pebble games with algebraic rules.
\newblock In Artur Czumaj, Kurt Mehlhorn, Andrew~M. Pitts, and Roger
  Wattenhofer, editors, {\em Automata, Languages, and Programming - 39th
  International Colloquium, {ICALP} 2012, Warwick, UK, July 9-13, 2012,
  Proceedings, Part {II}}, volume 7392 of {\em Lecture Notes in Computer
  Science}, pages 251--262. Springer, 2012.

\bibitem{DawarRicherbyRossman2008}
Anuj Dawar, David Richerby, and Benjamin Rossman.
\newblock Choiceless polynomial time, counting and the
  {Cai-F{\"{u}}rer-Immerman} graphs.
\newblock {\em Ann. Pure Appl. Logic}, 152(1-3):31--50, 2008.

\bibitem{Furer2001}
Martin F{\"{u}}rer.
\newblock {W}eisfeiler-{L}ehman refinement requires at least a linear number of
  iterations.
\newblock In {\em Automata, Languages and Programming, 28th International
  Colloquium, {ICALP} 2001, Crete, Greece, July 8-12, 2001, Proceedings},
  volume 2076 of {\em Lecture Notes in Computer Science}, pages 322--333.
  Springer, 2001.

\bibitem{GraedelGrohe15}
Erich Gr{\"{a}}del and Martin Grohe.
\newblock Is polynomial time choiceless?
\newblock In Lev~D. Beklemishev, Andreas Blass, Nachum Dershowitz, Bernd
  Finkbeiner, and Wolfram Schulte, editors, {\em Fields of Logic and
  Computation {II} - Essays Dedicated to Yuri Gurevich on the Occasion of His
  75th Birthday}, volume 9300 of {\em Lecture Notes in Computer Science}, pages
  193--209. Springer, 2015.

\bibitem{GradelPakusa19}
Erich Gr{\"{a}}del and Wied Pakusa.
\newblock Rank logic is dead, long live rank logic!
\newblock {\em J. Symb. Log.}, 84(1):54--87, 2019.

\bibitem{Grohe2008}
Martin Grohe.
\newblock The quest for a logic capturing {PTIME}.
\newblock In {\em Proceedings of the Twenty-Third Annual {IEEE} Symposium on
  Logic in Computer Science, {LICS} 2008, 24-27 June 2008, Pittsburgh, PA,
  {USA}}, pages 267--271. {IEEE} Computer Society, 2008.

\bibitem{Grohe2010}
Martin Grohe.
\newblock Fixed-point definability and polynomial time on graphs with excluded
  minors.
\newblock In {\em Proceedings of the 25th Annual {IEEE} Symposium on Logic in
  Computer Science, {LICS} 2010, 11-14 July 2010, Edinburgh, United Kingdom},
  pages 179--188. {IEEE} Computer Society, 2010.

\bibitem{GroheN19}
Martin Grohe and Daniel Neuen.
\newblock Canonisation and definability for graphs of bounded rank width.
\newblock In {\em 34th Annual {ACM/IEEE} Symposium on Logic in Computer
  Science, {LICS} 2019, Vancouver, BC, Canada, June 24-27, 2019}, pages 1--13.
  {IEEE}, 2019.

\bibitem{GurevichS96}
Yuri Gurevich and Saharon Shelah.
\newblock On finite rigid structures.
\newblock {\em J. Symb. Log.}, 61(2):549--562, 1996.

\bibitem{Hella96}
Lauri Hella.
\newblock Logical hierarchies in {PTIME}.
\newblock {\em Inf. Comput.}, 129(1):1--19, 1996.

\bibitem{Holm11}
Bjarki Holm.
\newblock {\em Descriptive complexity of linear algebra}.
\newblock PhD thesis, University of Cambridge, {UK}, 2011.

\bibitem{immerman87}
Neil Immerman.
\newblock Expressibility as a complexity measure: results and directions.
\newblock In {\em Proceedings of the Second Annual Conference on Structure in
  Complexity Theory, Cornell University, Ithaca, New York, USA, June 16-19,
  1987}. {IEEE} Computer Society, 1987.

\bibitem{Laubner2011}
Bastian Laubner.
\newblock {\em The structure of graphs and new logics for the characterization
  of Polynomial Time}.
\newblock PhD thesis, Humboldt University of Berlin, 2011.

\bibitem{Lichter21}
Moritz Lichter.
\newblock Separating rank logic from polynomial time.
\newblock In {\em 36th Annual {ACM/IEEE} Symposium on Logic in Computer
  Science, {LICS} 2021, Rome, Italy, June 29 - July 2, 2021}, pages 1--13.
  {IEEE}, 2021.

\bibitem{LichterS21}
Moritz Lichter and Pascal Schweitzer.
\newblock Canonization for bounded and dihedral color classes in choiceless
  polynomial time.
\newblock In Christel Baier and Jean Goubault{-}Larrecq, editors, {\em 29th
  {EACSL} Annual Conference on Computer Science Logic, {CSL} 2021, January
  25-28, 2021, Ljubljana, Slovenia (Virtual Conference)}, volume 183 of {\em
  LIPIcs}, pages 31:1--31:18. Schloss Dagstuhl - Leibniz-Zentrum f{\"{u}}r
  Informatik, 2021.

\bibitem{NeuenSchweitzer17}
Daniel Neuen and Pascal Schweitzer.
\newblock Benchmark graphs for practical graph isomorphism.
\newblock In Kirk Pruhs and Christian Sohler, editors, {\em 25th Annual
  European Symposium on Algorithms, {ESA} 2017, September 4-6, 2017, Vienna,
  Austria}, volume~87 of {\em LIPIcs}, pages 60:1--60:14. Schloss Dagstuhl -
  Leibniz-Zentrum f{\"{u}}r Informatik, 2017.

\bibitem{Otto1997}
Martin Otto.
\newblock {\em Bounded variable logics and counting - a study in finite
  models}, volume~9 of {\em Lecture Notes in Logic}.
\newblock Springer, 1997.

\bibitem{pakusa2015}
Wied Pakusa.
\newblock {\em Linear Equation Systems and the Search for a Logical
  Characterisation of Polynomial Time}.
\newblock PhD thesis, {RWTH} Aachen University, 2016.

\bibitem{PakusaSchalthoeferSelman2016}
Wied Pakusa, Svenja Schalth{\"{o}}fer, and Erkal Selman.
\newblock Definability of {Cai-F{\"{u}}rer-Immerman} problems in choiceless
  polynomial time.
\newblock In Jean{-}Marc Talbot and Laurent Regnier, editors, {\em 25th {EACSL}
  Annual Conference on Computer Science Logic, {CSL} 2016, August 29 -
  September 1, 2016, Marseille, France}, volume~62 of {\em LIPIcs}, pages
  19:1--19:17. Schloss Dagstuhl - Leibniz-Zentrum fuer Informatik, 2016.

\bibitem{Sachs63}
Horst Sachs.
\newblock Regular graphs with given girth and restricted circuits.
\newblock {\em J. London Math. Soc.}, s1-38(1):423--429, 1963.

\bibitem{AbuZaidGraedelGrohePakusa2014}
Faried~Abu Zaid, Erich Gr{\"{a}}del, Martin Grohe, and Wied Pakusa.
\newblock Choiceless polynomial time on structures with small abelian colour
  classes.
\newblock In Erzs{\'{e}}bet Csuhaj{-}Varj{\'{u}}, Martin Dietzfelbinger, and
  Zolt{\'{a}}n {\'{E}}sik, editors, {\em Mathematical Foundations of Computer
  Science 2014 - 39th International Symposium, {MFCS} 2014, Budapest, Hungary,
  August 25-29, 2014. Proceedings, Part {I}}, volume 8634 of {\em Lecture Notes
  in Computer Science}, pages 50--62. Springer, 2014.

\end{thebibliography}

\end{document}